\documentclass[10pt, conference, letterpaper]{IEEEtran}
\IEEEoverridecommandlockouts
\usepackage{cite}
\usepackage{amsmath,amssymb,amsfonts}
\usepackage{algorithmic}
\usepackage{graphicx}
\usepackage{textcomp}
\usepackage{xcolor}
\def\BibTeX{{\rm B\kern-.05em{\sc i\kern-.025em b}\kern-.08em
    T\kern-.1667em\lower.7ex\hbox{E}\kern-.125emX}}

\usepackage{booktabs} %added by Ish for nice looking tables
\usepackage{pifont}% http://ctan.org/pkg/pifont for tick and cross symbols
\usepackage{xcolor}
\usepackage{amsfonts}
\usepackage{balance}
\PassOptionsToPackage{hyphens}{url}
\usepackage{hyperref}
\usepackage{color}
\usepackage{graphics}
\usepackage{graphicx}
\usepackage{url}
\usepackage{listings}
\usepackage{multicol}
\usepackage{multirow}
\usepackage[scaled]{helvet}
\usepackage{rotating}
\usepackage{xspace}
\usepackage{algorithm}
\usepackage{comment}
\usepackage{enumitem}
\usepackage{amsmath}
\usepackage{mathrsfs}
\usepackage{cancel}
\usepackage{cleveref}
\usepackage{subfigure}
\usepackage{commath}
\usepackage[font=small,labelfont=bf]{caption}
\usepackage{amsthm}
\theoremstyle{plain}

\newcommand{\newtodo}[1]{\ClassWarning{NOT READY TO SUBMIT}{There is something left todo} \textcolor{blue}{}}

\newcommand{\review}[1]{\ClassWarning{NOT READY TO SUBMIT}{There is something left todo} \textcolor{orange}{}}%[Review: #1 ]

\newcommand{\name}{mmFlexible\xspace}
\newcommand{\algoname}{FSDA\xspace}

\newcommand{\dpa}{DPA\xspace}
\newcommand{\ttd}{TTD\xspace}

\definecolor{deepblue}{rgb}{0,0,0.5}
\definecolor{deepred}{rgb}{0.6,0,0}
\definecolor{deepgreen}{rgb}{0,0.5,0}
\definecolor{backcolour}{rgb}{0.95,0.95,0.92}

\def\beq{\begin{equation}}
\def\eeq{\end{equation}}
\def\beqa{\begin{eqnarray}}
\def\eeqa{\end{eqnarray}}
\def\beqan{\begin{eqnarray*}}
\def\eeqan{\end{eqnarray*}}

\newtheorem{theorem}{$\blacksquare$ Theorem}
\newtheorem{corollary}{$\blacksquare$ Corollary}

\def\tm1{t\! - \! 1}
\def\tp1{t\! + \! 1}

\newcommand{\Gbf}{\mathbf{G}}

\newcommand{\Ibf}{\mathbf{I}}

\newcommand{\Ubf}{\mathbf{U}}
\newcommand{\Vbf}{\mathbf{V}}
\newcommand{\Wbf}{\mathbf{W}}

\def\degree{^{\text{o}}}
\def\order{{\mathcal O}}

\def\ant{\text{ant}}
\def\round{\text{round}}
\def\des{\text{desired}}
\begin{document}

% \IEEEaftertitletext{\vspace{-1\baselineskip}}

\title{\name: Flexible Directional Frequency Multiplexing for Multi-user mmWave Networks
}

\author{\IEEEauthorblockN{Ish Kumar Jain, Rohith Reddy Vennam, Raghav Subbaraman, Dinesh Bharadia}
\IEEEauthorblockA{University of California San Diego, La Jolla, CA}
\IEEEauthorblockA{Email: \{ikjain, rvennam, rsubbaraman, dineshb\}@ucsd.edu}
}

% \author{Ish Kumar Jain, Rohith Reddy Vennam, Raghav Subbaraman, Dinesh Bharadia
% % \affiliation{University of California San Diego, La Jolla, CA}\\
% % \email{\{ikjain, rvennam, rsubbaraman, dineshb\}@ucsd.edu}
% }
% \affiliation{UXC}
% \thanks{ss}

\maketitle
% !TEX root = main.tex

\begin{abstract}

% Hence, in addition to traditional time-slot and spectrum resources, mmWave communications add the spatial resource as well, since the mmWave link can not be served to different directions simultaneously.
% \todo{need to fix this}

    % What problem is this paper solving?
    % What have other people done for this problem?
    % What are the limitations of other people’s work?
    % What are the key insights of this work to address these limitations?
    % How to implement systems and evaluate results?
    % What is the impact of the solution?

Modern mmWave systems have limited scalability due to inflexibility in performing frequency multiplexing. All the frequency components in the signal are beamformed to one direction via pencil beams and cannot be streamed to other user directions. We present a new flexible mmWave system called \name that enables flexible directional frequency multiplexing, where different frequency components of the mmWave signal are beamformed in multiple arbitrary directions with the same pencil beam. Our system makes two key contributions: (1) We propose a novel mmWave front-end architecture called a delay-phased array that uses a variable delay and variable phase element to create the desired frequency-direction response. (2) We propose a novel algorithm called FSDA (Frequency-space to delay-antenna) to estimate delay and phase values for the real-time operation of the delay-phased array. Through evaluations with mmWave channel traces, we show that \name provides a 60-150\% reduction in worst-case latency compared to baselines\footnote{This is an extended version of Infocom'23 paper with additional Appendix \ref{app:proof} that provide detailed mathematical analysis and closed-form expressions for delays and phases in DPA. Open-source link \url{https://wcsng.ucsd.edu/dpa}.}.

\end{abstract}


% !TEX root = main.tex
\section{Introduction}\label{sec:intro}

Millimeter-wave (mmWave) networks have the potential to provide wireless connectivity to a growing number of users with their vast bandwidth resources. However, current mmWave systems have a significant limitation in that they are unable to simultaneously serve multiple users by distributing small chunks of frequency resources to different users who are in different directions. Unlike sub-6 systems, that use Omni-antennas to radiate signal in all directions, mmWave systems use pencil beams that illuminate a small region in space, meaning that all the frequency components are directed towards a fixed direction and cannot be distributed to other directions. This inflexibility leads to two main issues, as illustrated in Figure \ref{fig:TDMA_OFDMA_3D}(a). Firstly, it leads to high latency, as the base station (gNB) must serve different user directions in a time-division manner, causing some users to experience long wait times, which is detrimental to latency-sensitive applications. Secondly, it leads to low effective spectrum usage. When a gNB serves one device at a time, each device gets a lot of instantaneous capacity which it may fail to utilize due to limited demand.
% \textcolor{red}{When demand in a particular direction is low, a small frequency band would be sufficient to meet the demand,} 
But because the gNB cannot direct the remaining frequency resources in other directions, those resources are wasted, leading to low effective spectrum usage. Furthermore, other users in other directions could have used these wasted bands to improve overall spectrum usage.

In this paper, we ask the question ``\textit{whether a mmWave base station can transmit or receive to any set of arbitrary directions using any set of contiguous frequency bands, creating a flexible frequency-direction beamforming response}". This is possible using massive antennas and digital beamforming, but traditional mmWave systems rely on analog phased arrays with a single RF chain for cost and power efficiency. However, analog phased arrays cannot create such a frequency-direction response, as they take a single input from the RF chain and radiate all frequencies in the signal in one fixed direction. One naive solution is to split the phased array into multiple sub-arrays and program them to radiate in different directions, but this reduces the directivity in each direction, reducing signal strength, range, reliability, and even data rate.

% Our \dpa differs from the traditional \ttd array, which uses a fixed set of delays at each antenna to create a fixed prism-like pattern, i.e., distributing a fixed sub-set of frequencies in each direction without any flexibility and control~\cite{li2022rainbow}. The problem with this approach is that some frequency subcarriers are wasted in directions without active users, and the active user receives a tiny fraction of bandwidth in one timeslot (TTI), which may not meet the user demand in a reasonable time frame. 

\begin{figure}[t]
    \centering
    \includegraphics[width=0.43\textwidth]{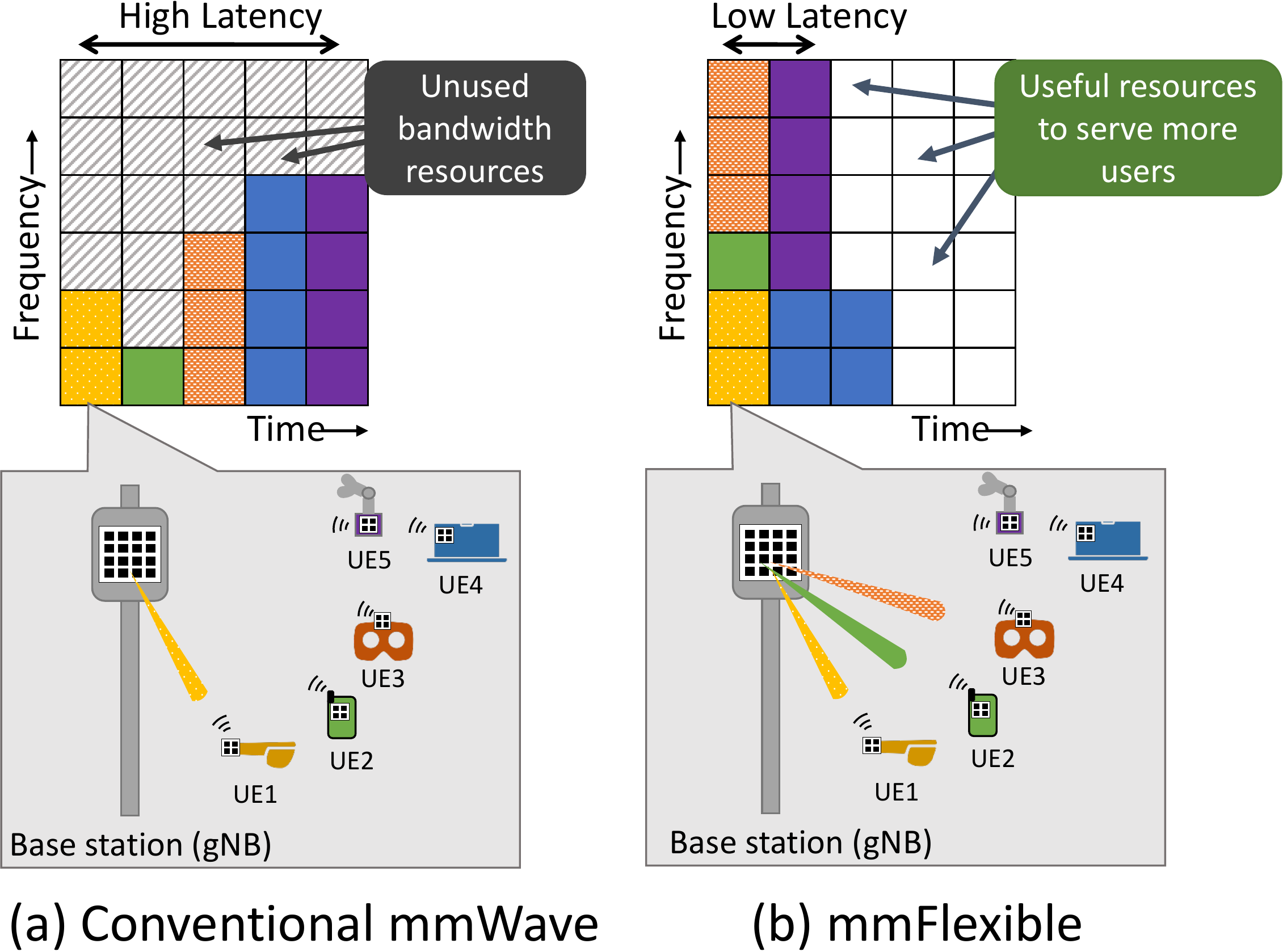}
    \caption{\name enables efficient use of the available mmWave spectrum resources through flexible directional-frequency multiplexing, allowing multiple users to be served simultaneously with low latency and high spectrum utilization.}
    % \name provides flexible directional-frequency multiplexing to multiple concurrent users providing low latency and higher effective spectrum usage. }
    % \name enables frequency multiplexing to multi-user mmWave networks with high flexibility and resource-efficiency.
    % \caption{\name improves scalability and resource-efficiency of mmWave networks while maintaining high throughput and enables OFDMA with frequency-selective multi-beamforming technique to serve multiple users simultaneously.}
    \label{fig:TDMA_OFDMA_3D}
\end{figure}
Recently, new mmWave front-end architectures such as true-time delay (TTD) array~\cite{li2022rainbow,yan2019wideband,boljanovic2021fast} and leaky-wave antenna~\cite{ghasempour2020single} have been proposed for frequency-dependent beamforming that spread different frequencies in different directions. However, these beam patterns have limitations in that each user only receives a tiny fraction of the bandwidth in one timeslot, which may not meet their demand in a reasonable time frame. Additionally, these architectures do not provide control over the number of beams, beam directions, and \textbf{\textit{beam-bandwidth}}\footnote{\textbf{Beam-bandwidth} is defined as a fraction of system bandwidth that has high beamforming gain in the desired beam direction and low elsewhere.}, resulting in large chunks of frequency resources being wasted in directions where there is no active user, leading to low spectrum utilization.

We propose a system called \name that performs \textit{flexible directional-frequency multiplexing} by allocating a sub-set of contiguous frequency resources to each user, regardless of their direction, preventing spectrum wastage. This is achieved by creating multiple concurrent pencil beams (multi-beams) in different user directions, where each beam carries a separate frequency band according to the user's demand in that beam direction. A key feature of \name's multi-beam response is that it preserves beamforming gain across all beams while re-distributing power to only desired frequency-direction pairs with minimal leakage in other directions and frequency bands. This set of frequency-direction pairs can be chosen arbitrarily, providing flexibility in performing directional-frequency multiplexing. This enables efficient use of the entire frequency band (up to 800 MHz for 5G NR and 2.3 GHz for IEEE 802.11ax bands), reducing spectrum wastage and providing low latency network access, as shown in Fig. \ref{fig:TDMA_OFDMA_3D}(b).

To implement \name, we introduce a novel mmWave analog array architecture called delay phased array (\dpa) that can generate multi-beams with a flexible number of beams, with arbitrary beam directions and beam bandwidths (Refer to Figure~\ref{fig:dpa_desired_4beam}). Unlike fixed delay architectures, \dpa uses variable delay and phase elements at each antenna to create any desired frequency-direction response. Our insight is to use delays and phases in a complementary manner, with variable delay providing \textit{frequency selectivity} and variable phase providing \textit{direction steerability}, allowing \dpa to create multi-beams towards multiple arbitrarily chosen frequency-direction pairs.

% We architect our \dpa and provide valuable insights on the hardware requirements for creating these antenna arrays. We perform the range and resolution analysis on \dpa to arbitrarily and flexibly transmit/receive any amount of bandwidth in any direction and use the entire spectrum. Implementing large delays on circuits as large transmission lines is size prohibitive and integrating them onto an IC is even more challenging in mmWave due to bandwidth and matching constraints~\cite{ghaderi2019integrated}. Therefore, we design to ensure that the requirement on the range of delay values is significantly lower than that suggested in the literature for other \ttd applications~\cite{li2022rainbow}. \ttd array applications require monotonic increasing delays at each antenna; for even a 16 element array, the delay requirement is 20 ns, 13x more than state of art mmWave delay design~\cite{ghaderi2019integrated}. In contrast, \dpa consists of increasing and decreasing delay values for consecutive antennas with a net effect of lower delay range requirements than traditional \ttd arrays. Moreover, the delay range is independent of the number of antenna elements for large arrays, making our design practical and easy to manufacture.

We architect the design of \dpa and provide insights on the hardware requirements for creating the antenna array, and perform an analysis on the range and resolution of variable delay values required to generate our desired multi-beams. One of the challenges in implementing large delays on circuits is the size and complexity of the transmission lines, and integrating them onto an IC becomes even more difficult at mmWave frequencies due to bandwidth and matching constraints~\cite{ghaderi2019integrated}. Our design addresses this by significantly reducing the range of delay values required, compared to traditional \ttd array designs, making it practical and easy to manufacture. For example, \ttd array requires monotonically increasing delays at each antenna, requiring a delay of 20 ns even for a 16-element array, which is 13x more than state-of-the-art mmWave delay designs~\cite {ghaderi2019integrated}. In contrast, our \dpa design consists of increasing and decreasing delay values for consecutive antennas, resulting in lower delay range requirements than traditional \ttd arrays. Furthermore, the delay range is independent of the number of antenna elements, making it scalable for large arrays.

%\todo{four principles of writing what is the problem, why is it hard, what is simple solution and what is your solution.... this and previous paragraph are not exciting, it should land... }

% A natural next challenge during the operation of \name is to program the \dpa to create the desired frequency selective multi-beam response. To achieve this response, it has to estimate values for variable delays and phases at each antenna. Furthermore, the optimization has to solve for the discrete variable delay and phase elements, which makes the problem non-convex and NP-hard. An exhaustive search over all possible delay and phase combinations at each antenna would be computationally challenging during real-time operation. Moreover, storing the delay and phase value for every possible frequency-space (directions) is impossible, as the combinations of frequency-space are innumerable ($10^{28}$). We develop a single-shot solution for the delays and phases by developing a transform from delay-antenna to frequency-space domain. We represent the desired frequency selective multi-beam response as a 2D image and develop a 2D transform that takes the desired frequency-space image and transforms it into 2D delay-antenna space. We then extract the corresponding delays and phases per-antenna from this image. Our algorithm only requires the angles for each user and corresponding frequency resources allocated to users in the current time slot. The typical 5G mmWave standard protocol estimates each user's direction by default, thus making our algorithm easy to deploy. 
The next challenge is the software programming of \dpa to meet the beamforming requirements of \name. The software should determine the appropriate values for the variable delays and phases at each antenna of \dpa. Solving for these discrete values is computationally difficult as it is a non-convex and NP-hard problem. A naive solution is to pre-compute and store the delay and phase values for every possible frequency-direction pair, but this is infeasible due to a large number of such combinations ($10^{28}$). To overcome this, we develop a novel FSDA (frequency-space to delay-antenna) algorithm which provides a single-shot solution for estimating the delays and phases in real-time. It does this by mapping the desired frequency-space response to the delay-antenna space using a 2D transform and then extracting the corresponding delays and phases for each antenna. FSDA algorithm can be implemented in real-time using fast and efficient 2D FFT techniques. Our algorithm only requires the angles for each user and corresponding frequency resources allocated to users in the current time slot. The angles can be obtained from any standard compliant initial access protocol and so \name can be easily integrated into the standard 5G mmWave protocols.

% \todo{talk about your simulation and emulation results to evaluation efficacy of your performance...}

% \textbf{Contributions:} 
\noindent
In summary,\textit{\textbf{ we make the following contributions}}:
\begin{itemize}
    \item  We propose \name, the first system that enables flexible directional-frequency multiplexing in mmWave networks, achieving higher spectrum usage, low latency, and scalability to support a large number of users.
    \item We design a novel mmWave front-end architecture called delay phased array that can generate multi-beams with a flexible number of beams, beam directions, and beam-bandwidths while maintaining high beamforming gains.
    \item We provide a new algorithm called FSDA (Frequency-space to delay-antenna) which estimates delays and phases in real-time using 2D FFT techniques and can be easily integrated with standard 5G mmWave protocols.
    \item We evaluate the performance of \name using real mmWave traces and show an improvement in latency by 60-150\% compared to baselines. Furthermore, the multi-user sum-throughput is improved by 3.9x compared to true-time-delay array baseline~\cite{li2022rainbow}.
    % , 1.5x compared to split-antenna based design, and 1.3x compared to Phased Array with time division multiple access scheme. The worst case latency is reduced by 72\%.
    % \item \name is the first system that lays the requirements for flexible frequency multiplexing in the mmWave regime. We design a novel mmWave front-end architecture called delay-phased array that meets those requirements achieving higher spectrum usage, low latency, and scalability to support a large number of users.
    % \item We develop a novel 2D transform and algorithm called FSDA (Frequency-space to delay-antenna) to estimate the delays and phases in real-time. 
    % % \todo{talk about transform first and then algorithm}   % \item We present a simplified scheduler for mmWave MAC that builds on top of popular sub-6 proportional-fair scheduler with the ability to support more than one directional beams at the same time.
    % \item We evaluate \name on real mmWave traces and show overall multi-user sum-throughput improvement by 3.9$\times$ compared to a \ttd array baseline~\cite{li2022rainbow}, 1.5$\times$ compared to split-antenna baseline, and 1.3$\times$ compared to phased array with time-division multiple access (TDMA) scheme. The worst case latency is reduced by 72\%.
\end{itemize}

\begin{figure*}
  \begin{minipage}[b]{0.18\linewidth}
    \centering
    \includegraphics[width=\linewidth]{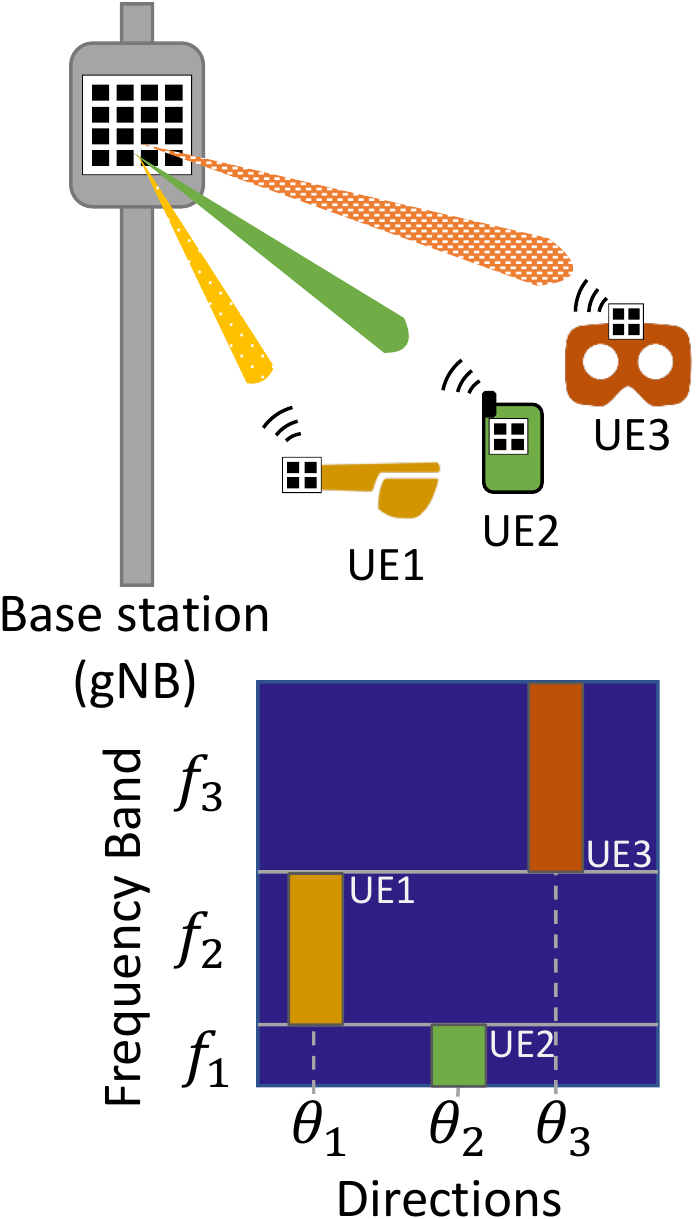}
    \caption{Desired frequency-space beam response for 3 users in 3 different angular directions.}
    \label{fig:dpa_desired_4beam}
  \end{minipage}
  \hspace{0.01\linewidth}
  \begin{minipage}[b]{0.79\linewidth}
    \centering
    \includegraphics[width=\linewidth]{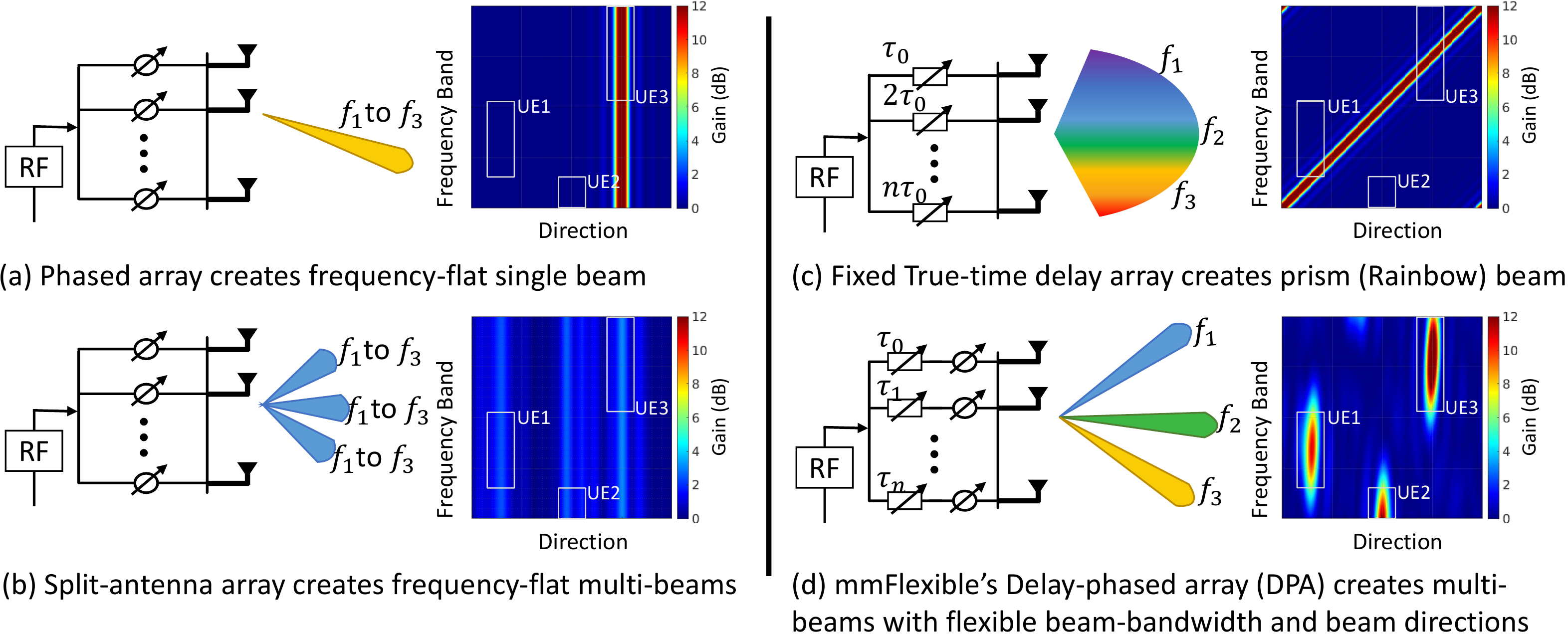}
    \caption{Comparison of different mmWave front-end architectures and corresponding frequency-space patterns they created. The traditional phased array and split-antenna array create a frequency-flat response, while \ttd array creates a frequency-selective rainbow-like pattern radiating in all directions. In contrast, \dpa provides a flexible frequency-direction response with a programmable number of beams, beam-bandwidth, and beam directions. }
    \label{fig:baseline_architecture_and_heatmap}
  \end{minipage}
\end{figure*}

\section{Background and Motivation}\label{sec:motivation}
% More in Infocom v2: Reverse primer and example.
% # New in Infocom v2: Provide short comparison with 3 baselines in table form. 
% New in v3: improve 3.2 first by removing all the maths and explain intuitively and then work in simplifying 3.1
% v4: Completely new way to put this section based on Dinesh's feedback. 2.1 Applications, resource inefficiency. 2.2 Sub-6 OFDMA is great but not possible with mmWave - hybrid, splitantenna baseline etc. 2.3 Our insight/motivation without giving delay architecture. Tell why you need it and not how you use it.

In this section, we would provide a primer on different analog antenna array architectures, which have a single baseband radio frequency (RF) chain and mechanism to best achieve flexibility in resource allocation to multiple users. Next, we would show a realistic example (Figure~\ref{fig:dpa_desired_4beam}) that none of the architecture can meet the requirements for flexible directional-frequency multiplexing, and how our flexible \dpa architecture can meet these requirements.

\subsection{Primer on phased arrays and true-time delay arrays}
\textbf{$\blacksquare$ Phased array:}
The phased array takes a single input from the digital chain, split into $N$ copies, apply appropriate phase-shift, and radiates from $N$ antennas (Fig. \ref{fig:baseline_architecture_and_heatmap}(a)). The input signal with all of its constituent frequency bands is radiated in a direction specified by the phase setting. The set of phases at each antenna constitutes a weight vector $w_{\text{phase}}$ as:
\begin{equation}
    w_{\text{phase}}(n) = e^{j\Phi_n}
\end{equation}
where $\Phi_n$ is the programmable phases for antenna index $n (n\in [0,N-1])$. Now, if we program the phases to create a directional beam at an angle say 30$^\circ$, then this phase shift is applied to all the frequency components in the signal radiating them along that same angle 30$^\circ$. Different frequency components in the signal cannot be radiated to other directions because of the frequency-independent nature of weights in a phased array. Therefore, a phased array can serve users in only one direction at a time and cannot perform flexible directional-frequency multiplexing with many users. Split-antenna phased array (Fig. \ref{fig:baseline_architecture_and_heatmap}(b)) uses the phased array architecture split into multiple sub-arrays to create multiple beams towards different users in one TTI, but suffers from lower antenna gain, range, and throughput.

\noindent
\textbf{$\blacksquare$ True-time delay array:} The true-time delay (\ttd) array uses a delay element to replace the phase shift element as shown in Fig. \ref{fig:baseline_architecture_and_heatmap}(c). The antenna weights with delay element is given by:
\begin{equation}
    w_{\text{delay}}(t,n) = \delta(t-\tau_n)
\end{equation}
Past work on \ttd array~\cite{yan2019wideband,boljanovic2020true,boljanovic2021fast, li2022rainbow} have shown that this architecture creates a prism-type response by radiating all the frequencies in all the directions in a linear fashion (Fig.~\ref{fig:baseline_architecture_and_heatmap}(c)) by using a fixed set of delays at each antenna given by $\tau_n = \frac{n}{\text{B}}$, for bandwidth $B$. This response has two major limitations: 1) most frequency bands would be wasted in space if there are no user in that direction and 2) each user gets a small frequency band and cannot scale it up arbitrarily. Due to these limitations, \ttd array is not suitable for flexible directional-frequency multiplexing.

\subsection{Current mmWave architectures are incapable for flexible directional-frequency multiplexing}
% We compare four analog beamforming architectures for mmWave front-end antenna array and show how they differ in terms of network performance metrics such as latency, effective spectrum usage, and throughput. Those are: 1. phased array, 2. split-antenna array, 3. true-time delay array, and our proposed 4. \dpa; each differ in how they create analog beams towards the users. The architecture details and beampatterns (along both frequency and space dimensions) for all these architectures is shown in Fig.~\ref{fig:baseline_architecture_and_heatmap}. 

Let us take a simple network scenario to understand the performance tradeoffs for the above systems. Typically for 4K VR applications, end-end latency should be $< 100$ ms, Where transmission goes from the public cloud, network provider, base station (edge), and device. Over-the-air (base station to device) transmission latency comes down to a strict latency requirement of $< 1$ ms \cite{qualcomm2017augmented}. Consider 10 users in the network with similar traffic demand: each user requires 60 Mbps (4K VR) throughput and $<1$ ms over-the-air latency due to the interactive nature of VR~\cite{pocovi2018achieving,3gpp2020,qualcomm2017augmented,mangiante2017vr}. The aggregate throughput provided by the users is 600 Mbps, only 30\% of the max physical layer throughput of a 400 MHz FR2 gNB in downlink (2.2 Gbps). We assume all users are distinctly located (in different angular directions) and RF conditions are good (e.g. LOS and no blockage). This allows us to control common external factors and exclusively evaluate architectural capabilities. We assume our \dpa and split array can create up to 8 concurrent beams and \ttd array creates a prism beam pattern. The comparison is shown in Table~\ref{tab:compare_baseline_background}. From our end-end evaluations (Fig.~\ref{fig:latency_only}), we can see that in edge scenarios all the baselines failed to meet the $1 $ ms over-the-air latency requirement. It is evident that \name is the best to support latency-critical applications while satisfying higher throughput demands.
\begin{table}[h]
    \centering
    \footnotesize
    \begin{tabular}{|c|c|c|c|}\hline
        Architecture & Latency* (ms) & Packet Loss & Throughput \\\hline
        phased array (TDMA) & $> 1.25$ ms & 24.0\% & 54.9 Mbps \\\hline
        split-antenna array & $> 2$ ms  & 33.3\% & 47.3 Mbps \\\hline
        \ttd array & $> 3.5 $ ms & 76.4\% & 18.3 Mbps \\\hline
        \dpa & $< 0.5$ ms & 0.0\% & 71.3 Mbps\\\hline
    \end{tabular}
    \caption{Delay-phased array \& other baselines for 10 4k VR users (*We mention over-the-air worst case latency).}
    \label{tab:compare_baseline_background}
\end{table}

\begin{figure*} [!t]
\subfigure[Desired F-S image]{
    \includegraphics[width=0.28\columnwidth]{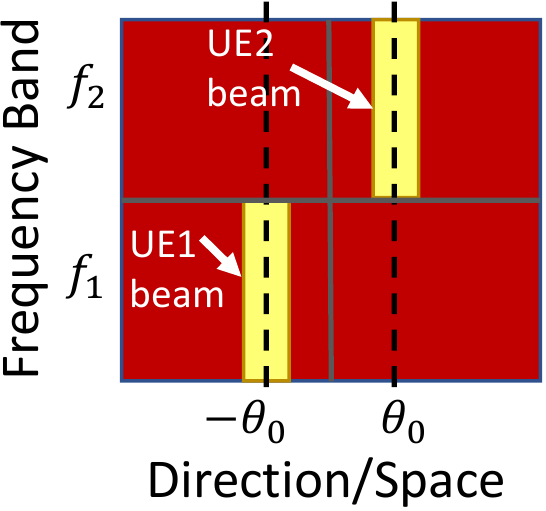}
    \label{fig:fsda_explain_desired}
  }
  \subfigure[Specific Frequency-Space (F-S) signature images]{
    \includegraphics[width=.88\columnwidth]{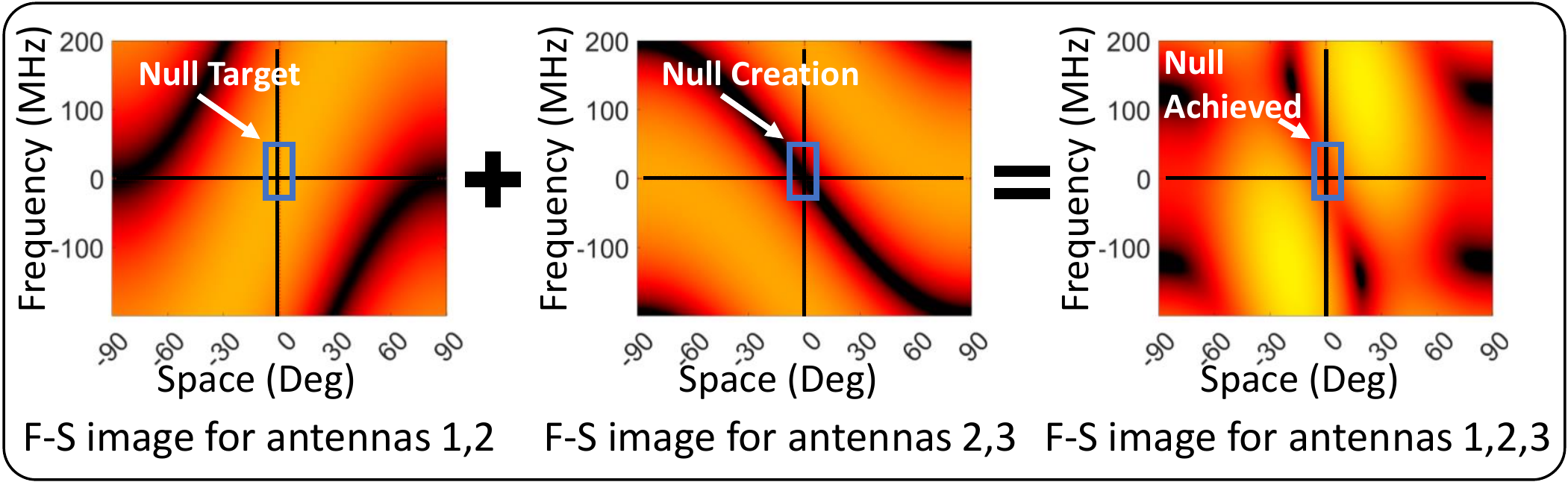}
    \label{fig:fsda_explain_v2}
  }
  \subfigure[8 antenna F-S image]{
    \includegraphics[width=.39\columnwidth]{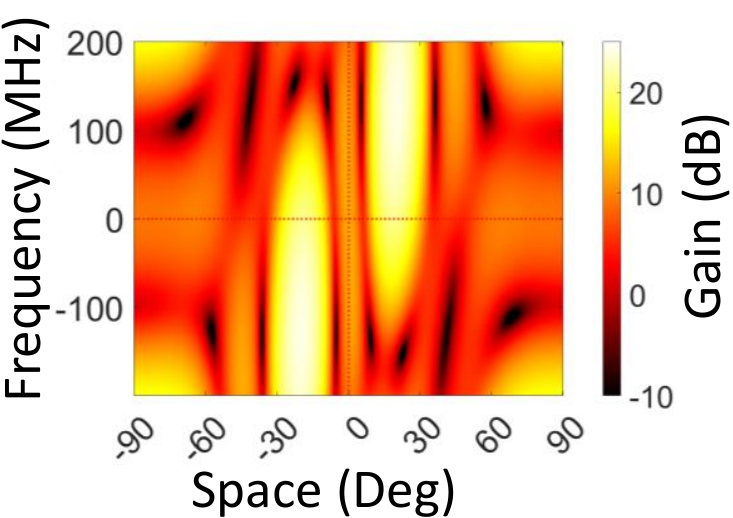}
    \label{fig:fsda_explain_8antenna}
  }
  \subfigure[Delays and phases]{
    \includegraphics[width=.3\columnwidth]{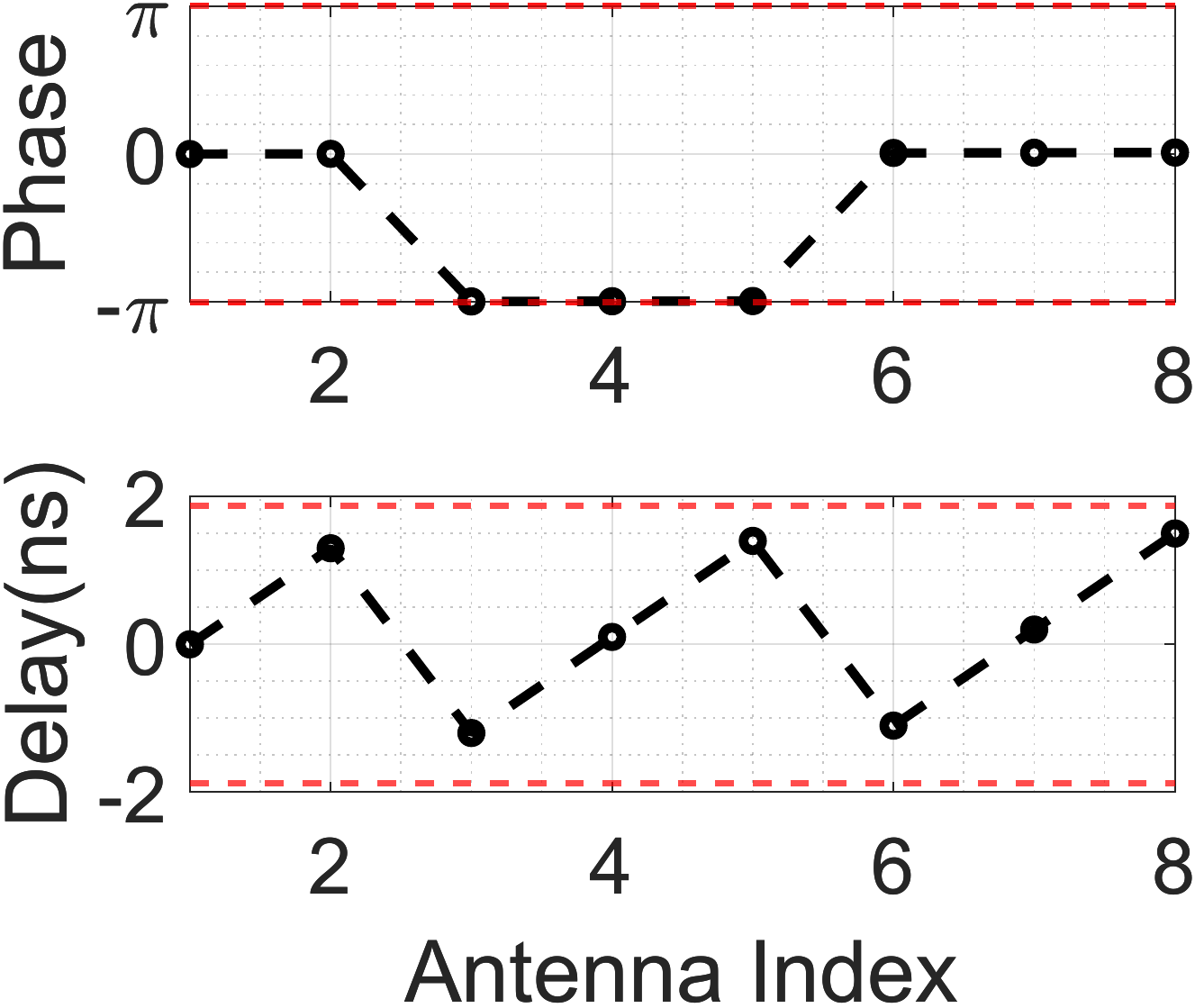}
    \label{fig:phase_delay_m20_20case}
  }
    \caption{Explanation on how \dpa creates a desired frequency-space image with two beams at $-20^\circ$ and $20^\circ$. The first beam occupies a frequency band in [-200,0] MHz and the second beam occupies a frequency band of (0,200] MHz. \dpa can create the desired image in (c) using delay and phase values in (d).}
    \label{fig:fsda_explain}
    % \vspace{2pt}
\end{figure*}

% \begin{figure*}[t]
%     \centering
%     \includegraphics[width=0.8\textwidth]{figures/fsda_explain.pdf}
%     \caption{Explanation on how \dpa creates a desired frequency-space image in (a) with two beams at $0^\circ$ and $30^\circ$. The $0^\circ$ beam occupies frequency range in [-200,0] MHz and $30^\circ$ beam occupies frequency range of [0,200] MHz. Specific frequency-space signatures are shown in (b) and (c) with antenna sub-arrays and final frequency-space image is shown in (d) and (e) with 4 and 8 antennas respectively.}
%     \label{fig:fsda_explain}
% \end{figure*}

\section{\name's \dpa hardware design} \label{sec:design}
\name introduces a new mmWave front-end architecture, delay-phased array (\dpa), to enable flexible directional-frequency multiplexing for mmWave networks. The \dpa design addresses the limitations of existing mmWave systems with  phased arrays and \ttd arrays by creating a multi-beam response with flexible beam directions and beam bandwidths. In this section, we discuss the practical implementation of \dpa and show how we can build it with shorter delays.

% \name designs a new mmWave front-end architecture called delay-phased array (\dpa) to enable flexible directional-frequency multiplexing for mmWave network. 
% We have seen in ~\cref{sec:motivation} how the existing mmWave phased array and \ttd array fail to create the desired frequency-direction response. 
% In this section, we describe \name's \dpa architecture and how we program it to create the desired multi-beam response with controllable beam-bandwidth. We then present a range and resolution analysis for the programmable delay element and provide a discussion on the practical implementation of \dpa.

% \noindent
\subsection{ Architecture of delay-phased array (DPA)} 
% \todo{Describe from Figure 2 and focus on simple expression of G(f,theta)}
\dpa architecture consists of programmable delay and programmable phase element per antenna with a single-RF chain as shown in Fig.~\ref{fig:baseline_architecture_and_heatmap}(d). These elements can be programmed together to create flexible beam responses that are not possible by either of the two elements alone.  Our insight is to control two knobs: delays $\tau_n$ and phase $\Phi_n$ to get the desired response. We define beam weights for \dpa as:
\begin{equation}
    w_{\text{dpa}}(t,n) = w_{\text{phase}}(n)w_{\text{delay}}(t,n)= e^{j\Phi_n}\delta(t-\tau_n)
\end{equation}
Notice the dependence of \dpa weights with time, which leads to a dependence on frequency. Upon taking FFT, the exponential term in the weights become $\Phi_n-2\pi f \tau_n$, which is a function of frequency $f$ at antenna index $n$.
% We can understand the frequency dependency of \dpa weights by taking FFT as: 
% Before we present our approach to achieving these requirements, we would develop intuition about the analog antenna arrays and how phase and delay in the analog antenna create different $G(f,\theta)$ radiation responses. We would then build on intuition to distill it into our architecture. At a high-level, our \dpa consists of delay and phase elements (Fig.~\ref{fig:baseline_architecture_and_heatmap}(d)). 
% \begin{equation}
%     w_\text{dpa}^f(f,n) = e^{j\Phi_n+j2\pi f \tau_n}
% \end{equation}
% Notice the weight vector is a function of both antenna index $n$ and frequency $f$. 
Therefore, the beamforming response of this weight vector will be a function of both frequency and direction.
The beamforming gain of an antenna array represents the power radiated by the antenna array in different directions. The expression for beamforming gain $G(f,\theta)$ for \dpa  at a frequency $f$ and direction $\theta$ is:
\begin{equation}
\begin{split}
        G(f,\theta)  &= \sum_{n=0}^{N-1} \mathcal{F}(w_\text{dpa}(t,n)) e^{-jn\pi\sin(\theta)}
        % &= \sum_{n=0}^{N-1} e^{j\Phi_n+j2\pi f \tau_n} e^{-jn\pi\sin(\theta)}
\end{split}
\end{equation}
where $\mathcal{F}$ is Fourier transform and  $e^{-jn\pi\sin(\theta)}$ is the standard steering vector transformation from the antenna to space ($\sin(\theta)$)-domain\footnote{Note that the steering vector for a linear antenna array is given by  $e^{-jn2\pi \frac{d}{\lambda} \sin(\theta)}$, which depends on the array geometry (antenna spacing $d$) and signal wavelength $\lambda$~\cite{benesty2019array}. We assume $d=\frac{\lambda}{2}$ and approximate the steering vector as $e^{-jn\pi\sin(\theta)}$.}.
Essentially, the response is the sum of the individual contribution from all antennas. By equating exponential terms of \dpa weights to that of array response, we get: $\Phi_n - 2\pi f \tau_n = n\pi \sin(\theta)$. The variable delay causes a slope in frequency ($f$) - space ($\sin \theta$) plot, while the variable phase causes a constant shift along the space axis. Together, they can create an arbitrary line with a configurable slope and intercept in the frequency-space domain. With this insight in place, we will discuss how to program the phase and delay values at each antenna to get the desired beam response in Section \ref{sec:design_software}.

% Note that we make an approximation above that the distance between antenna elements is the half wavelength (i.e., $d=\frac{\lambda}{2}$), valid for the frequency band of our interest~\cite{benesty2019array}.  

% \noindent
\subsection{ Range of delay element in \dpa}
Before discussing \dpa software programming, we emphasize that it is important to analyze the set of possible delay values that the hardware can support practically. Here we describe the requirements for the delay range and how it helps create a practical circuit board. Delay elements are implemented with variable-length transmission lines on a circuit board. Building large delay lines in IC at mmWave frequencies is prohibitive because of large size, bandwidth, and matching constraints~\cite{ghaderi2019integrated}. Therefore, our design ensures that the delay range is not too large. While traditional true-time delay (TTD) array requires a delay range proportional to the number of antenna $N$, which is large for large antenna array (18.75 ns for 16 antenna array)~\cite{li2022rainbow}. In contrast, the delay range for \name is independent of the number of antennas. For the two-beam case, the delay range for \name is $\frac{3}{2B}$ (shown later in (\ref{eq:tau_closed})), which is $3.7 ns$ for 400 MHz bandwidth for 5G NR; significantly less than that required by TTD arrays. The delay range increases with the number of concurrent beams, but is independent of the number of antennas, making it scalable to large arrays.

Delay control with sub-ns accuracy has been demonstrated in full-duplex circuits for interference cancellation~\cite{nagulu2021full}. Recently, authors of~\cite{ghaderi2020four} have shown accurate delay control with 0.1 ns resolution and with 6-bit control (64 values until 6.4 ns), which satisfies the requirement of \name.

\section{\name's \dpa Software Design}\label{sec:design_software}
\subsection{Requirements for \name}
Our goal is to construct arbitrary frequency-direction response $G(f,\theta)$ via \dpa architecture, which would be energy efficient and enable efficient resource utilization with low latency. If we carefully notice the example in section 2, we observe that the split antenna achieves the frequency-direction mapping but with a loss of 6 dB SNR, which is a corollary of the entire 400 MHz radiated in each of the four directions. It leads to our first requirement:   
\noindent
\textbf{Req.1:} \textit{
% The signal must be transmitted/received in frequency-direction used by users, i.e., each band of frequencies must be radiated only in the direction of users they are associated with, and frequencies not associated with the user-direction must have no energy, i.e., create a null at those frequencies.
The system must be able to transmit/receive signals in the specific frequency-direction pairs associated with each user, with minimal energy leakage in other directions and frequencies. 
}
Furthermore, we should be able to control the amount of bandwidth assigned to each user, which leads to: 
\noindent
\textbf{Req.2:} \textit{
% Creating beams narrow in space (for higher antenna gain) but arbitrarily wide in the frequency domain (to support high demand users).
Flexibility in allocating bandwidth to each user, allowing for narrow beams in space for higher antenna gain and wide beams in frequency to support high-demand users.
}

% Our insight is to control two knobs: delays $\tau_n$ and phase $\Phi_n$ to get the desired response. 
% By taking FFT and equating exponential phase to that of array response, we get: $\Phi_n + 2\pi f \tau_n = n\pi \sin(\theta)$. The variable-delay causes a slope in frequency ($f$) - space ($\sin \theta$) plot, while the variable-phase causes a constant shift along the space axis. Together, they can create an arbitrary line with configurable slope and intercept in frequency-space domain. 

% \noindent
% \textbf{Estimating delay and phase values:}
% \subsection{Two-beam example scenario}

% Now, we ask what set of delays and phases per antenna would give us the beamforming gain pattern with the desired beam-bandwidth and beam direction. We first consider a simple case of two beams with equal beam-bandwidth of $B/2$ each, where $B$ is the total system bandwidth. We assume the two beams are directed along $(-\theta_0, \theta_0)$ respectively as shown in Figure \ref{fig:fsda_explain}(a). We will later discuss a general case with arbitrary beam-bandwidth and beam direction.

% We emphasize key takeaways from these expressions before jumping into their proof. Note that the delay values are bounded within a range of $\frac{3}{2B}$ independent of the number of antennas. Within this range, the delay values will monotonically increase or decrease with antenna index $n$, but for large $n$, the delay wraps around with this range factor. 

\subsection{Meeting \name's requirements with \dpa} 
With this intuition in place, we revisit and explain how can \name achieve these requirements through a simple two-user example in Fig.~\ref{fig:fsda_explain}. Let us consider the two users are located at $-\theta_0$ and $\theta_0$ respectively, and the base station wishes to serve these two users with equal beam-bandwidth of $B/2$ each, where $B$ is the total system bandwidth. To support such flexible directional-frequency multiplexing, the base station must create a frequency-direction beam response shown in Fig.~\ref{fig:fsda_explain}(a). We call such 2D beam patterns as \textit{frequency-space (F-S) images} for simplicity. So how does \dpa create these images and meet the above requirements? 

We provide a closed-form expression for the set of delays $\tau_n$ and phases $\Phi_n$ for each antenna that would generate the above beamforming response as follow:
\begin{equation} \label{eq:tau_closed}
    \tau_n = \left(\frac{3}{2B}n\sin(\theta_0) +\frac{3}{4B}\right)\;\;\;\text{mod }\frac{3}{2B} 
\end{equation}
\begin{equation}\label{eq:phase_closed}
    \Phi_n = \text{round}(n\sin(\theta_0))\pi \;\;\;\text{mod }2\pi
\end{equation}

% \begin{equation} \label{eq:tau_closed}
%     \tau_n = \left(\frac{3}{2B}n\sin(\theta_0)\right)\;\;\;\text{mod }\frac{3}{2B} 
% \end{equation}
% \begin{equation}\label{eq:phase_closed}
%     \Phi_n = 
%     \begin{cases}
%         0 & \text{when}\;\;\cos(n\pi\sin(\theta_0))> 0 \\
%         \pi & \text{when}\;\;\cos(n\pi\sin(\theta_0))<0
%     \end{cases}
% \end{equation}
\textbf{\textit{A more generalized expression for an arbitrary number of beams, beam directions, and beam bandwidth along with their proof can be found in Appendix~\ref{app:proof}.}}
% is omitted for brevity and can be found at  \url{https://wcsng.ucsd.edu/dpa}.
% Let us take a simple example how we create a desired frequency-space image in Fig.~\ref{fig:fsda_explain}(d) with two distinctly located users where we stream half frequency bands to one user and another half to other user. We define two requirements for creating such response as follows:

Now, we achieve the requirements for \name by assigning a complementary F-S images to a subset of antennas. For instance, we create a positive slope in F-S image using antennas 1 and 2, and then create a complementary negative slope with antenna 2 and 3 as shown in  Fig.~\ref{fig:fsda_explain}(b). 
% When the two response are combined together, the resulting response for antenna 1,2, and 3 is the desired F-S image with high energy in the desired 
% two antenna sub-arrays of the given antenna array to create complementary responses as shown in Fig.~\ref{fig:fsda_explain}(b). The first sub-array uses the first two antennas to create a frequency-space image with a positive slope, while the second sub-array creates a negative slope. We can achieve a positive or negative slope by increasing or decreasing delay values at consecutive antennas. 
% Also, note there are two slopes in Fig.~\ref{fig:fsda_explain}(c) due to aliasing. It happens when relative delay between antenna is larger than $\frac{1}{B}$~\cite{li2022rainbow} and is akin to aliasing due to $d>\frac{\lambda}{2}$ in standard phased arrays~\cite{benesty2019array}. 
% \footnote{Fig.~\ref{fig:fsda_explain}(b) shown two slopes in frequency-space image which can be explained by aliasing feature in phased array. Interested readers are directed to~\cite{benesty2019array} for more details on aliasing}. 
When the two responses are combined together, we observe a frequency-space image where they combine constructively at desired user locations while creating a null (low gain) at other locations (Meeting \textbf{Req-1}). 

For the second requirement, we create such beams by choosing the number of antennas for creating a positive or negative slope. The intuition is that a higher number of antennas makes the corresponding signature image narrow in space. For instance, it is clear from Fig.~\ref{fig:fsda_explain}(d) that three consecutive antennas (e.g. antenna 3,4,5) have increasing delays, while only two antennae (e.g. 2,3) have decreasing delays. This helps in making the positive slope in the F-S image narrow (3 antenna contribution), while the negative slope remains wide (2 antenna contribution). This effect causes beams that are narrow in space, but arbitrarily wide in frequency as shown in Fig.~\ref{fig:fsda_explain}(c) (Meeting \textbf{Req-2}). 

% For instance, we show a sharp F-S image with 8 antenna due to higher constructive combining gains in Fig.~\ref{fig:fsda_explain}(c) and the corresponding delay and phase values that generate this image in Fig.~\ref{fig:fsda_explain}(d). We observe from delay plot that three consecutive antenna (e.g. antenna 3,4,5) generate a positive line, while only two antenna (e.g. 2,3) generate the negative line.
% it is easy to visualize from Fig.~\ref{fig:fsda_explain}(b) and (c) that the former response is narrow in space (obtained with 3 antennas) and the later is wider in space (obtained with 2 antennas). In fact, if we have more antennas, say 8 of them, then we can get an even `sharper' response as shown in Fig.~\ref{fig:fsda_explain}(e) due to higher constructive combining gains. 
% So, in summary, we could achieve an arbitrary frequency-space image by controlling the variable phases and variable-delay knobs of different sub-arrays and the number of antennas in each sub-arrays (Meeting \textbf{Req-2}).

This intuition helps in understanding how a simple 2-beam frequency-space image is created. We use this insight to develop a novel \algoname algorithm that estimates the delay and phase values for any frequency-space image with an arbitrary number of beams, beam directions, and beam-bandwidths.

% but it does not help yet on what values of delays and phases we should apply at each antenna to create this image or rather how to estimate these values of delays and phases accurately in real-time. We develop a novel \algoname algorithm to estimate delays and phases in real-time.
% in section \ref{sec:design-est}.

\begin{figure*}[t]
    \centering
    \includegraphics[width=0.7\textwidth]{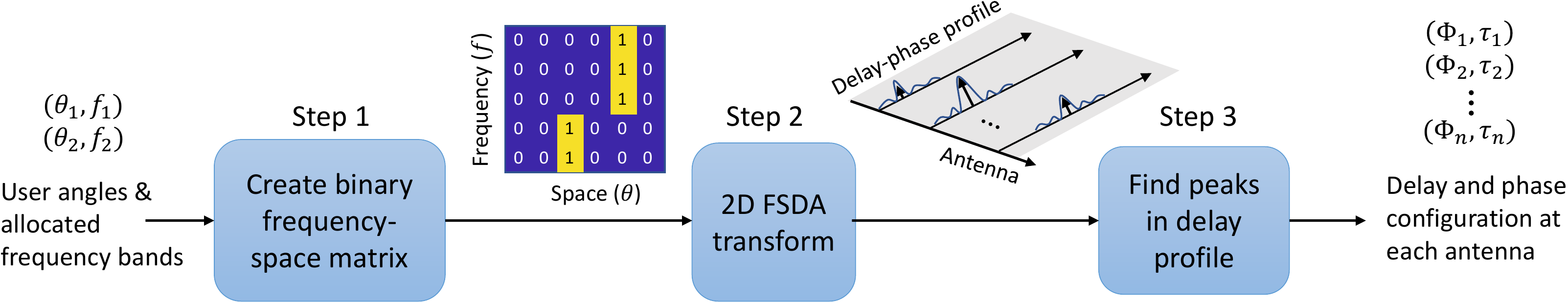}
    \caption{Three steps of our \algoname algorithm to estimate per-antenna delays and phases that generate a desired frequency-space image. The core step is our novel 2D FSDA transform with some pre-processing and post-processing of inputs and outputs respectively. }
    \label{fig:algo_fsda_steps}
\end{figure*}

%% New subsection
\subsection{\algoname algorithm for estimating delays and phases in \dpa}\label{sec:design-est}
To create  a desired frequency-direction beam response, the base station needs to estimate the corresponding delays and phases per-antenna in \dpa. One naive solution is to try different discrete values in a brute-force way using a look-up table to get the desired beams. But, this solution is computationally hard and memory intensive since there is a large set of possibilities for the delays and phases at each antenna. For instance, with 64 delays and 64 phases per antenna (assume both are 6-bit, so $2^6=64$), a brute-force look-up table search would require $64^N\times 64^N \approx 10^{28}$ probes to try each combination and then storing it all in memory, which is impossible to solve with even high memory and high computing machines. Our insight is that we can pose this problem as an optimization framework and solve them in a computationally efficient way.

We now formulate the optimization problem with insights we have obtained from the previous subsection and from the fundamentals of digital signal processing. The goal is to relate weights (delays and phases) to the antenna gain pattern in (\ref{eq:gain}) in a way that simplifies our estimation problem. 

Our insight is that similar to how frequency and time are related by a Fourier transform, there is a similar transform that relates space and antenna using steering matrices. So, there are two transforms that bridges the world of antenna weights to the desired gain pattern: time to frequency transform and antenna to space transform.
Mathematically, we re-write the gain pattern of \dpa to emphasize this 2D transform:
\begin{equation}\label{eq:gain}
    G(f,\theta) = \sum_{k = 0}^{K-1} \sum_{n=0}^{N-1} \Ubf(f,k) w_\text{dpa}(k,n) V(n,\theta)
\end{equation}
where $\Ubf(f,k)$ is a discrete domain Fourier transform (DFT) and the steering matrix $V$ is defined per-element as $V(n,\theta) = e^{-jn\pi\sin(\theta)}$. Now since the signal is actually sampled only discretely with a sampling time of $T_s$, our original delay weight element $w_{\text{delay}}(t,n)$ would reduce to $w_{\text{delay}}(k,n) = \delta(kT_s-\tau_n)$ as described in (\ref{eq:gain}).

We then represent the gain pattern by a discrete frequency-space matrix $\Gbf$ and the weights as discrete time-antenna matrix $\Wbf$ and relate them with the following 2D transform:
\begin{equation}
    \Gbf = \Ubf\Wbf\Vbf
\end{equation}
where $\Ubf$ is time to frequency transform matrix and $\Vbf$ is antenna to space transform matrix. Here we formulate $\Wbf$ as $K\times N$ matrix, where $K$ is the number of discrete time values and $N$ number of antennas. 

We follow a three-step process to estimate the weight matrix $\Wbf$ that creates our desired frequency-space image intuitively explained in Fig. \ref{fig:algo_fsda_steps}. There are two inputs to our algorithm: Angles and desired frequency bands for each user. These two inputs are enough to represent the given frequency-space image. As a first, we create a binary frequency-space matrix that consists of 1s at desired frequency-space locations and 0s otherwise, we denote it by $\Gbf_{\des}$. We then formulate the following optimization problem:
\begin{equation}
\begin{split}
        \hat{\Phi_n},\hat{\tau_n} &= \min ||\Gbf_\des - \Ubf\Wbf\Vbf||^2\\
        \text{s.t.}\quad &\Wbf(k,n) = e^{j\Phi_n}\delta(kT-\tau_n)
\end{split}
\end{equation}
This optimization is a non-convex due to the non-linear terms such as exponential in phase and delta in delay. Moreover, the constraint of having a discrete set of values for delays and phases makes it NP-hard. We make an approximation by relaxing the delta constraint and letting the weights at each antenna take any variation over time. It means that we allow weights to take the form of a continuous profile over time  at each antenna rather than a delta function which is non-zero at only one value and zero otherwise. We call it the delay-phase profile at each antenna. How do we estimate this delay-phase profile?

Our insight is that we can write an inverse transform of $\Ubf$ and $\Vbf$ to go from the frequency-space domain to the time-antenna domain. The logic behind such formulation is that using an appropriate discrete grid along the time and space axis, we can formulate $\Ubf$ and $\Vbf$ as linear transforms, i.e., $\Ubf^\dagger\Ubf=\Ibf$ and $\Vbf\Vbf^\dagger=\Ibf$ for identity matrix $\Ibf$ (Note $(.)^\dagger$ is pseudo-inverse of a matrix). Therefore, it is easy to write their inverse by simply taking the pseudo-inverse. We estimate $\hat{\Wbf}$ as:
\begin{equation}
    \hat{\Wbf} = \Ubf^{\dagger}\Gbf_\des\Vbf^{\dagger}
\end{equation}
The final step of our algorithm is to extract delays and phases from $\hat{\Wbf}$. Note that each column in $\hat{\Wbf}$ contains the delay-phase profile. We find the maximum peak in this profile and the index corresponding to this peak gives the delay and the max value at this peak gives the phase term. Note that since we did not put any restriction on the number of non-zero delay taps, we could get more than one delay tap per antenna. We empirically found that the estimated delay profile has only one significant peak with high magnitude than other local peaks (See Section \ref{sec:evaluation}). Also, the intuition comes from our insights from the previous section that usually one delay per antenna suffices in creating the desired response. 

\noindent
\textbf{$\blacksquare$ Weights Quantization:} The delay and phase values obtained from \algoname algorithm are still continuous in nature and must be discretized to be fed into the \dpa hardware. We quantize both the phase and delay values with a 6-bit quantizer in software before feeding to the array. The quantized phase takes one of the 64 values in $[0^\circ,360^\circ)$ and the quantized delay varies in the range [0, 6.4ns) with an increment of 0.1 ns.

\noindent
\textbf{$\blacksquare$ Computation complexity of \algoname:} 
% \name's \algoname algorithm computes the delays and phases at each antenna given the number of beams, their beam directions, and beam bandwidths. 
The run-time complexity of \algoname is dominated by the 2D FFT transform on a given frequency-space image. Given the frequency, the axis is divided into $M$ subcarriers and the space axis into $D$ directions, the run-time complexity is $\order{(MD(\log(M)+\log(D))}$.

\begin{figure}[t]
    \centering
    \includegraphics[width=0.48\textwidth]{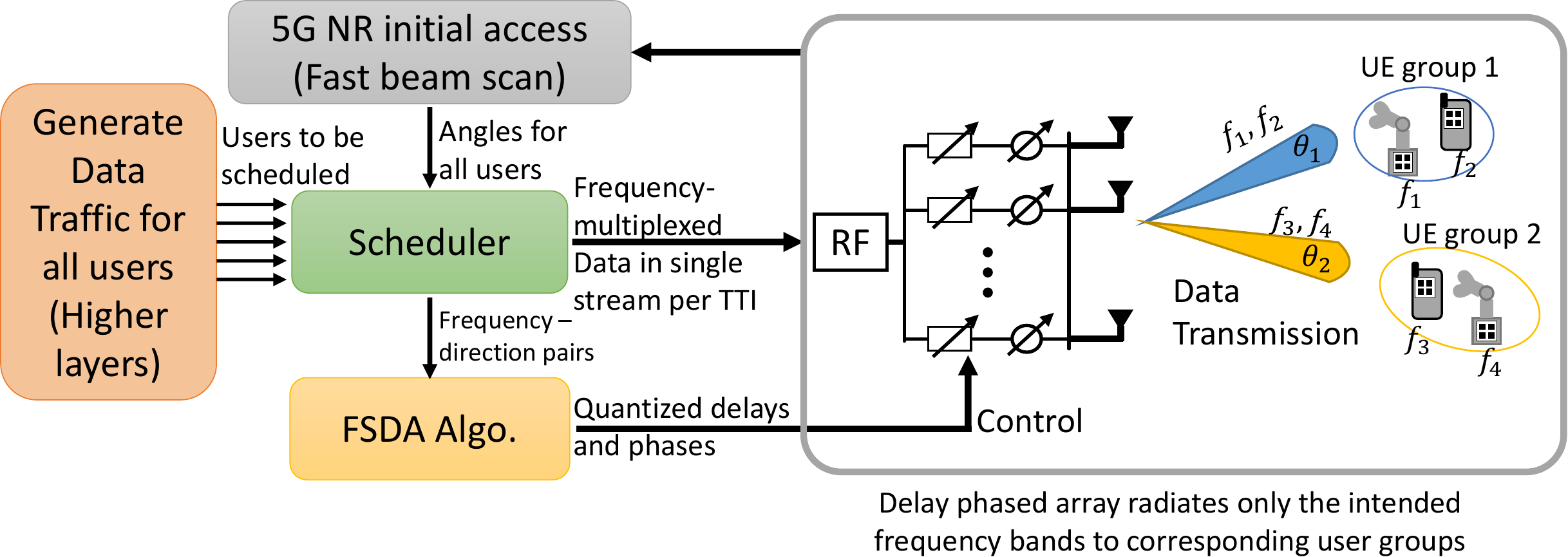}
    \caption{Implementation overview for \name with four main components: Data traffic generator, Fast beam scan angle estimation, a scheduler, and our \algoname algorithm. 
    % The scheduler frequency-multiplex multi-user data in orthogonal set of subcarriers in a single stream per TTI that is sent to the \dpa with a single RF chain. The \algoname algorithm creates appropriate delays and phases to configure the \dpa to radiate only the intended frequency bands to the corresponding user groups.
    }
    \label{fig:overview_fsda_scheduler}
\end{figure}

% \input{4_puttingtogether}
% \input{5_implementation}
% !TEX root = main.tex
\section{Evaluation}\label{sec:evaluation}
\begin{figure*} [!t]
\centering
\subfigure[Latency]{
    \includegraphics[width=0.17\textwidth]{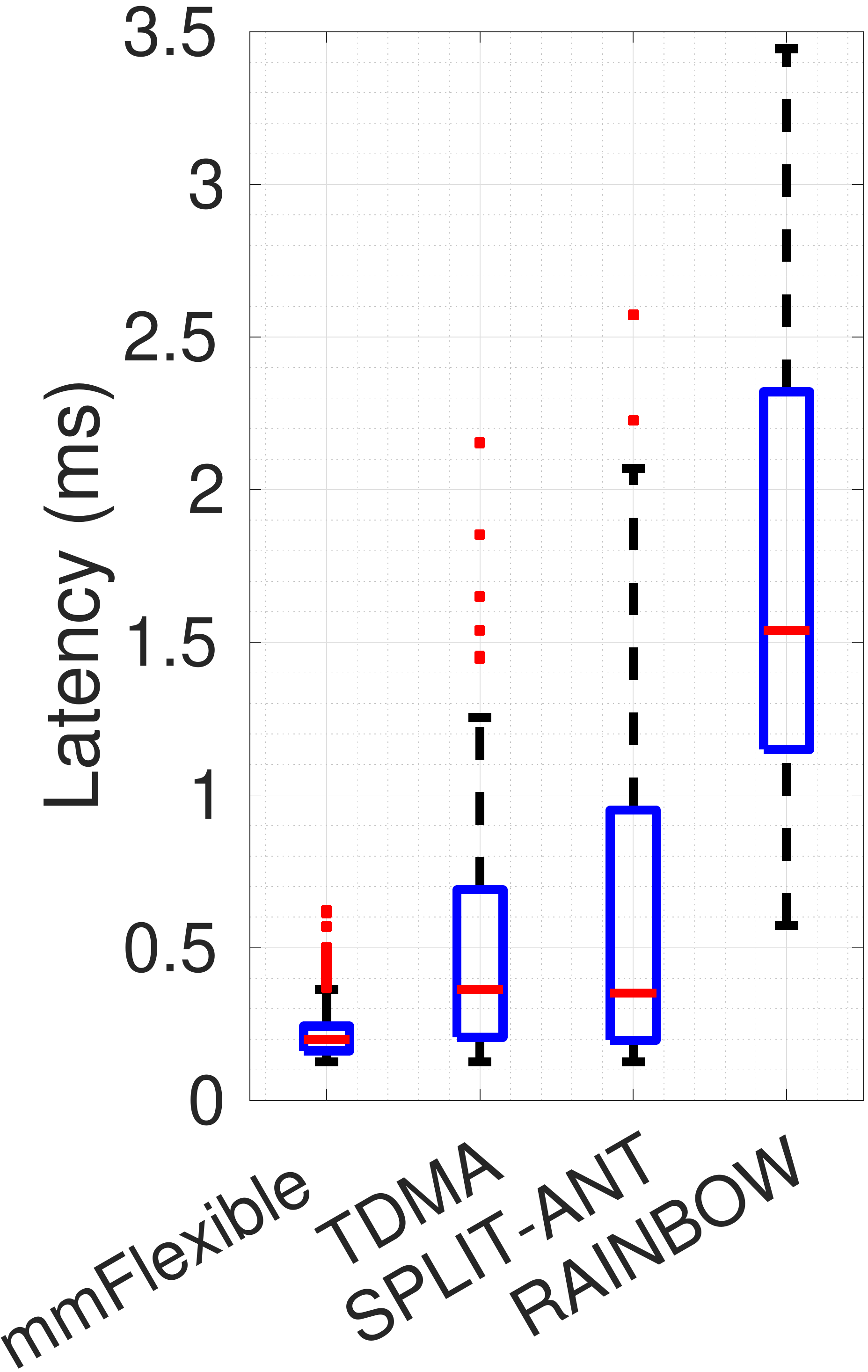}
    \label{fig:latency_only}
  }\hfill
  \subfigure[Loss for $<$1ms latency]{
    \includegraphics[width=.17\textwidth]{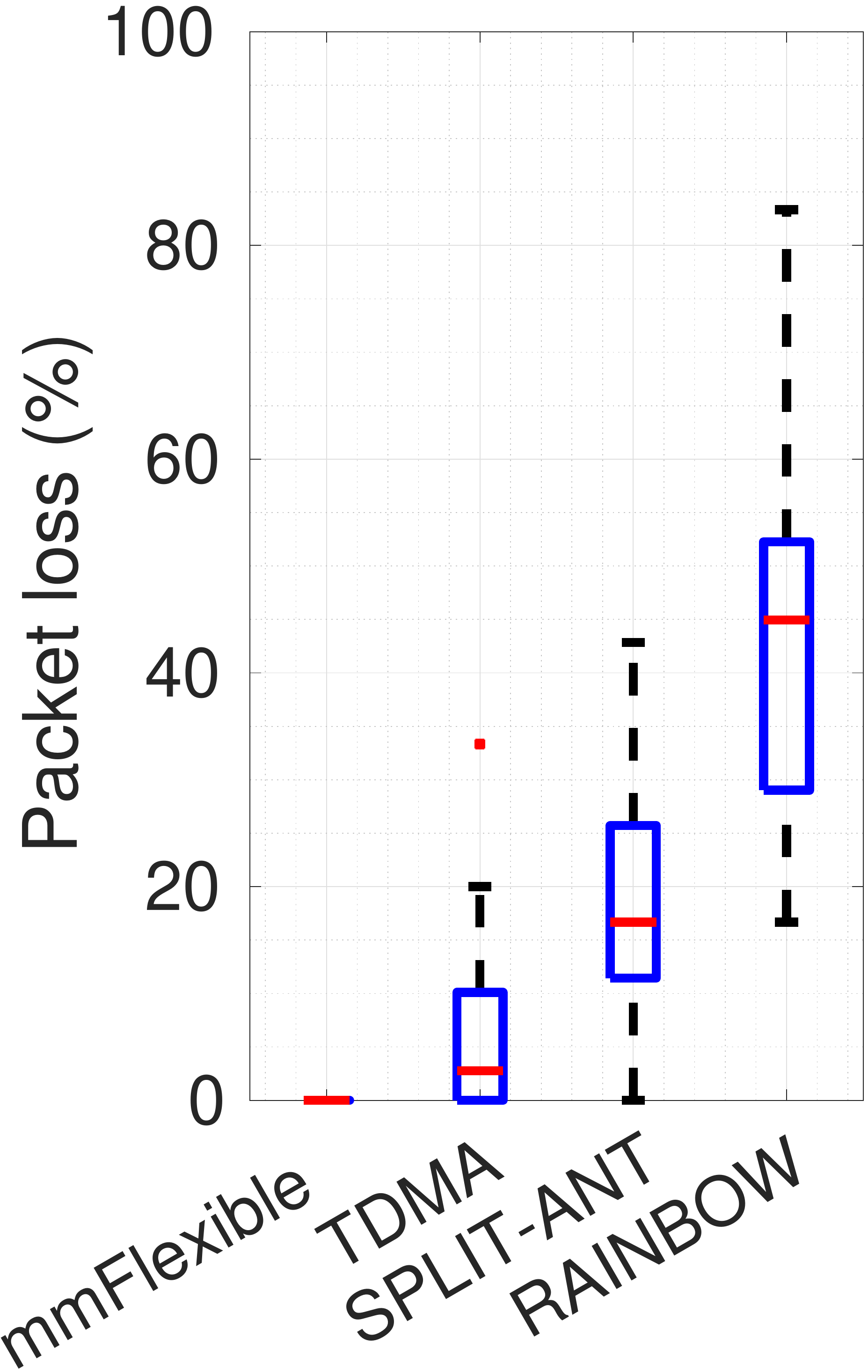}
    \label{fig:droppack_1ms}
  }\hfill
  \subfigure[Thput for $<$1ms latency]{
    \includegraphics[width=.17\textwidth]{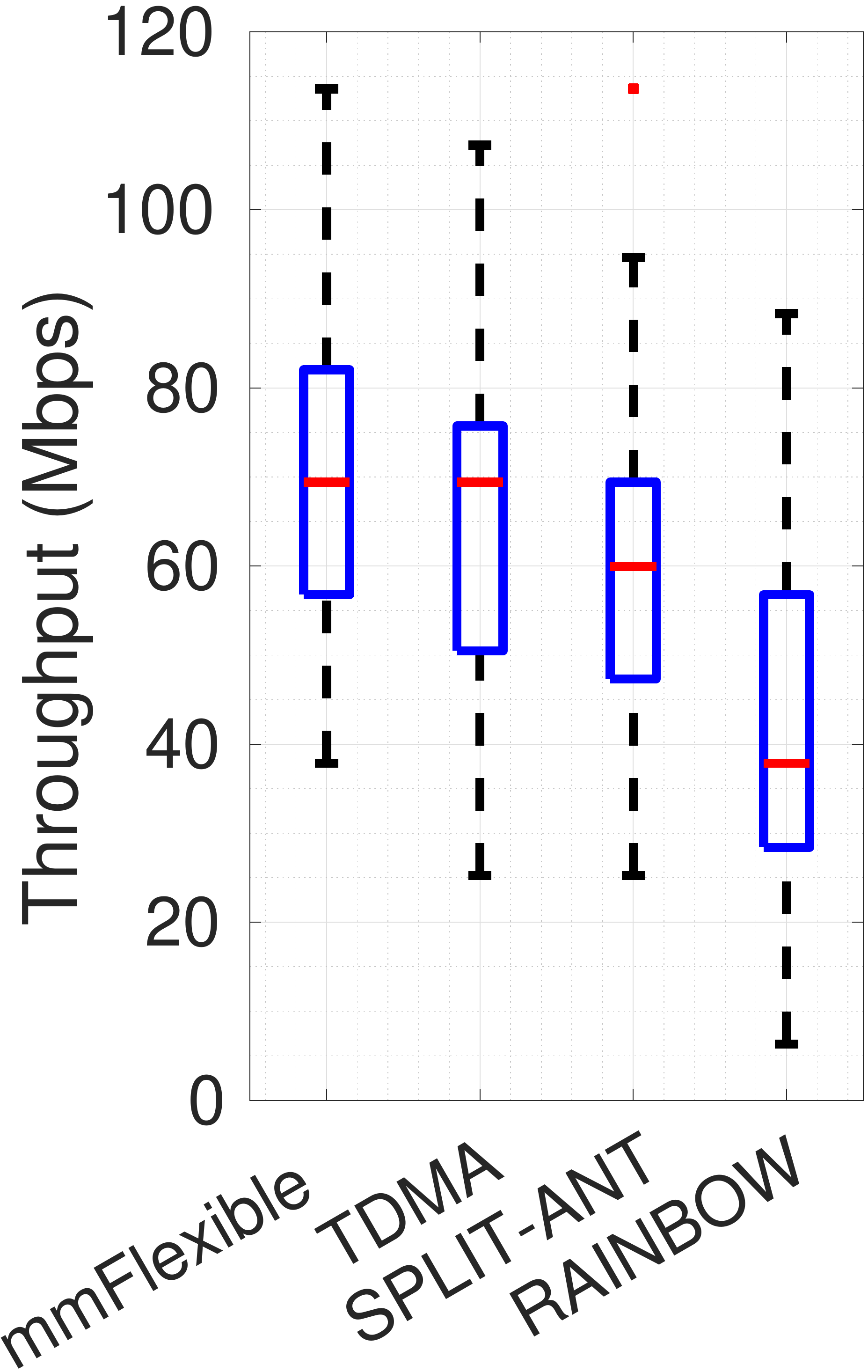}
    \label{fig:throughput_1ms}
  }\hfill
  \subfigure[Loss for $<$0.5ms latency]{
    \includegraphics[width=.17\textwidth]{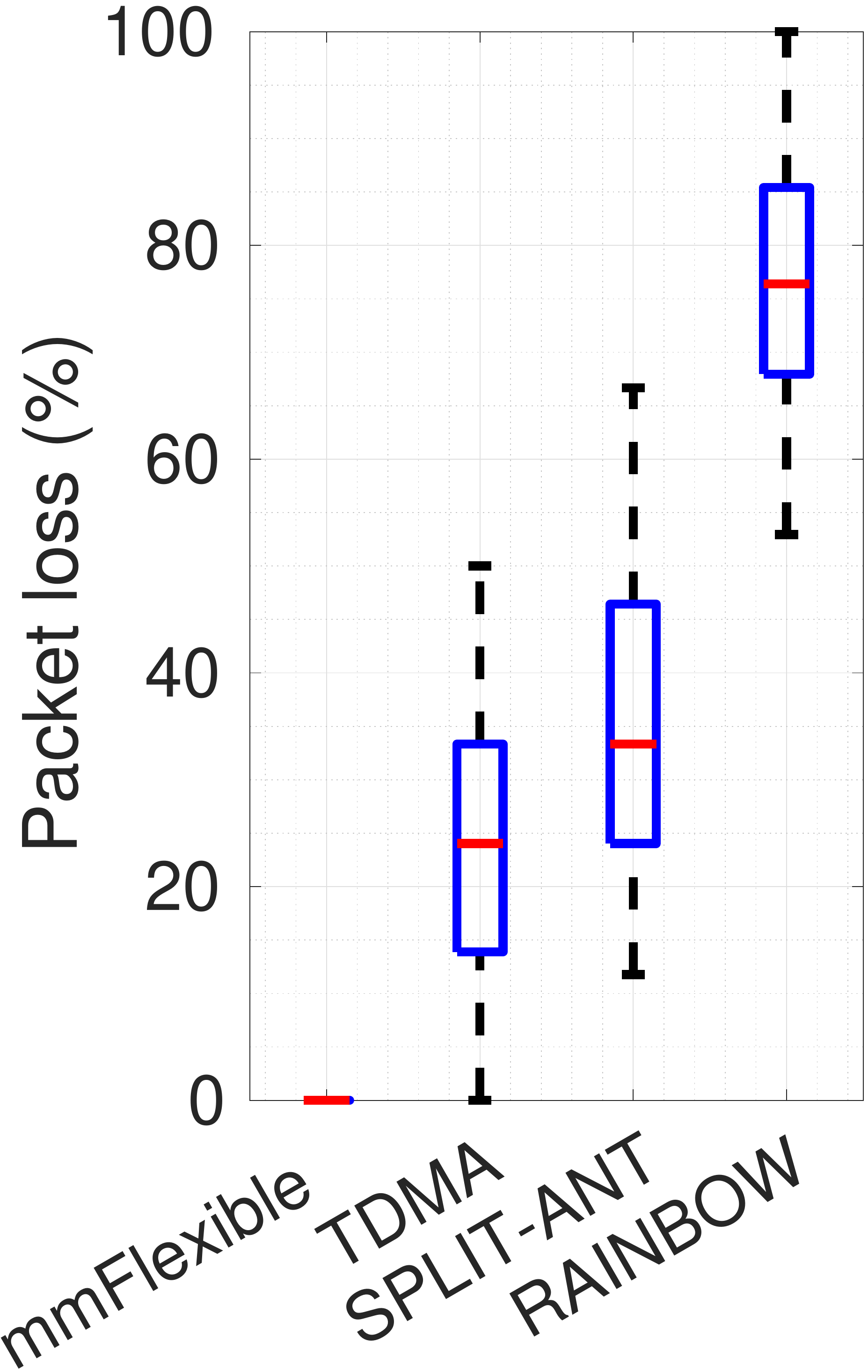}
    \label{fig:droppack_halfms}
  }\hfill
  \subfigure[Thput for $<$0.5ms latency]{
    \includegraphics[width=.17\textwidth]{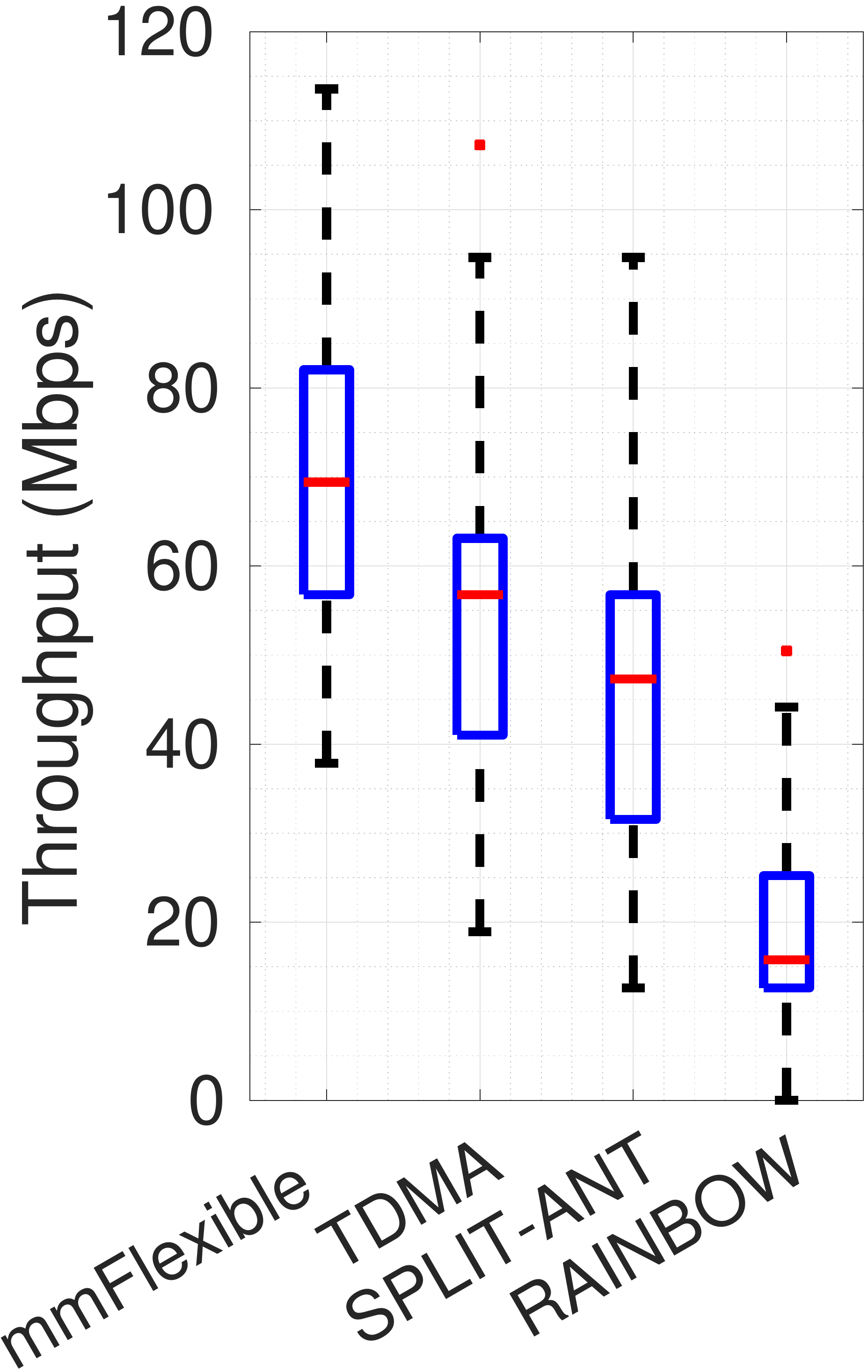}
    \label{fig:throughput_halfms}
  }\hfill
    \caption{End to End performance gain with DPA compared to baselines under different scenarios: (a) \name is able to meet strict latency constraints ($<$ 1 ms), but the baselines are not; (b) to (d) Percentage of packet loss and Throughput performance for latency constraint $<$ 1 ms and $<$ 0.5 ms respectively.
    % (c) Throughput (latency constraint $<$ 1 ms); (d) Percentage of packet loss (latency constraint $<$0.5 ms; (e) Throughput (latency constraint $<$ 0.5 ms).
    }
    % \label{fig:droppack_halfms}
    % \vspace{2pt}
\end{figure*}

\subsection{Implementation and emulation with 28 GHz dataset}

We implement an end-end system of \name with four major components, as illustrated in Fig.~\ref{fig:overview_fsda_scheduler}.
% , and test it on the representative example from Sec.~\ref{sec:motivation}: 
The components of the systems are as follows:
\begin{enumerate}[leftmargin=*]
    \item \noindent\textbf{\emph{Data traffic generator:}} We generate MAC layer packets with the throughput and latency constraints mentioned in the example in Sec.~\ref{sec:motivation} based on~\cite{mangiante2017vr,3gpp2020,qualcomm2017augmented}.  We test with two different latency constraints: 1 ms and 0.5 ms. We use the same traffic generator for our system and all baselines.
    % and implement mmWave front-end and scheduler separately.
    \item \noindent\textbf{\emph{Users angle estimation (initial access):}} We use the \textit{channel collected from mobile 28 GHz testbed~\cite{jain2020mmobile}}, using switched beamforming techniques~\cite{caudill2021real, palacios2018adaptive} for user's angle estimation.  We leverage the existing 5G NR SSB Beam scan~\cite{giordani2018tutorial} using an exhaustive search to estimate angles. The gNB scans 64 beams in the codebook, and each UE reports the best beam index that maximizes the received signal strength from which the gNB determines the user's angle.
    % From the beam index, the gNB will determine the user's angle. 
    % We assume 10 users in distinct directions.
    % Since our example consists of 10 users, we assume they are randomly deployed in 10 distinct angular directions.
    \item \noindent\textbf{\emph{Data Scheduling:}} We implement a Proportional Fair (PF) scheduler~\cite{github-nokia-wireless} to allocate spectrum resources to users on a per-TTI basis. The available $120^\circ$ field of view is mapped into 10 groups with each $12^\circ$ half-power beam width. The scheduler uses user grouping, demand generation information, and SNRs (mapped to CQIs) to determine which user group to support and how many subcarriers to allocate for each user. Throughput is then calculated as a function of the allocated resources and channel to each user. % and the users' channel.
    \item \noindent\textbf{\emph{FSDA Algorithm:}} 
     Our system's front-end uses \dpa, requiring delays and phases as input. The FSDA algorithm provides quantized delays and phases, which are applied to the \dpa to generate beams in desired directions and frequency bands. Array gain from the FSDA algorithm or respective baselines is then fed back to the scheduler for SNR computation.
\end{enumerate}
We evaluate the performance of the baselines and \name with a PF scheduler by emulating $2000$ TTIs of $125 \mu s$ each, for a total duration of 0.25 seconds, across all users. The system undergoes the above four-step procedure for each TTI.

% \noindent
% \textbf{Emulation with 28 GHz channel dataset: }
% We collect data with mMobile 28 GHz testbed~\cite{jain2020mmobile}, using switched beamforming techniques~\cite{caudill2021real, palacios2018adaptive}.

\noindent
\textbf{$\blacksquare$ Baselines:} Our paper compares the performance of \name with three baselines: TDMA, Split antennas~\cite{jain2019impact}, and Rainbow-link~\cite{li2022rainbow}. The TDMA approach has the scheduler assign one direction per TTI and beam in that direction over all frequencies. The Split antennas baseline has the scheduler assign one or multiple directions based on SNR \& demand requirements support the given directions in all subcarriers. The Rainbow-link~\cite{li2022rainbow} transmits in all directions regardless of the number of user directions, using only subcarriers in user directions and wasting the rest. 

\subsection{End-to-end Results}

1) \textbf{Latency}: Latency is the time for a packet to travel from source (gNB) to destination (UE) over the air. We evaluate the latency distribution across the baselines and present the results in Fig~\ref{fig:latency_only}. We see that \name has a median latency of 0.2 ms, while TDMA and Split-antenna baselines have a higher median latency of 0.32 ms and 0.26 ms respectively. Our implementation features equal offered throughput for every user direction, which is the best case scenario for Rainbow-link operation; despite this, the median latency for the Rainbow-link is 1.5 ms because of its inability to assign bandwidth to a user that is proportional to its demand. The worst-case latency of \name is well below 1 ms. Notably, all baselines have weighted right tail distributions of latency, making their worst case much worse than 1 ms. The inability of the baseline methods to honor the latency constraint leads to dropped packets and ultimately lower link throughput and reliability.

% \begin{equation}
%     throughput = \sum_{i\: \in \: subcarriers} \; log_2(1 + SNR_i)
%     \label{eq:throughput}
% \end{equation}

% As shown in (\ref{fig:throughput}), \name performs better than all the baselines. Around $80\%$ CDF, \name, split antennas and TDMA converges, these are scenarios where first user group requires all the sub carriers. Effectively there is only direction to serve for the desired user group. In those scenarios with first user group takes all subcarriers and no subcarriers left user group 2, it will act as serving single beam where \name and both baselines converges to same beam.

\begin{figure*}[!t]
    \centering
    \includegraphics[width=0.9\textwidth]{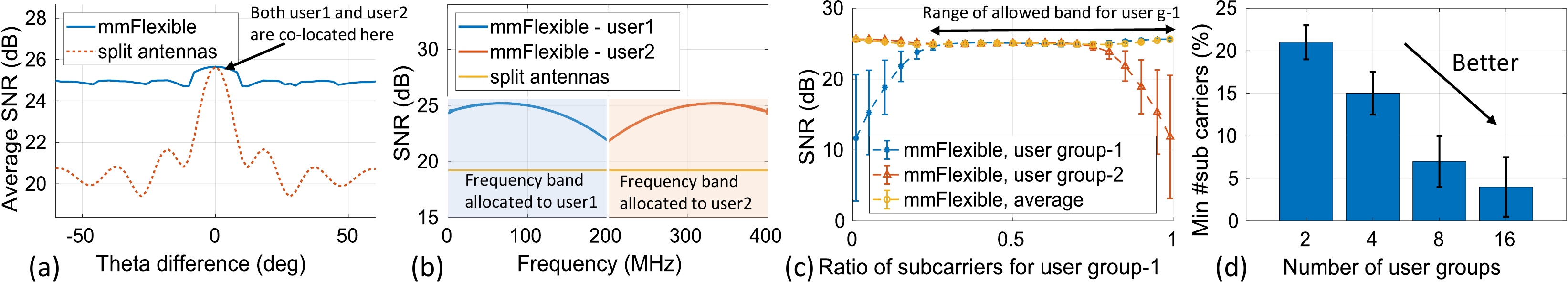}
    \caption{The benchmark performance gain of \name compared to baselines under different scenarios: (a) Angle separation between two users (b) (c) Subcarrier allocation between users (d) How much minimum frequency resources can be allocated without SNR degradation. }
    \label{fig:benchmark_figure}
\end{figure*}

% \begin{figure*} [!t]
%     \subfigure{{\includegraphics[width=0.23\textwidth]{figures/benchmark_angle_separation.pdf}\label{fig:angle_separation}}}
%     \subfigure{{\includegraphics[width=0.23\textwidth]{figures/benchmark_interference.pdf}\label{fig:interference}}}
%     \subfigure{{\includegraphics[width=0.23\textwidth]{figures/benchmark_scs_unequal_allocation.pdf}\label{fig:unequal_allocation}}}
%     \subfigure{{\includegraphics[width=0.23\textwidth]{figures/benchmark_convergence.pdf}\label{fig:convergence}}}
%     \caption{The benchmark performance gain of \name compared to baselines under different scenarios: (a) Angle separation between two users (b) (c) Subcarrier allocation between users (d) How much minimum frequency resources can be allocated without SNR degradation. }
%  \end{figure*}
 
 \begin{figure} [!t]
    \subfigure{{\includegraphics[width=0.23\textwidth]{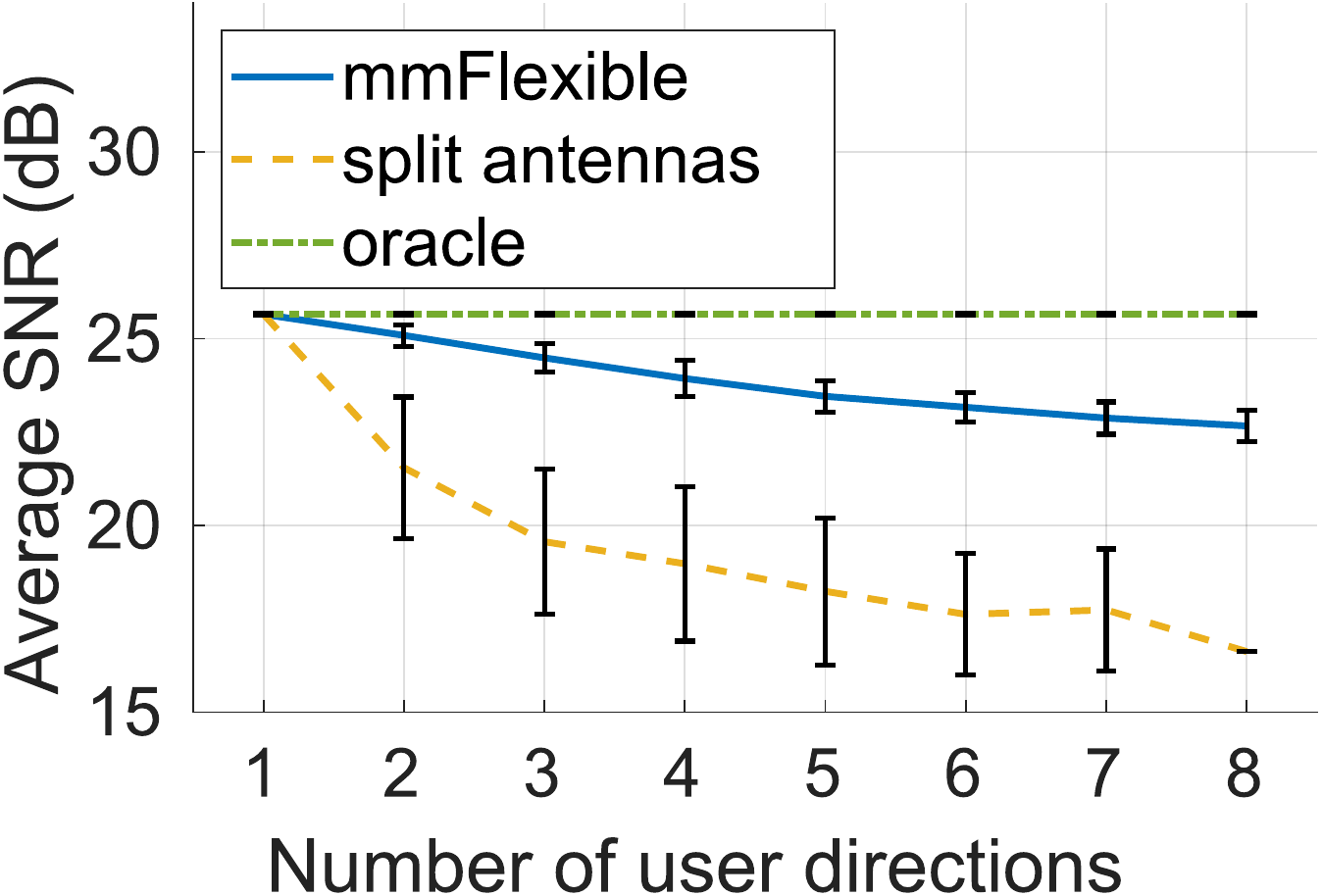}\label{fig:user_variations}}}
    \subfigure{{\includegraphics[width=0.23\textwidth]{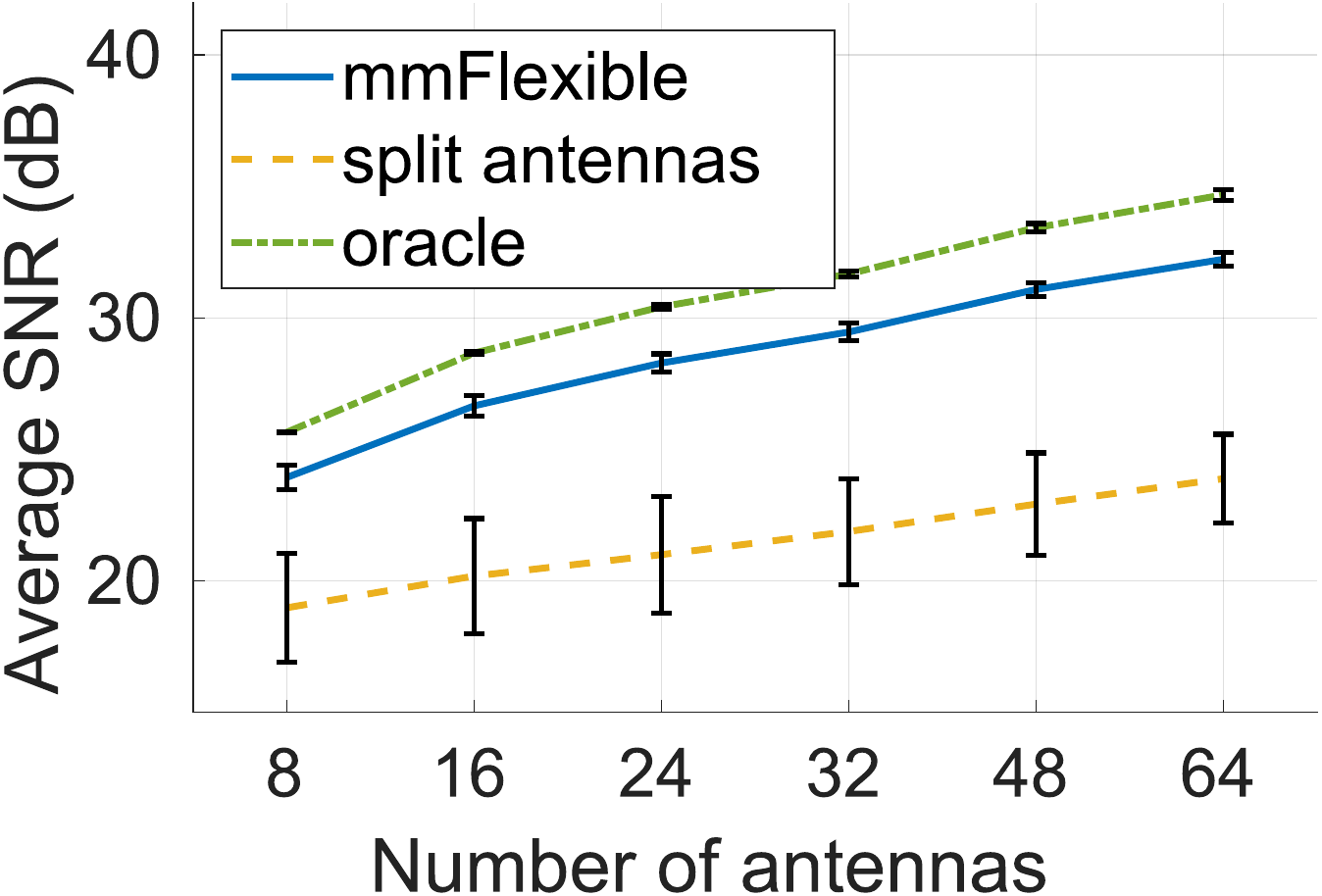}\label{fig:antenna_variations}}}
    % \subfigure{{\includegraphics[width=0.24\textwidth]{figures/benchmark_convergence.pdf}\label{fig:convergence}}}
    \caption{Scaling the \name's performance with (a) Increasing number of user directions and (2) Increasing number of antennas.}
    \vspace{4pt}
 \end{figure} 

 % \begin{figure*} [!t]
 %    \subfigure{{\includegraphics[width=0.23\textwidth]{figures/benchmark_angle_separation.pdf}\label{fig:angle_separation}}}
 %    \subfigure{{\includegraphics[width=0.23\textwidth]{figures/benchmark_scs_unequal_allocation.pdf}\label{fig:unequal_allocation}}}
 %    \subfigure{{\includegraphics[width=0.23\textwidth]{figures/benchmark_user_variations.pdf}\label{fig:user_variations}}}
 %    \subfigure{{\includegraphics[width=0.23\textwidth]
 %    {figures/benchmark_antenna_variations.pdf}\label{fig:antenna_variations}}}
    
 %    \caption{The benchmark performance gain of \name compared to baselines under different scenarios: (a) Angle separation between two users(b) Subcarrier allocation between users (c) Scaling with the number of user directions(d) Scaling with the number of antennas.}
 % \end{figure*}

2) \textbf{Packet loss}: If a packet's latency constraint is not met, then the packet is considered undelivered and lost. In Fig~\ref{fig:droppack_1ms} and Fig~\ref{fig:droppack_halfms}, we see that \name is able to function without any packet loss. However, due to their inability to honor the latency constraints, TDMA, Split-antenna, and Rainbow-link baselines result in a median packet loss of 24.0\%, 33.3\%, and 76.4\% respectively. \name's ability to serve multiple users in different directions enables it to optimally allocate resources without any power degradation and meet both throughput and latency constraints. TDMA is forced to serve one use direction at a time, resulting in a violation of latency constraints. Split-antenna baseline attempts to serve users in multiple directions simultaneously, but suffers from reduced throughput due to SNR degradation, resulting in high latency and packet loss compared to \name and TDMA. The Rainbow-link baseline allocates too few resources to each user and is the slowest and most unreliable in delivering packets.

3) \textbf{Per-user Throughput}:    
As shown in Fig.\ref{fig:throughput_1ms} and Fig.\ref{fig:throughput_halfms}, \name outperforms all three baselines TDMA, Split antennas, and Rainbow-link. TDMA can only support one user group direction out of all the presented user demand directions, reducing its efficiency. However, in scenarios with heavy throughput demand, it performs similarly to \name. The mean throughput of \name is 1.3$\times$ that of TDMA throughput in the case of 0.5 ms latency requirement. The Split-antenna approach serves users in different directions by splitting its antennas, resulting in reduced overall throughput. \name provides 1.5$\times$ more throughput than the split scenario. Rainbow-link performs better only in scenarios with low throughput demand and users in all directions, but performance degrades in all other cases. \name provides 3.9$\times$ more throughput than the Rainbow-link baseline.

\subsection{Benchmarks}
We benchmark various theoretical and systems aspects of \name and compare the results with a split-antenna baseline. We chose the split-antenna baseline as it is closest to \name in creating multiple simultaneous beams in different directions. We present our results by calculating SNR for various scenarios. We use a channel dataset (28 GHz testbed~\cite{jain2020mmobile}) and compute SNR per user for all subcarriers by evaluating over LOS channel model~\cite{3gpp_138_901} with no blockage and equal distance consideration for all users.

%\subsubsection{SNR Gain with \dpa}
%\name serves multiple users in a single time-slot by creating frequency-selective multi-beams, i.e., different frequency bands are transmitted in different directions. In this way, each user receives the signal only in their desired band and not in other bands. In other words, \dpa creates a different frequency filters in different directions to filter out-of-band signa total radiated power is maintained the same, the energy saved from other bands is used to enhance the signal strength in the desired band. In contrast, the split-antenna baseline severely suffers from reduced SNR as the signal is spread in both space and frequency. Figure \ref{fig:csi_plot_all} compares the SNR of four users in four different direction and allocated four different frequency bands. It shows that \name creates a frequency-filter to enhance the SNR in the desired band of the user while suppressing the SNR in the undesired band. In contrast, the baseline maintains a flat frequency response for all users.

1) \textbf{Effect of angular separation between users on SNR:}
We evaluate the performance of \name and a split-antenna approach in serving two users in different locations $(\theta_1, \theta_2)$ with a single RF chain and an equal number of subcarriers. Figure \ref{fig:benchmark_figure}(a) shows the mean SNR of the two users for different angular separations $(\theta_1 - \theta_2)$. We observe that \name performs optimally in all cases compared to the baseline. When the angular separation is close to zero, both users are in the same direction and can be served by a single beam. In these scenarios, both approaches converge to a single beam and perform similarly. The baseline becomes inefficient as the angular separation increases, while \name provides a stable response after $12\degree$ separation with a 3-5 dB higher gain than the baseline approach.

% 1) \textbf{Effect on SNR due to angular separation between users}:
% % \noindent
% % \textbf{Angle separation between users:}
% Here we show the effectiveness of \name in serving users who are separated in space. We evaluate \name and split antennas approach in supporting two users in different locations $(\theta_1, \theta_2)$ by a single RF chain and an equal number of sub-carriers. As shown in figure \ref{fig:benchmark_figure}(a), we compute mean SNR of two users for different angular separations $(\theta_1 - \theta_2)$. We observe that \name performs optimally in all the cases compared to the baseline. When theta difference is close to zero, both users are in the same direction and can be served by a single beam. In this case, both approaches converge to a single beam and perform similarly. The baseline becomes inefficient as the angle separation increases. On the other hand, \name provides almost a flat response after $12\degree$ separation and 3-5 dB higher gain than the baseline approach.

2) \textbf{Does \name have interference due to multiple user transmissions from different directions?}
% No, \name is designed to support frequency multiplexing of users in different directions without any interference. To evaluate this, we showed a simulation comparing \name with split antennas baseline for two users separated by $30\degree$ angle with equal bandwidth. As given in Fig.~ \ref{fig:benchmark_figure}(b) SNR over subcarriers is plotted for both approaches. For \name receiver, user-1 will transmit only in 0-200 MHz bandwidth, and user-2 will transmit only between 200-400 MHz bandwidth, thus leading to no interference because of multiple user transmissions simultaneously. Though edge subcarriers may suffer little degradation, there is no interference. Comparing average SNR over all subcarriers give \name a gain of ~5dB greater than split antennas, even at edge subcarriers has 3dB higher gain. 
No, \name is designed to support frequency multiplexing of users in different directions without interference. We evaluated this by comparing \name to a split-antenna baseline for two users separated by $30\degree$ angle and equal bandwidth, as shown in Fig.~\ref{fig:benchmark_figure}(b). Both approaches receive data from user-1 in the 0-200 MHz bandwidth and from user-2 in the remaining 200-400 MHz bandwidth, without sharing any subcarriers between users. This results in no interference from simultaneous multiple-user transmissions for both approaches. Additionally, \name provides a higher gain to users by radiating only in allocated subcarriers as illustrated in Fig.~\ref{fig:baseline_architecture_and_heatmap}(d), resulting in higher SNR than the baseline. The average SNR over all subcarriers shows that \name has a gain of ~5dB compared to the split-antenna approach, even a 3dB higher gain at edge subcarriers.

% 3) \textbf{SNR variations with subcarrier allocation between user/user groups}:
3) \textbf{Impact of subcarrier allocation on SNR}:
\name can reliably support both low-bandwidth \textit{Ultra-Reliable Low-Latency Communication (URLLC)} and high-bandwidth \textit{AR/VR} devices in different directions simultaneously. Here we evaluate the effect of subcarrier allocation on SNR for two user groups. Fig.\ref{fig:benchmark_figure}(c) shows the SNR achieved with subcarrier allocation to user group-1 from 1\% to 99\% (remaining subcarriers are allocated to user group-2). The average SNR over all users remains at 25 dB and does not depend on the proportion of subcarriers allocated to individual user groups. It is observed that the SNR of user group-1 converges to the average SNR when it is allocated more than 20\% of the total subcarriers. Therefore, to achieve the benefits of \name in a URLLC link, at least 20\% of the total subcarriers should be allocated to a user group. Additionally, as presented in Fig.\ref{fig:benchmark_figure}(d), as the number of user groups or user directions increases, the convergence point occurs at a lower percentage of subcarriers allocated. For instance, when there are 4 users in the system, then each user can be allocated with a minimum 15\% of subcarriers without degrading the SNR performance. This convergence point drops down to 4\% for 16 users which is favorable for URLLC applications which requires low-bandwidth and low latency.

% For 4 users, the convergence point is at 15\%, and for 8 and 16 users, the convergence points are at 7.5\% and 4\% respectively. The error bars in Fig.\ref{fig:benchmark_figure}(c) and Fig.\ref{fig:benchmark_figure}(d) indicate the deviation of the convergence point for various angle separations between users.

% \subsubsection{Frequency bands scheduling of users}
% % \noindent
% % \textbf{Frequency bands scheduling of users:}
% We evaluated various scenarios to understand the need for frequency scheduling of users. For instance, we describe with an example of four users with equal subcarriers allocation to each user. From the Figure~\ref{fig:frequency_scheduling}, it is observed that users with edge frequency slots are getting higher throughput compared to the users with frequency slots sand-witched in between the end-slots, which are having lower gain. The reason behind such distribution of signal strength across frequency bands is that the end bands observes one-sided filter while the middle bands observe two-sided filter. Two-sided filter provide low gain because of the narrow filter response from both sides, while one-sided filter has wide response on one side and narrow on the other side. Creating a narrow response requires a large number of antennas and so with a fixed number of antennas, narrower response leads to lower gain. Nonetheless, the gains are still better than the baseline. In fact, we exploit this phenomenon to selectively allocate high-gain end-slots to throughput intensive applications and other slots for less-demanding applications.

\textit{\textbf{Theoretical Baseline (Oracle):}} 
% In addition to the split antenna baseline, we have also included a theoretical oracle for figures \ref{fig:user_variations} and \ref{fig:antenna_variations}.
Oracle is a theoretical entity that can perfectly transmit only in desired subcarriers without any degradation in power at edge subcarriers (as shown in Fig.~\ref{fig:dpa_desired_4beam} desired frequency-space response). 

% 4) \textbf{SNR gain with the number of user directions}:
4) \textbf{Impact of the number of user directions}:
The performance of the \name improves as the number of supported user directions increases. We tested the system using a base station antenna array with 8 transmit/receive antennas, varying the number of user directions from 1 to 8. Fig.~\ref{fig:user_variations} illustrates that the relative gain (difference between the average SNR of \name and split antennas) increases with an increase in user directions. As the number of user directions increases, the split antenna approach divides the antennas per beam, resulting in reduced gain. Conversely, \name transmits power only in desired frequency bands and angular directions resulting in a higher gain. The Oracle creates a digital frequency filter at each antenna, resulting in a perfect frequency-space slicing which ensures that the average SNR remains constant regardless of whether it serves in one direction or multiple directions simultaneously. In contrast, \name has one delay per antenna (for hardware feasibility), which makes it difficult to create an ideal frequency-space slicing, leading to power degradation at the edge subcarriers. Despite this, even after eight splits, the degradation is less than 2.5 dB with the Oracle, and the gain is more than 6 dB higher compared to the split antenna baseline. Error bar in Fig.~\ref{fig:user_variations} indicates average SNR variations with users in different angle separations. 

% 5) \textbf{SNR gain variations with the number of antennas:}
5) \textbf{Impact of the number of antennas}:
% \noindent
% \textbf{Performance scale with the number of transmit antennas}
% mmWave communications paved the way to accommodate more antennas within small lengths. Future base stations may vary from 8 antennas to 64 antennas or more.
We show that \name performs better with the increase in the number of antennas. We evaluated this hypothesis by serving four users and varying antennas from 8 to 64. Fig.~\ref{fig:antenna_variations} shows gain variations from 8 antennas to 64 antennas for 4 users; it is clearly evident that \name outperforms with the increase in antennas over the split antennas baseline. The error bar in the figure shows the variations in the average SNR when serving four users at different angle separations. Similar to the Oracle, \name's performance remains constant even as the number of antennas increases because the number of frequency splits is determined solely by the number of user directions, which are the same in all cases. 

%  We tested this hypothesis by serving four users and varying the number of antennas from 8 to 64. Fig.~\ref{fig:antenna_variations} illustrates the gain variations from 8 antennas to 64 antennas for 4 users and it is clearly evident that \name outperforms the split antennas baseline as the number of antennas increases. The error bar in the figure illustrates the variations in the average SNR when serving four users at different angle separations. Similar to the Oracle, \name's performance remains constant as the number of antennas increases because the number of frequency splits is determined solely by the number of user directions, which are the same in all cases.

% \subsubsection{Hardware feasibility analysis}
% \todo {Need to show feasibility with existing Phased Arrays (PAs).}
% Here we show feasibility with two phased arrays supporting two different directions. Results from these evaluations are compared with proposed \name and all other baselines. 
% !TEX root = main.tex

\section{Related Work} \label{sec:related}

% \name is related to exhaustive work on mmwave and THz communications, however, none of the literature has accomplished energy-efficient frequency slicing without comprising range or throughput performance. \name for the first time designs a system that can achieve flexible frequency multiplexing for 5G mmWave, all while radiating each of these resources via pencil beams in desired arbitrary directions. Furthermore, \name ensures to radiate signal only on the frequencies and only in the directions desired, thereby ensuring similar range, throughput, and link reliability. 
\name builds upon previous work in mmWave and THz communications, but sets itself apart by introducing a system that can perform flexible directional-frequency multiplexing, while maintaining energy efficiency and high performance in terms of range, throughput, and link reliability. To the best of our knowledge, no existing literature has achieved this level of frequency slicing without compromising on performance. \name's unique approach enables efficient use of the entire frequency band, reducing spectrum wastage and providing low latency network access.

\noindent
\textbf{$\blacksquare$ Split antenna phased array:}
% Split antenna phased array is common way to generate multiple beams from a single phased array, which reduces the signal strength in each direction as it radiates on entire BW on every direction instead of being selective like \name. It 
% Split-antenna array is used to improve mmWave link reliability~\cite{jain2021two, aykin2019multi}, joint communication and sensing~\cite{zhang2018multibeam}, multi-beam-based fast initial-access~\cite{hassanieh2018fast} or mobility tracking~\cite{zhu2018high}. Beam splitting can support two or more users in distinct directions, but it always comes with the cost of low throughput for each user. Another possibility is to use wide beams to serve multiple nearby users in a single time TTI~\cite{ismayilov2018adaptive}. But, the wider the beam, the lesser the antenna gain, leading to lower data rates. In contrast, \name creates a split-beam mechanism with frequency selectivity without compromising on directivity, signal strength, and throughput.
Traditional phased array beamforming does not support flexible directional-frequency multiplexing because of a single narrow pencil beam. In the past, split-antenna arrays have been used to create concurrent multi-beams across multiple directions~\cite{jain2021two, aykin2019multi, zhang2018multibeam, hassanieh2018fast, zhu2018high}. However, this approach often results in lower beamforming gain and throughput for each user. The array gain reduces proportionally to the number of beams as the total available power is distributed along multiple directions and across the entire bandwidth. In contrast, \name uses a unique split-beam mechanism with frequency selectivity that radiates only in the desired frequency band, preserving high directivity, signal strength, and throughput.

% the desired frequency subcarriers for each user have high gain even when the users are located in different directions. 

% Delay phased-arrays have low hardware complexity than multi-RF architectures (e.g. hybrid or fully-digital) 

% Hybrid beamformer can even allocate different subcarriers per-direction with no interference or same subcarriers per-direction with large angular seperation to avoid interference in that frequency band.

% \cite{mounika2021downlink} proposes 3 dimensional scheduler across frequency, time and space using a fully digital MIMO system similar to 4G. It's not applicable to 5G mmWave system which has phased arrays to reduce power, cost and hardware complexity that comes from high-frequency circuits design which can create only one beam at a time, thus can serve only one direction for all frequency resources at a time.

% Fully digital beamforming "require digitization at every element and increase IO between mm-wave front-end and mixed-signal backend ICs \cite{garg202028}". "Although DBF improves capacity compared to single-beam arrays, the absence of spatial filtering in a DBF array imposes higher dynamic range requirements" \cite{van2012full}.

\noindent
\textbf{$\blacksquare$ True-time delay array architecture:}
Previous work on True-time delay arrays (TTD) has primarily focused on beam steering for ultra-wideband signals~\cite{rotman2016true,garakoui2015compact} and, more recently, single-shot beam training~\cite{yan2019wideband,wadaskar20213d,boljanovic2021fast}, compressive channel estimation~\cite{boljanovic2021compressive}, wideband tracking~\cite{tan2021wideband} and THz communication~\cite{tan2019delay}. However, none of these works address the problem of flexible low-latency multi-user communication. Rainbow-link~\cite{li2022rainbow} uses TTD arrays for multi-user communication, but it is limited to fixed low-throughput IoT applications (limited to 7.8 MHz per user~\cite{li2022rainbow}) and cannot flexibly allocate a large number of subcarriers to a single direction for broadband users.
\name addresses this limitation by using variable delay elements to create arbitrary frequency slicing and the ability to radiate those frequencies in any desired direction. \name's beamforming is orthogonal to previous work in this area, but it can leverage the fast beam training capability of TTD arrays.

% Recently, new derivatives of TTD arrays have been proposed for creating frequency-dependent beams similar to that of DPA~\cite{ratnam2022joint}. But,~\cite{ratnam2022joint} is different from \name in many ways. Firstly, they do not provide any insight into the hardware architecture and the range of delay values required to support their beamforming patterns. In contrast, we are the first to show that our multi-beam patterns can be realized with shorter delay lines which makes it practical and realizable with current technology. Moreover, the maximum delay value is independent of the number of antennas which makes DPA scalable to large antenna arrays. Secondly, their architecture is not fully analog as it requires digital processing over each frequency subcarrier. On the other hand, DPA programming is fully-analog and does not require any digital programming over the frequency subcarriers. Finally, they have high run-time complexity for estimating the delay values in real time. In contrast, we provide closed-form mathematical expression for delays which can be used as plug-and-play with a constant run-time complexity.

Recently, a new approach called Joint Phase-Time Array (JPTA) has been proposed for creating frequency-dependent beams~\cite{ratnam2022joint} that is similar to \name in its goals. However, \name and JPTA have several key differences in their implementation and performance. Firstly, JPTA does not provide any details on the specific hardware architecture and the range of delay values required to support their beamforming patterns. \name is the first work that provides a clear explanation of the DPA hardware architecture and the range of delay values required to support our beamforming patterns. We have shown that our DPA patterns can be realized with shorter delays and independent of the number of antennas making it practical and realizable with current technology, and scalable to large antenna arrays. Secondly, JPTA requires digital processing over each frequency subcarrier, while DPA programming is fully analog and does not require any digital programming over the frequency subcarriers. Finally, JPTA has a high run-time complexity for estimating the delay values which makes it inadequate for real-time operation, while \name provides a closed-form mathematical expression for delays that can be used as plug-and-play with a constant run-time complexity.

% We are implementing a circuit for \name based on the insights provided in this paper and evaluating many new applications of \name as future work.
% TTD long ~\cite{yan2019wideband, boljanovic2020design,wadaskar20213d,boljanovic2020true,boljanovic2021fast},

\noindent
\textbf{$\blacksquare$ Other front-end mmWave architectures:}
% Garg et al. ~\cite{garg202028} proposed a new mmWave receiver architecture for frequency multiplexing. They create a network of mixers at each antenna to receive different frequency components from different directions. However, their hardware is fixed and not scalable to a large number of users. In fact, a different hardware is required for a different number of users. \name is complimentary to this work as \dpa is flexible and programmable.
In~\cite{garg202028}, a new mmWave receiver architecture for frequency multiplexing was proposed, which utilizes a network of mixers at each antenna to receive different frequency components from different directions. However, this approach has a fixed hardware structure that is not scalable to support a large number of users, requiring different hardware for different numbers of users. In contrast, \name's delay-phased array (\dpa) is flexible and programmable, providing a scalable solution for supporting a large number of users.
\section{Discussion and Future Work}\label{sec:discussion}

% \review{R2: The key claim here is that this is a energy efficient mmWave front-end architecture. However, it does not look like the authors have built this hardware and hence this central claim of energy-efficiency has not been validated.} \todo{Discuss energy efficiency claims. - Help from Prof. Subhanshu.}

% \review{R3: On the one hand, I sympathize with the authors since building a mmWave array from scratch is probably a multi-year project. On the other hand, it is really unclear how feasible variable delays are going to be especially at such high mmWave frequencies. It is not clear how well they will work, how much loss they will introduce, how accurate their delays are, and most importantly whether the results in this paper will continue to hold or not.} \todo{Discuss practical aspects of delay-phased array, including the limitations.}

% \textcolor{blue}{}
% Creating a delay element at mmWave frequency is challenging as it requires a large area for laying inductive transmission lines and suffers from non-linearity and noise~\cite{cho2018true, hu20151}
We discuss potential future work ideas:
% such as building circuit for \dpa, extending it to multi-RF scenarios and enabling new applications for directional mmWaves systems.

\noindent
\textbf{$\blacksquare$ Circuit for \dpa:} Implementing circuit delays in mmWave frequencies is challenging due to non-linearity, bandwidth and matching constraints~\cite{cho2018true, hu20151}. In contrast, delay elements can be more accurately implemented at intermediate frequencies (IF) (sub-6 GHz) using techniques such as voltage-time converters~\cite{ghaderi2020four} and switched-capacitor arrays~\cite{ghaderi2019integrated}. Recent work~\cite{ghaderi2019integrated} has shown that efficient mixers, phase-shifters, and IF true-time-delays can be used to make a \dpa that meets the requirements of \name.
% The IF delay circuit increases system complexity and requires N mixers for N antenna elements for upconversion compared to having only one mixer in the former case. Our delay-phased arrays require a low delay range and can be created by integrating efficient mixers, phase-shifters, and true-time-delay elements~\cite{ghaderi2019integrated}.

\noindent
\textbf{$\blacksquare$ Single RF vs Multi-RF systems:} \name works with a single-RF chain and radiating each frequency component in different directions, while past work on the single-RF system stream all frequencies in one direction. Multi-RF systems (Hybrid arrays)~\cite{shen2020mobility,zhang2019joint} offer freedom to create multi-stream to multi-directions, but each stream carries the entire bandwidth. Thus they also cannot perform flexible directional-frequency multiplexing. Our work can be extended to the multi-RF system where each RF is connected to a \dpa that creates large sectors where each \dpa can serve in a sector with our flexible directional-frequency multiplexing technology.

% Typically, a single phased array cannot cover the $360^\circ$ space due to array geometry constraints (typical field of view of phased arrays is $120^\circ$). Thus, the cellular base station requires multiple antenna arrays and multiple RF-chains to cover the entire $360^\circ$ space. Rather than using multiple phase-only arrays, we propose to use multiple delay-phased arrays so that the proposed flexible frequency multiplexing can be done for users in the entire $360^\circ$ space.\\
% \textbf{Improving sustainablilty with power-efficient next generation networks:}
\noindent
\textbf{$\blacksquare$ Applications beyond flexible frequency multiplexing:} Delay-phased arrays have the potential to enable a plethora of applications in communication and sensing beyond flexible directional-frequency multiplexing. For instance, the ability to create arbitrary and controllable frequency-space beams can help faster localization and tracking of multiple targets. Delay-phased arrays also enable simultaneous communication and sensing paradigms where some frequency bands are used for communication while other bands can be used for sensing. All these applications can be enabled with simple software or firmware updates on the same underlying hardware. We leave these applications for our future work.

\section{Acknowledgements}\label{sec:acks}
We are grateful to the anonymous reviewers their valuable feedback, as well as to the WCSNG group at UC San Diego for their input. The research was supported by NSF \#2211805.

% \newpage
\appendix

\subsection{Generalized closed-form delay and phase expression}\label{app:proof}

In this section, we propose a closed-form mathematical expression for the delays and phases that can be used to generate a desired multi-beam pattern with DPA. This approach eliminates the need for complex computation and allows for real-time operation with run-time complexity reduced to $\order(1)$. Additionally, this mathematical expression also provides insight into the maximum range of delay values that the DPA hardware must support. We show that this range is significantly lower than that required by traditional true-time delay arrays and is independent of the number of antennas, making it scalable for large arrays. We start by deriving the expression for a simple two-beam case and then generalize it for an arbitrary number of beams, beam directions, and beam-bandwidth.

% Delay phased array (DPA) is a revolutionary array architecture because of its ability to generate a desired multi-beam pattern with any number of beams, beam directions, and beam-bandwidth. An important aspect of DPA is to accurately estimate the delay and phase values for each antenna. To achieve this, we have developed a heuristic algorithm called FSDA (Frequency-Space to Delay Antenna) that uses a 2D FFT transform to calculate the necessary delays and phases. While FSDA is effective, it has a high computational complexity that increases with the number of frequencies and antennas used.

% To address this issue, we have developed a closed-form mathematical expression for the delays and phases that can be used as a plug-and-play approach for real-time operation without any additional computational complexity. Using this approach, the calculation complexity is reduced to $\order(1)$. Additionally, this mathematical expression provides insight into the maximum range of delay values that the DPA hardware must support to generate the desired beam patterns. This range of delay values is shown to be significantly less than that required by traditional TTD arrays and is independent of the number of antennas, making it scalable for large arrays.

% We begin by demonstrating the method for a simple two-beam case and then present a generalized case for an arbitrary number of beams, beam directions, and beam-bandwidth.

\begin{figure*}[!t]
  \begin{minipage}[b]{0.18\linewidth}
    \centering
    \includegraphics[width=\linewidth]{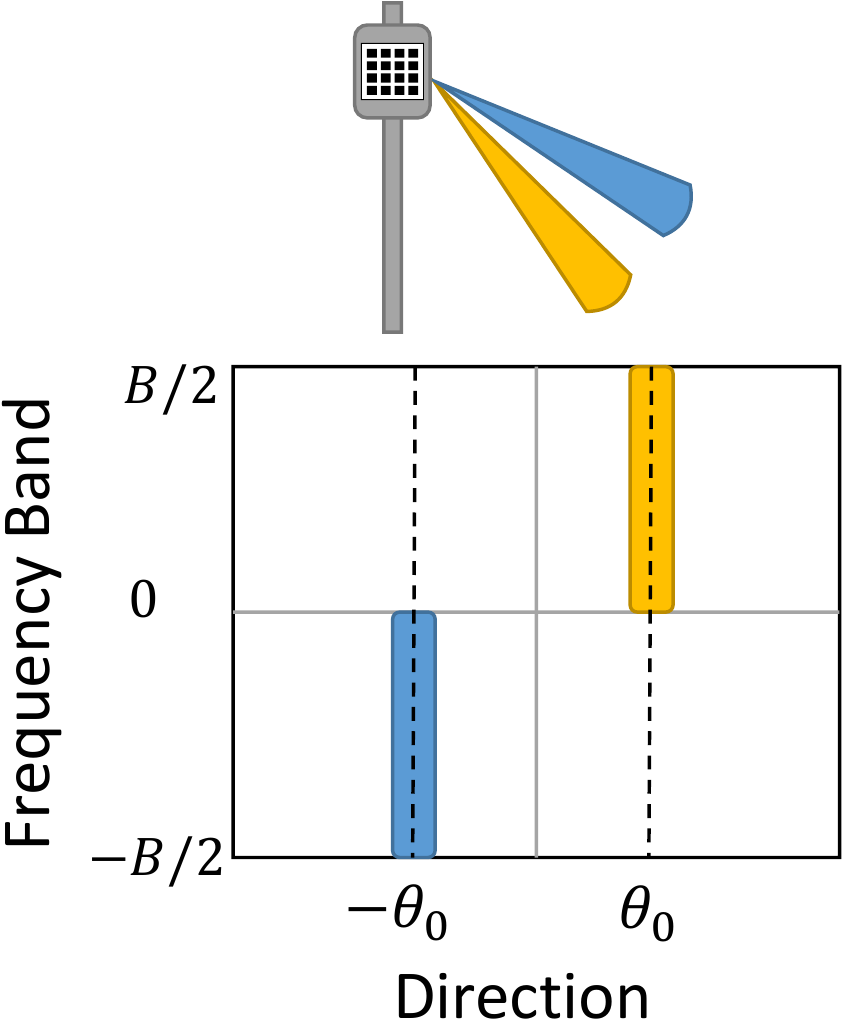}
    \caption{Desired frequency-space beam response for 2 users at directions $-\theta_0$ and $\theta_0$.}
    \label{fig:proof_desired_response}
  \end{minipage}
  \hspace{0.01\linewidth}
  \begin{minipage}[b]{0.79\linewidth}
    \centering
    \includegraphics[width=\linewidth]{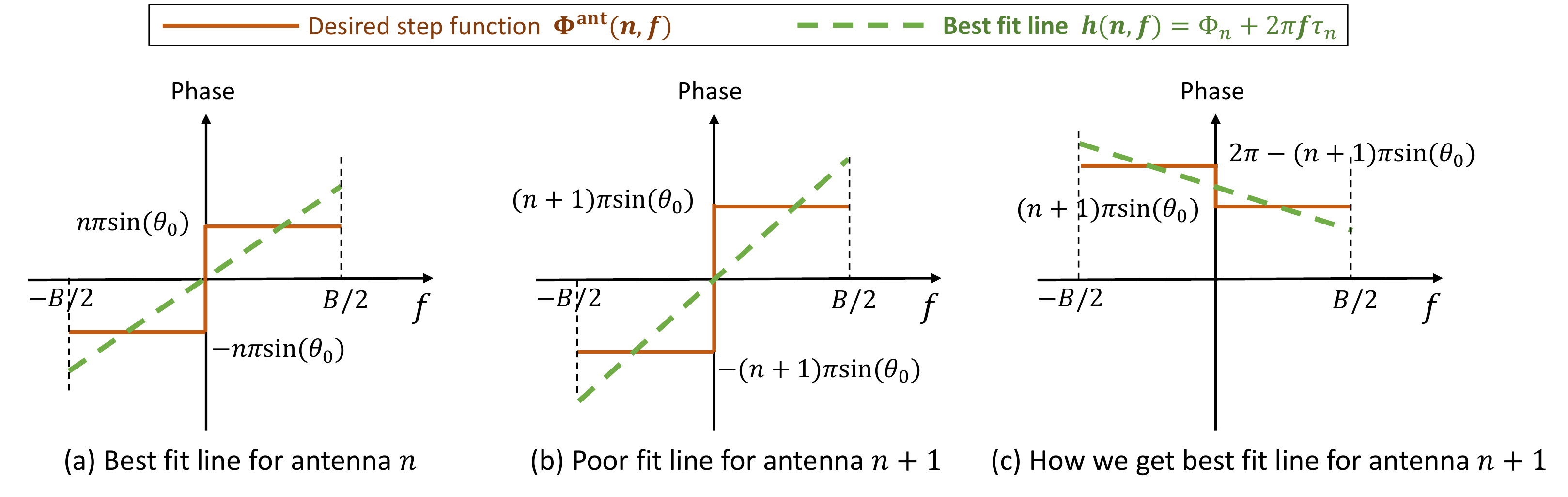}
    \caption{
    Sketch of proof: We show the phase response for each antenna is a step function with a variable step size that depends on the antenna index, $n$. (a) The objective is to fit a line to this step function in such a way that the error in the line fitting is minimized. (b) we show that as the antenna index increases to say, $n+1$, the step size of the step function also increases linearly with the antenna index. An increase in step-size leads to higher errors in line fitting. (c) 
    We address this issue by reducing the step size by $2\pi$, which will not change the desired phase response but will help to reduce the error.}
    \label{fig:proof_best_line_fit}
  \end{minipage}
\end{figure*}

\subsubsection{Simple two-beam case}
Here we derive what set of delays and phases per antenna would give us the beamforming gain pattern with the desired beam-bandwidth and beam direction for the two-beam case. We consider a simple case of two beams with equal beam-bandwidth of $B/2$ each, where $B$ is the total system bandwidth. We assume the two beams are directed along $(-\theta_0, \theta_0)$ respectively, as shown in Figure \ref{fig:proof_desired_response}. We will later discuss a general case with arbitrary beam-bandwidth and beam directions.

% \noindent
% \textbf{$\blacksquare$ Theorem (2-beam case):} 
\begin{theorem}\label{th:2beamcase}
\textbf{(2-beam case)}
The closed-form expression for the set of delays $\tau_n$ and phases $\Phi_n$ for each antenna $n$ ($n=0,1,\ldots, N-1$) that would generate a given two-beam response with equal beam-bandwidth $B/2$ and angles $\pm \theta_0$ respectively is as follow:
\begin{equation} \label{eq:tau_closed_proof}
    \tau_n = \left(\frac{3}{2B}n\sin(\theta_0) +\frac{3}{4B}\right)\;\;\;\text{mod }\frac{3}{2B} 
\end{equation}
% Alternatively, another expression for delay is 
% \begin{equation} \label{eq:tau_closed_proof}
%     \tau_n = =\frac{3}{2B}(n\sin(\theta_0)-\text{round}(n\sin(\theta_0))) 
% \end{equation}

% Correct form $\tau_n = (\frac{3}{2B} n\sin(\theta_0) +\frac{3}{4B})\;\;\text{mod }\frac{3}{2B}  = \frac{3}{2B}( n\sin(\theta_0) +1/2) \text{mod }\frac{3}{2B}$
% And phase is
% \begin{equation}\label{eq:phase_closed_proof}
%     \Phi_n = 
%     \begin{cases}
%         0 & \text{when}\;\;\cos(n\pi\sin(\theta_0))> 0 \\
%         \pi & \text{when}\;\;\cos(n\pi\sin(\theta_0))<0
%     \end{cases}
% \end{equation}
% Alternate expression for phase is:
\begin{equation}\label{eq:phase_closed_proof}
    \Phi_n = \text{round}(n\sin(\theta_0))\pi \;\;\;\text{mod }2\pi
\end{equation}
% $n\pi\sin(\theta_0) = a>0$ means $a\in (-\pi/2, \pi/2)+k2\pi$, means $n\pi\sin(\theta_0)+\pi/2 \in (0, \pi)+k2\pi$, means $n\sin(\theta_0) +1/2  mod 2 \in (0,1) $, 
% Case 2: $n\sin(\theta_0) +1/2  mod 2 \in (1,2) $

\end{theorem}
We emphasize key takeaways from these expressions before jumping into their proof. Note that the delay values are bounded within a range of $\frac{3}{2B}$ independent of the number of antennas. Within this range, the delay values will monotonically increase or decrease with antenna index $n$, but for large $n$, the delay wraps around with this range factor. This allows for a smaller range of delay values in the DPA hardware independent of the number of antennas. 

We now provide proof for the delay and phase expression and throw more insights into the bounded nature of delay and phase values. 

% Second, within this range, the delay values are linearly related to antenna index $n$, that means the delay values would increase or decrease monotonously with $n$. 
% \todo{So what?}

% \begin{figure}[t!]
%     \centering
%     \includegraphics[width=0.46\textwidth]{figures/dpa_proof.pdf}
%     \caption{Sketch of proof of estimating closed-form delay and phase values for a given two-beam response at $\pm\theta_0$. }
%     \label{fig:dpa_proof}
% \end{figure}
% \noindent

% \noindent
% \textbf{$\blacksquare$ Derivation for delay and phase values:} 
\begin{proof}[Proof for Theorem~\ref{th:2beamcase}] \textbf{(2-beam case)}
We provide a derivation along with high-level intuition on obtaining a closed-form expression for the delay and phase values. We will first formulate the objective function as an NP-hard problem and then provide an alternate optimization strategy as a set of linear equations that best approximates the solution. We first define the beamforming gain function as a function of frequency $f$ and direction $\theta$ and then look for maximizing this function at the desired frequency-direction pairs. The beamforming gain is given by:
\begin{equation}
    G(f,\theta) = \sum_{n=0}^{N-1} e^{j\Phi_n+j2\pi f\tau_n} e^{-jn\pi\sin(\theta)}
\end{equation}

Now, the objective is to maximize the beamforming gain $ \left\|G(f,\theta)\right\|^2$ at desired beam-bandwidth at given directions $\pm\theta_0$ as follows:
% , i.e., maximize $G(f,\theta)$, given the constraints as follow: For the first beam at $\theta = -\theta_0$, beam-bandwidth is $[\frac{-B}{2},0)$, and for the second beam at $\theta = \theta_0$, the beam-bandwidth is $[0,\frac{B}{2}]$. 
\begin{equation}
\begin{split}
    &\max_{\tau_n,\Phi_n} \quad\left\|G(f,\theta)\right\|^2\\
    & s.t. \quad  \theta = \begin{cases}
                -\theta_0 & f\in [\frac{-B}{2},0)\\
                \theta_0 & f \in [0,\frac{B}{2}]
            \end{cases}
\end{split}
\end{equation}
Notice that the direction and frequency have a non-linear relationship in the constraints. Specifically, the direction $\theta$ is a non-linear step function of frequency $f$, with a jump at frequency $f=0$. It jumps from the value $-\theta_0$ to $+\theta_0$ at this frequency. Because of this non-linear dependence of direction with frequency, the underlying optimization problem is NP-hard and cannot be solved optimally. We provide insight into the problem from a different angle and formulate a near-optimal optimization that can be solved to a closed-form expression.

To simplify the above optimization, We define two new functions $h(n,f)$ and $\Phi^\ant(n,f)$ below to simplify the non-linear constraint and re-write the optimization problem as follows:
\begin{equation}\label{eq:constraint}
\begin{split}
    \max_{\tau_n,\Phi_n} &\left\|\sum_{n=0}^{N-1} e^{j h(n,f)} e^{-j\Phi^\ant(n,f)}\right\|^2\\
     s.t. \quad  &h(n,f)=\Phi_n+2\pi f \tau_n\\
    \text{and}\quad &\Phi^{\text{ant}}(n,f) = \begin{cases}
        -n\pi\sin(\theta_0) & f\in[\frac{-B}{2},0]\\
        n\pi\sin(\theta_0) & f\in(0,\frac{B}{2}]\\
    \end{cases}
\end{split}
\end{equation}
% We visualize in Figure~\ref{fig:proof_best_line_fit} that $\Phi^{\text{ant}}(n,f)$ is a step function with frequency, while $h(n,f)$ is a linear function in frequency.
where $h(n,f)$ is a function of variable phase $\Phi_n$ and delay $\tau_n$ at antenna $n$ and the function $\Phi^{\text{ant}}(n,f)$ represents the constraints from the desired frequency-direction response. Let's see how these functions help us to simplify the optimization problem. Specifically, we apply triangle inequality to find an upper bound on the optimization variable and then maximize this upper bound.  Triangle inequality states that the `norm-of-sum is upper bounded by sum-of-norms', which we can apply to our optimization as:
\begin{equation}
    \left\|\sum_{n=0}^{N-1} e^{j h(n,f)} e^{-j\Phi^\ant(n,f)}\right\| \le \sum_{n=0}^{N-1}\left\| e^{j h(n,f)} e^{-j\Phi^\ant(n,f)}\right\| = N
\end{equation}

So, the expression is maximized if each term in the sum is unity, i.e., $e^{j h(n,f)} e^{-j\Phi^\ant(n,f)}=1$, or, in other words, the two exponential terms are equal, i.e., $h(n,f) = \Phi^\ant(n,f)$ for each antenna and for each frequency. It is impossible to achieve this solution for all frequencies because the two functions $h$ and $\Phi^\ant$ vary differently with frequency; $h$ is linear, while $\Phi_\ant$ is a step function as shown in Figure~\ref{fig:proof_best_line_fit}(a). The only case when the optimal results are possible is when the step size of the step function $\Phi^\ant$ is zero, i.e., the two beams align to the same angle. In this case, the step function is reduced to a line, and optimal $h$ can be obtained. However, this case only produces a single beam without any dependence on frequency. A natural question is how we can obtain general frequency-dependent multi-beams. We propose an optimization framework that can help to find a closed-form expression for delays and phases. Our optimization problem is formulated in a way that finds the line $h$ that best fits the given step function $\Phi^\ant$. We achieve this by solving the following optimization problem on a per-antenna basis:
\begin{equation}
\begin{split}
        &\min_{\tau_n,\Phi_n} ||h(n,f) -  \Phi^\ant(n,f)||^2\\
\end{split}
\end{equation}
We can visualize this optimization in Figure \ref{fig:proof_best_line_fit}(a), where the line $h(n,f)$ is fit over the step function $\Phi^\ant(n,f)$. The slope of the best-fit line gives the delay value, and the y-intercept gives the phase value. In this way, we can estimate both delay and phase values by solving for the best-fit line.

However, as antenna index $n$ increases, the error in line fitting also increases due to the linear increase in the step size with $n$, as shown in Figure~\ref{fig:proof_best_line_fit}(b). This could lead to high error for large antenna arrays and limit our solution to scale with antennas. We have an innovative and simple solution to address this issue.  To address this issue, we utilize the concept of wrapping the phase of a signal by $2\pi$, i.e., adding an integer multiple of $2\pi$ to the phase does not change the signal. We use this idea to strategically add a phase of multiple of $2\pi$ to a specific set of frequencies in order to minimize the error in line fitting as shown in Figure~\ref{fig:proof_best_line_fit}(c). With this insight, we redefine the step function $\Phi^\ant$ as:
\begin{equation}
    \Phi^{\text{ant}}(n,f) = \begin{cases}
        k2\pi-n\pi\sin(\theta_0) & f\in[\frac{-B}{2},0]\\
        n\pi\sin(\theta_0) & f\in(0,\frac{B}{2}]\\
    \end{cases}
\end{equation}
where $k$ is a constant integer. A natural question is how do we estimate this integer to minimize the error in line fitting? Our solution is a two-step process: we solve for the delays and phases as a function of $k$ and then find the optimal value of $k$ to minimize the error. We now describe how we solve for the best-fit line to estimate the per-antenna delay and phase values.

To solve for per-antenna delays and phases, we form a system of linear equations. We first discretize the frequency as $f=m\Delta f$ for $m\in [-M/2, M/2]$, where the bandwidth is $B=(M+1)\Delta f$. Note that there are $M$ frequency bins that can be a large number, i.e., $M\rightarrow \infty$ for creating a continuous frequency axis. We then formulate a set of linear equations for each frequency term to solve for the variable delay $\tau_n$ and phase $\Phi_n$ for each antenna $n$. Specifically, we have the following linear equations:
\begin{equation}
    \Phi_n + 2\pi m\Delta f\tau_n = \Phi^\ant(n,m\Delta f) \;\forall m\in [-M/2, M/2]
\end{equation}
We re-write the equations in a matrix form as follows:
\begin{equation}
    Ax=b
\end{equation}
where $x$ is a $2 \times 1$ vector of variable phase and delay given by:
\begin{equation}\label{eq:x_def}
    x =   \begin{bmatrix}
            \Phi_n\\
            2\pi \Delta f\tau_n
     \end{bmatrix} 
\end{equation}
and the matrix $A$ and vector $b$ are constants given by:
\begin{equation}
    A =  \begin{bmatrix}
           1&\frac{-M}{2} \\
           1&\frac{-M}{2}+1 \\
           \vdots&\vdots \\
           1&\frac{M}{2}
     \end{bmatrix}    
\end{equation}
\begin{equation}
    b = \begin{bmatrix}
         \underbrace{k2\pi-\phi \;  \ldots \;k2\pi-\phi}_{ M/2} \;\underbrace{\phi \ldots \phi}_{ M/2+1}
         \end{bmatrix}^T
\end{equation}

where $\phi=n\pi \sin(\theta_0)$ is introduced for simplicity. Specifically, we add $k2\pi$ to the first half of the frequency subcarriers, and this is reflected in the value of $b$.

% \begin{equation}
%     b = \begin{bmatrix}
%          \underbrace{-\phi \;  \ldots \;-\phi}_{ M/2} \;\underbrace{\phi \ldots \phi}_{ M/2+1}
%          \end{bmatrix}^T
% \end{equation}
% \begin{equation}
%     b = a\begin{bmatrix}
%          -\one(M/2)\\
%          \vdots\\
%          \one(M/2+1)
%          \end{bmatrix}
% \end{equation}

The solution can be obtained by solving a system of linear equations as follow:
\begin{equation}
    x = (A^TA)^{-1} A^Tb
\end{equation}

We first solve for $(A^TA)^{-1}$ and $A^Tb$ separately and then multiply them together to get $x$:

\begin{equation}
\begin{split}
    A^TA &= \begin{bmatrix}
        1&1&\hdots&1\\
        \frac{-M}{2}&\frac{-M}{2}+1&\hdots&\frac{M}{2}
    \end{bmatrix}
    \begin{bmatrix}
           1&\frac{-M}{2} \\
           1&\frac{-M}{2}+1 \\
           \vdots&\vdots \\
           1&\frac{M}{2}
     \end{bmatrix} \\
    &= \begin{bmatrix}
           M+1 & 0\\
           0 & 2\sum_0^{\frac{M}{2}}k^2
    \end{bmatrix} = \begin{bmatrix}
           M+1 & 0\\
           0 & \frac{M(M+1)(M+2)}{12}
    \end{bmatrix}
\end{split}
\end{equation}

Taking the inverse of the above 2x2 matrix, we get

\begin{equation}
    (A^TA)^{-1} = \begin{bmatrix}
           \frac{1}{M+1}&0\\
           0 & \frac{12}{M(M+1)(M+2)}
    \end{bmatrix}
\end{equation}

This solves for $A^TA$. Next, we obtain $A^Tb$ as:
% and to simplify the maths, we first consider an equal bandwidth case, i.e., $N=M/2$, so
% \begin{equation}
% \begin{split}
%     A^Tb &=  \begin{bmatrix}
%         1&1&\hdots&1\\
%         \frac{-M}{2}&\frac{-M}{2}+1&\hdots&\frac{M}{2}
%     \end{bmatrix}
%     \begin{bmatrix}
%          -a\one(M/2)\\
%          \vdots\\
%          a\one(M/2+1)
%      \end{bmatrix}\\
%      & = \begin{bmatrix}
%           a\\
%           2a\sum_1^{M/2}k
%      \end{bmatrix}
%      \\&= \begin{bmatrix}
%           a\\
%           a\frac{M(M+2)}{4}
%      \end{bmatrix}
%      \end{split}
% \end{equation}

% Finally we obtain the expression of unknown $x$ as
% \begin{equation}
% \begin{split}
%     x &= (A^TA)^{-1}A^Tb\\
%     &= \begin{bmatrix}
%           \frac{1}{M+1} &0\\
%           0 & \frac{12}{M(M+1)(M+2)}
%     \end{bmatrix}
%     \begin{bmatrix}
%           a\\
%           a\frac{M(M+2)}{4}
%      \end{bmatrix}
%      \\& = \begin{bmatrix}
%           \frac{a}{M+1}\\
%           \frac{3a}{M+1}
%      \end{bmatrix}
% \end{split}
% \end{equation}
% Therefore, we get the final expression for delay and phase per antenna as
% \begin{equation}
%     \begin{split}
%         \Phi_n &= a/(M+1)\\
%         & = n\pi\sin(\theta_0)/(M+1)\\
%         &\approx 0
%     \end{split}
% \end{equation}
% where the approximation is when $M\rightarrow\infty$
% \begin{equation}
%     \begin{split}
%         \tau_n &= \frac{3a}{2\pi\Delta f (M+1)}\\
%             &=\frac{3n\sin(\theta_0)}{2B}
%     \end{split}
% \end{equation}

% It only changes $A^Tb$ as
\begin{equation}
\begin{split}
    A^Tb &=  \begin{bmatrix}
        1&1&\hdots&1\\
        \frac{-M}{2}&\frac{-M}{2}+1&\hdots&\frac{M}{2}
    \end{bmatrix}\\&\times
    \begin{bmatrix}
         \underbrace{k2\pi-\phi \;  \ldots \;k2\pi-\phi}_{ M/2} \;\underbrace{\phi \ldots \phi}_{ M/2+1}
         \end{bmatrix}^T\\
     & = \begin{bmatrix}
          \phi+kM\pi\\
          (2\phi-k2\pi)\sum_1^{M/2}m
     \end{bmatrix}
     \\&= \begin{bmatrix}
          \phi+kM\pi\\
          (\phi-k\pi)\frac{M(M+2)}{4}
     \end{bmatrix}
     \end{split}
\end{equation}
This solves for $A^Tb$. We now obtain the solution for unknown $x$ as follows:
\begin{equation}
\begin{split}
    x &= (A^TA)^{-1}A^Tb\\
    &= \begin{bmatrix}
           \frac{1}{M+1} &0\\
           0 & \frac{12}{M(M+1)(M+2)}
    \end{bmatrix}
    A^Tb
     \\& = \begin{bmatrix}
          A^Tb[1]\frac{1}{M+1}\\
          A^Tb[2]\frac{12}{(M(M+1)(M+2))}
     \end{bmatrix}
\end{split}
\end{equation}
Finally, using the definition of $x$ from (\ref{eq:x_def}), we get the solution for the per antenna phase and delay. The phase is given by $x[1]$ and delay is given by x[2] as follows:  
\begin{equation}\label{eq:phase2sol}
    \begin{split}
        \Phi_n &= \lim_{M\rightarrow\infty}x[1] = \lim_{M\rightarrow\infty}A^Tb[1]\frac{1}{M+1}\\
        &=\lim_{M\rightarrow\infty}\frac{\phi+kM\pi}{M+1}\\
        % & \approx n\pi\sin(\theta_0)/(M+1) +k\pi\\
        & = k\pi
    \end{split}
\end{equation}
where the approximation is taken for a large number of frequency bins, i.e., $M\rightarrow\infty$. Similarly, we get the expression for the delay:
\begin{equation}\label{eq:del2sol}
    \begin{split}
        \tau'_n &= \frac{x[2]}{2\pi\Delta f} 
        = A^Tb[2]\frac{12}{M(M+1)(M+2)}\frac{1}{2\pi\Delta f}\\
        &=(\phi-k\pi)\frac{M(M+2)}{4}\frac{12}{M(M+1)(M+2)}\frac{1}{2\pi\Delta f}\\
        &=\frac{3(\phi-k\pi)}{2\pi\Delta f (M+1)}\\
            &=\frac{3(n\sin(\theta_0)}{2B} - \frac{3k}{2B}\\
            &=\frac{3}{2B}(n\sin(\theta_0)-k)
    \end{split}
\end{equation}

This gives a closed-form expression of delay. Note that the expression for delay and phase both depend on the unknown constant integer $k$. 
So, how do we solve for $k$ to obtain a generalized formula for delay and phases? Out of many possible solutions, because of the presence of a random integer $k$, we need to choose the value of $k$ that minimizes the error in line fitting. To solve for $k$, we make an  observation that the step size of the step function is part of the exponential and therefore is bounded by $2\pi$.
Let's find the condition when the step size is bounded between $-\pi$ and $\pi$ as follows:
\begin{equation}
    \begin{split}
        &-\pi<\phi-(-\phi+k2\pi) \le \pi\\
        \Rightarrow\;\;&\pi<2\phi-k2\pi \le \pi\\
        \Rightarrow\;\;&\pi<2n\pi\sin(\theta_0)-k2\pi \le \pi\\
        \Rightarrow\;\;&\frac{1}{2}<n\sin(\theta_0)-k \le \frac{1}{2}\\
        \Rightarrow\;\;&\sin(\theta_0)-\frac{1}{2}<k \le \sin(\theta_0)+\frac{1}{2}\\
        \Rightarrow\;\;&k = \text{round}(n\sin(\theta_0))\\
    \end{split}
\end{equation}
Since $k$ is an integer, the only possible value of $k$ is given by $\text{round}(n\sin(\theta_0))$. This gives us a unique solution for the integer constant $k$. Putting the in the expression of phase and delay, we get the required phase:
\begin{equation}\label{eq:phase_derive}
    \Phi_n = \text{round}(n\sin(\theta_0))\pi 
\end{equation}
Similarly, we obtain the final expression of delay as:
\begin{equation} \label{eq:tau_derive_noshift}
    \begin{split}
        \tau'_n &= \frac{3}{2B}(n\sin(\theta_0)-\text{round}(n\sin(\theta_0)))
    \end{split}
\end{equation}
We notice the above expression of delay varies in the range of $\frac{-3}{4B}$ to $\frac{3}{4B}$, resulting in a total range of $\frac{3}{2B}$. Since negative delays are not possible to generate in a causal system, we can add a constant delay to all antennas to make delays positive. The minimum constant delay factor that we can add without compromising on the performance is $\frac{3}{4B}$. After adding this factor and simplifying the delay expression, we get the following: 
\begin{equation} \label{eq:tau_derive}
    \begin{split}
        \tau_n &= \tau'_n + \frac{3}{4B}\\
        &=\frac{3}{2B}(n\sin(\theta_0)-\text{round}(n\sin(\theta_0))) +\frac{3}{4B}\\
            &=\left(\frac{3}{2B}n\sin(\theta_0) +\frac{3}{4B}\right)\;\;\;\text{mod }\frac{3}{2B} 
    \end{split}
\end{equation}
where the last step is a simplification of the round function into a modulo function for ease of understanding and emphasizing that the range of values of delays is independent of the number of antennas. The delay range is $\frac{3}{2B}$ for two beam case which depends inverse
This is how we estimate the final expression of optimal phase $\Phi_n$ and delay $\tau_n$ in (\ref{eq:phase_closed_proof}) and (\ref{eq:tau_closed_proof}), respectively. 

\end{proof}
% \qedsymbol

%%%%%%%%%%%%%%%-------------------------------%%%%%%%%%%%%%%

\begin{figure}[!t]
    \centering
    \includegraphics[width=\columnwidth]{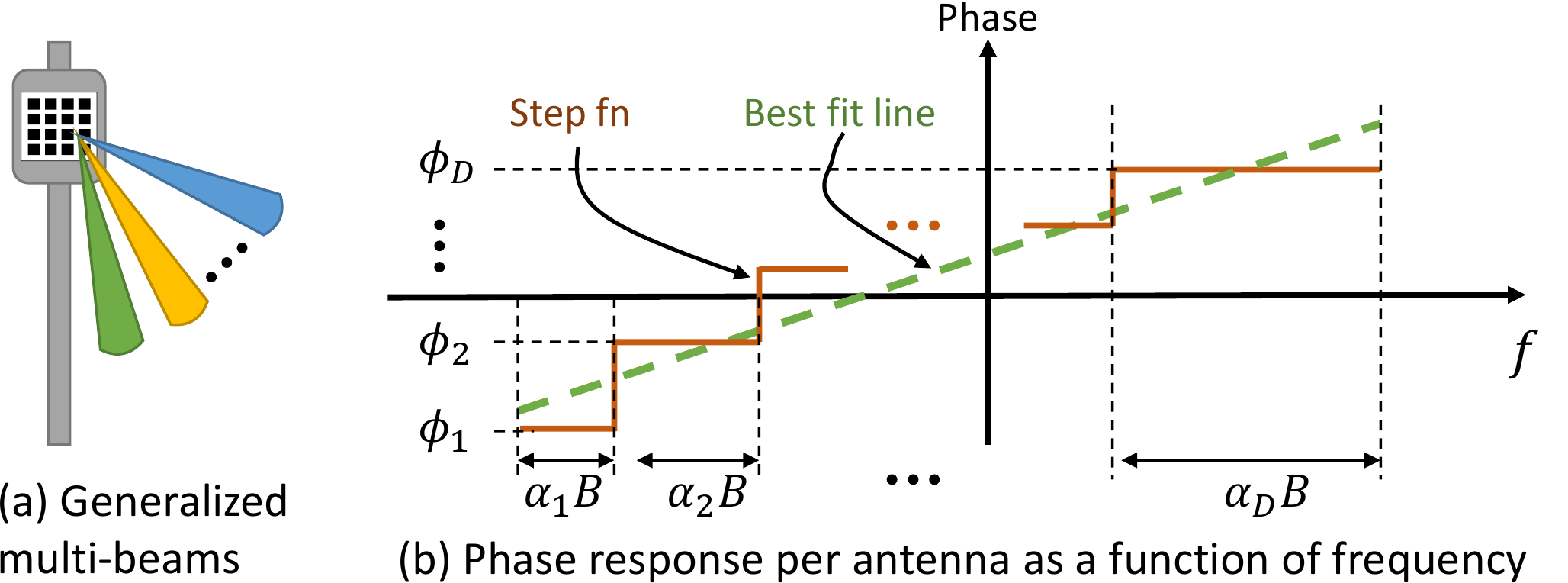}
    \caption{Generalized setting for multi-beams with an arbitrary number of beams $D$, beam angles, and beam-bandwidths. The beam angles contributes to $\phi_d$, and the beam-bandwidth is defined as $\alpha_d B$ for system bandwidth $B$ and fraction $\alpha_d<1$ for beam index $d\le D$.}
    \label{fig:proof_generalized}
\end{figure}
\subsubsection{Generalization to an arbitrary number of beams, beam directions, and beam-bandwidth}
We now generalize the beamforming response to an arbitrary number of beams with arbitrary beam directions and arbitrary beam-bandwidths as shown in Fig.~\ref{fig:proof_generalized}.

\begin{theorem}\label{th:generalized}
\textbf{(Generalized case):} Let there are $D$ beam directions with beam angles $\theta_d$ and beam-bandwidth $\alpha_d B$ for $\sum_d \alpha_d = 1$. We define $\phi_d = n\pi\sin(\theta_d)$ for simplicity. The per-antenna phases and delays in realizing such a generalized beamforming response is given by:
\begin{equation}\label{eq:phase_gen}
    \Phi_n = \sum_{d=1}^D\alpha_d(\phi_d+2k_d\pi)
\end{equation}
\begin{equation}\label{eq:delay_gen}
    \tau_n = \sum_{d=1}^D\frac{3}{\pi B}(\phi_d+2\pi k_d)\alpha_d(\sum_{\ell=1}^{d}2\alpha_{\ell} -\alpha_d -1)
\end{equation}
where,
\begin{equation}
    \phi_d = n\pi\sin(\theta_d)
\end{equation}
and the constant integer $k_d$ for beam $d$ and antenna $n$ is:
% \begin{equation}
%     \phi_d = n\pi\sin(\theta_d)
% \end{equation}
% and
\begin{equation}\label{eq:kd_gen}
    k_d =\begin{cases}
    &0\quad d=1\\
    &k_{d-1}\!+\!\text{round}(\frac{n\sin(\theta_{d-1})-n\sin(\theta_d)}{2}) \;\;   d\ge 2
    \end{cases}
\end{equation}
% is a constant integer.\\

\end{theorem}
Before we provide a derivation for the generalized delay and phase expression, we first verify the expression for the two-beam case by putting $D=2$ in the above expression. We would obtain a generic expression  with arbitrary beam-bandwidth and arbitrary beam directions and then simplify it to the case of equal beam-bandwidth and symmetric direction similar to the previous analysis.

% \noindent
% \textbf{$\blacksquare$ Corollary :} 
\begin{corollary}\label{cor:2beam} (Generalized two-beam case)
For two beams with beam directions $\theta_1$ and $\theta_2$ and the beam-bandwidths $\alpha B$ and $(1-\alpha) B$ respectively, such that the bandwidth fraction $0\le \alpha\le 1$, the corresponding delays and phases that generate this two-beam response is given by:
\begin{equation}
    \Phi_n = \alpha\phi_1+ (1-\alpha)(\phi_2+2k_2\pi)
\end{equation}
\begin{equation}
    \tau_n = \frac{3}{\pi B}[(-\phi_1
        +\phi_2+2\pi k_2)\alpha(1-\alpha)]
\end{equation}
for $\phi_1 = n\pi\sin(\theta_1)$ and $\phi_2 = n\pi\sin(\theta_2)$, and the constant integer $k_2$ defined as
\begin{equation}
    k_2 = \round(\frac{n\sin(\theta_1)-n\sin(\theta_2)}{2})
\end{equation}
\end{corollary}
% \noindent
% \textbf{$\blacksquare$ Proof for corollary:}
\begin{proof}[Proof of Corollary \ref{cor:2beam}]
We prove the corollary for the two-beam case by putting the number of beams $D=2$ in the generalized expressions and further simplify them with $\alpha_1=\alpha$ and $\alpha_2=1-\alpha$. Also, note that $k_1=0$, by definition. 
We now simplify the phase as follows:
\begin{equation}
    \begin{split}
        \Phi_n &= \alpha_1(\phi_1+2k_1\pi)+ \alpha_2(\phi_2+2k_2\pi)\\
        & = \alpha\phi_1+ (1-\alpha)(\phi_2+2k_2\pi)\\
    \end{split}
\end{equation}
and delay
\begin{equation}
    \begin{split}
        \tau_n &= \frac{3}{\pi B}[(\phi_1+2\pi k_1)\alpha_1(\alpha_1 -1)\\
        &\quad+ (\phi_2+2\pi k_2)\alpha_2(2\alpha_1 +\alpha_2 -1)]\\
        &=\frac{3}{\pi B}[(\phi_1)\alpha(\alpha-1)
        + (\phi_2+2\pi k_2)(1-\alpha)(\alpha)]\\
        &=\frac{3}{\pi B}[(-\phi_1
        +\phi_2+2\pi k_2)\alpha(1-\alpha)]\\
    \end{split}
\end{equation}
This proves the corollary.
\end{proof}
We will now verify the delays and phases for the special case of equal beam-bandwidth, i.e., $\alpha=1/2$ and symmetric angles, i.e., $\phi_1 = -n\pi\sin(\theta_0)$ and $\phi_2=n\pi\sin(\theta_0)$. In this case, we get $k_2=-\round(n\sin(\theta_0))$ and a simplified phase as:
\begin{equation}
    \begin{split}
        \Phi_n &=\alpha\phi_1+ (1-\alpha)(\phi_2+2k_2\pi)\\
            &=\phi_1/2 + \phi_2/2 +k_2\pi\\
            &=k_2\pi\\
            &=-\round(n\sin(\theta_0))\pi\\
            &=\round(n\sin(\theta_0))\pi
    \end{split}
\end{equation}
which is an integer multiple of $\pi$, where the integer multiple is given by $\round(n\sin(\theta_0))$. This expression is the same as what we derived earlier in (\ref{eq:phase_derive}). 
We now get delay expression as:
\begin{equation}
    \begin{split}
        \tau_n &=  \frac{3}{\pi B}[(-\phi_1
        +\phi_2+2\pi k_2)\alpha(1-\alpha)]\\
        &=\frac{3}{\pi B}[(n\pi\sin(\theta_0)
        +n\pi\sin(\theta_0)+2\pi k_2)\alpha(1-\alpha)]\\
        &=\frac{3}{ B}[(2n\sin(\theta_0)
        +2 k_2)\alpha(1-\alpha)]\\
        &=\frac{3}{ 2B}(n\sin(\theta_0)
        +k_2)\\
        &=\frac{3}{ 2B}(n\sin(\theta_0)
        - \round(n\sin(\theta_0)))\\
    \end{split}
\end{equation}
This delay expression is the same as what we have derived earlier for the simple 2-beam case in (\ref{eq:tau_derive_noshift}). This validates the generalized multi-beam expressions for two beams. We now prove the generalized expression for an arbitrary number of beams, beam directions, and beam-bandwidths.

% \noindent
% \textbf{$\blacksquare$ Proof for generalized expression:}
\begin{proof}[Proof of Theorem \ref{th:generalized}] \textbf{(Generalized case)}
We follow the same formulation as before as a set of linear equations with a new set of constraints for the generalized case. We notice that the constant matrix $A$ doesn't change for the generalized case, and only the constant vector $b$ is modified. We would solve for delays and phases similar to the simple 2-beam case using pseudo-inverse (Recall $x = (A^TA)^{-1} A^Tb$).

We start with a generic expression for vector $b$ as:

\begin{equation}
    b = 
    \begin{bmatrix}
         \underbrace{\phi_1+k_12\pi   \ldots, }_{\alpha_1 M}\underbrace{\;\ldots\;}_{\ldots} \underbrace{,\ldots\phi_D+k_D2\pi }_{(\alpha_D) M}
         \end{bmatrix}^T
\end{equation}

Therefore, solution for $A^Tb$ is now:
\begin{equation}
\begin{split}
    A^Tb &=  \begin{bmatrix}
        1&1&\hdots&1\\
        \frac{-M}{2}&\frac{-M}{2}+1&\hdots&\frac{M}{2}
    \end{bmatrix}
    b
     \end{split}
\end{equation}

We first solve for phase. Since the matrix $A$ is unchanged, we directly apply $(A^TA)^{-1}$ from the two-beam case (\ref{eq:phase2sol}) and solve for the per-antenna phase as follows:
\begin{equation}
    \begin{split}
        \Phi_n &=\lim_{M\rightarrow \infty} A^Tb[1] \frac{1}{M+1}\\
        &= \lim_{M\rightarrow \infty}\frac{\sum_{d=1}^D\alpha_dM(\phi_d+2k_d\pi)}{M+1}\\
        & =  \sum_{d=1}^D\alpha_d(\phi_d+2k_d\pi)
    \end{split}
\end{equation}
This proves the expression of phase.
We now solve for the delay in a similar way as we solved for the two-beam case in (\ref{eq:del2sol}) as follows:
\begin{equation}
    \begin{split}
        \tau_n &= \lim_{M\rightarrow \infty} A^Tb[2]  \frac{12}{M(M+1)(M+2)} \frac{1}{2\pi \Delta f}\\
        &= \lim_{M\rightarrow \infty} A^Tb[2]  \frac{12}{M(M+2)} \frac{1}{2\pi B}\\
            &= \lim_{M\rightarrow \infty} \bigl[(\phi_1+k_12\pi)(-\sum_{m=0}^{M/2}m + \sum_{m=0}^{\alpha_1M-M/2}m)\\
            &\quad\quad + (\phi_2+k_22\pi)(-\sum_{m=0}^{\alpha_1M-M/2}m + \sum_{m=0}^{(\alpha_1+\alpha_2)M-M/2}m)\\
            & \quad \cdots +(\phi_D+k_D2\pi)(\sum_{m=0}^{\sum_{\ell=1}^D\alpha_\ell M-M/2}\!\!\!\!\!m - \sum_{m=0}^{\sum_{\ell=1}^{D-1}\alpha_\ell M-M/2}\!\!\!\!\!m) \bigr]\\
            & \quad  \quad \times\frac{12}{M(M+2)} \frac{1}{2\pi B}\\
    \end{split}
\end{equation}
where we put the expression for $B = (M+1)\Delta f$ for simplification. Still, the above expression is complex and messy, which we intend to write in a simplified closed-form solution. 
Our idea to simplify the above expression is that since there is a $M^2$ term in the denominator, we only collect the $M^2$ terms from the numerator and ignore other constant or linear terms. This is because we are interested in the case when the number of frequency bins is very high, i.e., $\lim_{M\rightarrow \infty}$. This leads to a simplified expression as follows:
% Upon further simplification for delay values, we get:
\begin{equation}
    \begin{split}
        \tau_n 
            &=\frac{3}{\pi B}\bigl[(\phi_1+k_12\pi)(-(1/2)^2+(\alpha_1-1/2)^2)\\
            &\quad\quad + (\phi_2+k_22\pi)(-(\alpha_1-1/2)^2+(\alpha_1+\alpha_2-1/2)^2 )\\
            & \quad \quad \cdots (\phi_D+k_D2\pi)((\sum_{\ell=1}^{D}\alpha_\ell-1/2)^2-(\sum_{\ell=1}^{D-1}\alpha_\ell -1/2)^2 ) \bigr]\\
            &=\frac{3}{\pi B}\bigl[(\phi_1+k_12\pi)(\alpha_1(\alpha_1-1))\\
            &\quad\quad + (\phi_2+k_22\pi)(\alpha_2(2\alpha_1+\alpha_2-1))\\
            & \quad \quad \cdots (\phi_D+k_D2\pi)(\alpha_D(\sum_{\ell=1}^{D-1}2\alpha_{\ell} +\alpha_D - 1)) \bigr]\\
    \end{split}
\end{equation}
This proves the generalized expression for phases and delays in (\ref{eq:phase_gen}) and (\ref{eq:delay_gen}), respectively, for an arbitrary number of beams, beam directions, and beam-bandwidth.
\end{proof}

% We obtain the expression for the integer constant $k_d$ for the beam $d$. The insight is that we can add an integer multiple of $2\pi$ to the per-beam phase to reduce the error in line fitting, without changing the actual value of phase, due to phase wrapping. In the two-beam case, we have seen that we get the best line fitting when we scale the phase of one beam by $2\pi$ with respect to the other beam whenever the difference between the phase between two beams grows larger than $\pi$. Adding $2\pi$ brings the phase difference magnitude to become less than $\pi$ and that minimizes the error in line fitting. We generalize this insight for the arbitrary $D$ beam case. In this case, we want to make sure that the phase difference between two consecutive beams remains less than $\pi$, or otherwise, we can always add $2\pi$ to achieve that relative phase difference. This insight leads to the expression for $k_d$ in the final equation. Note that we fix the phase of beam $D$, i.e., not change it, and scale the phase of all other beams with respect to beam $D$. This setting ensures that no two consecutive beams have a phase difference of more than $\pi$ and this bound on phase difference leads to a bound of error in line fitting. 

\begin{figure*} [!t]
\centering
\subfigure[Desired F-S response]{
    \includegraphics[width=0.21\textwidth]{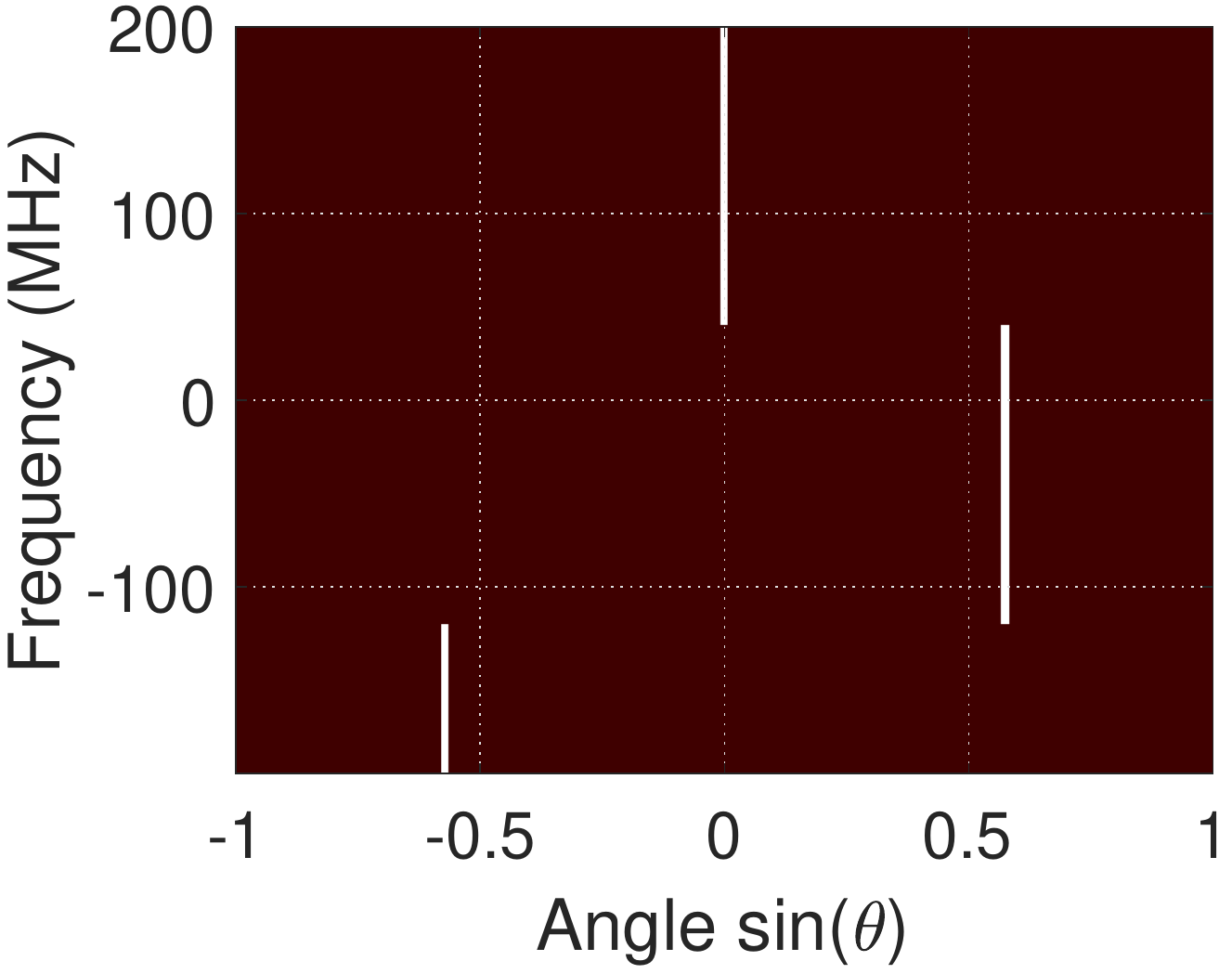}
    \label{fig:image_desired_case1}
  }\hfill
  \subfigure[FSDA F-S response]{
    \includegraphics[width=.23\textwidth]{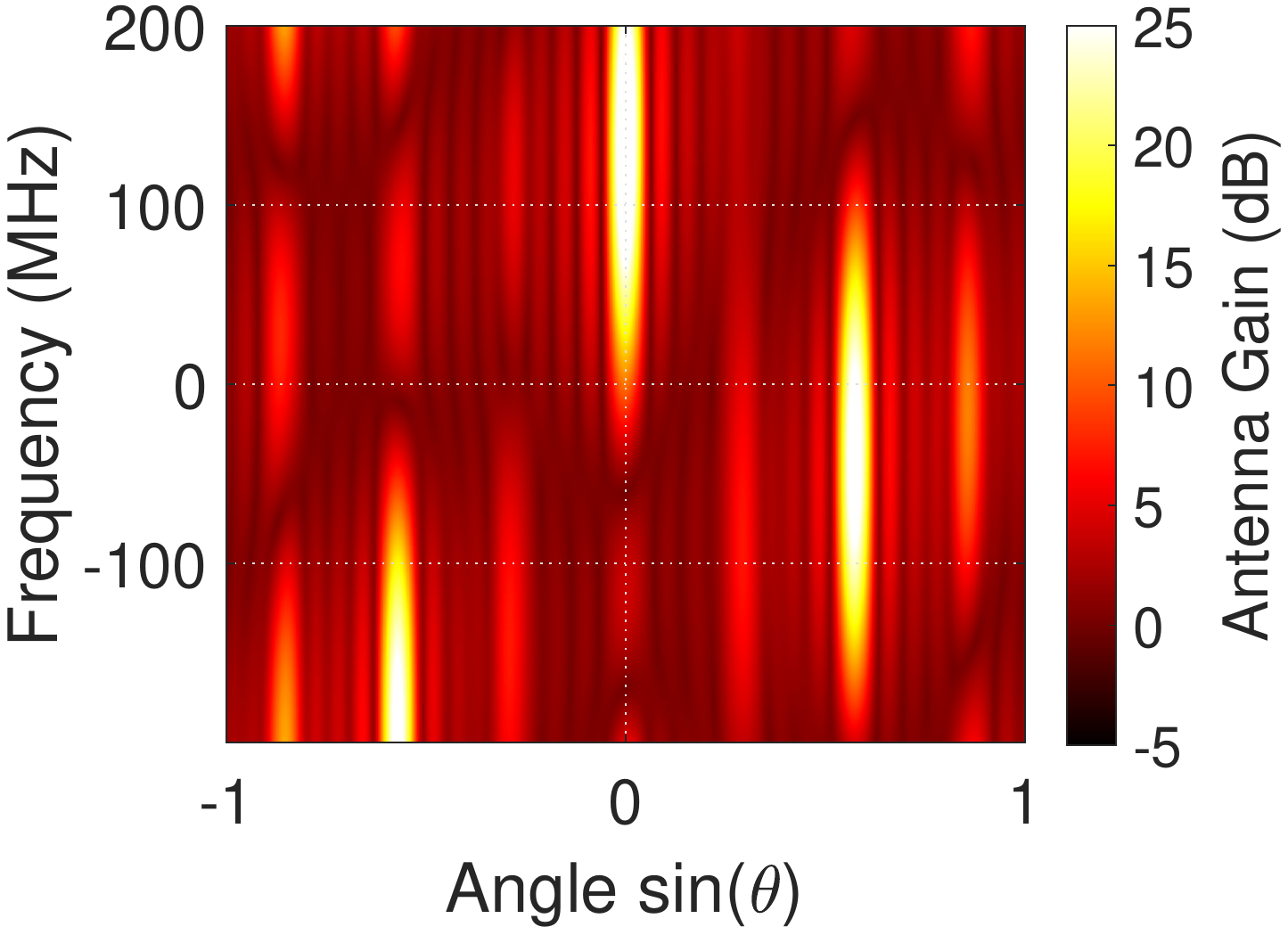}
    \label{fig:image_fsda_case1}
  }\hfill
  \subfigure[Math F-S response]{
    \includegraphics[width=.23\textwidth]{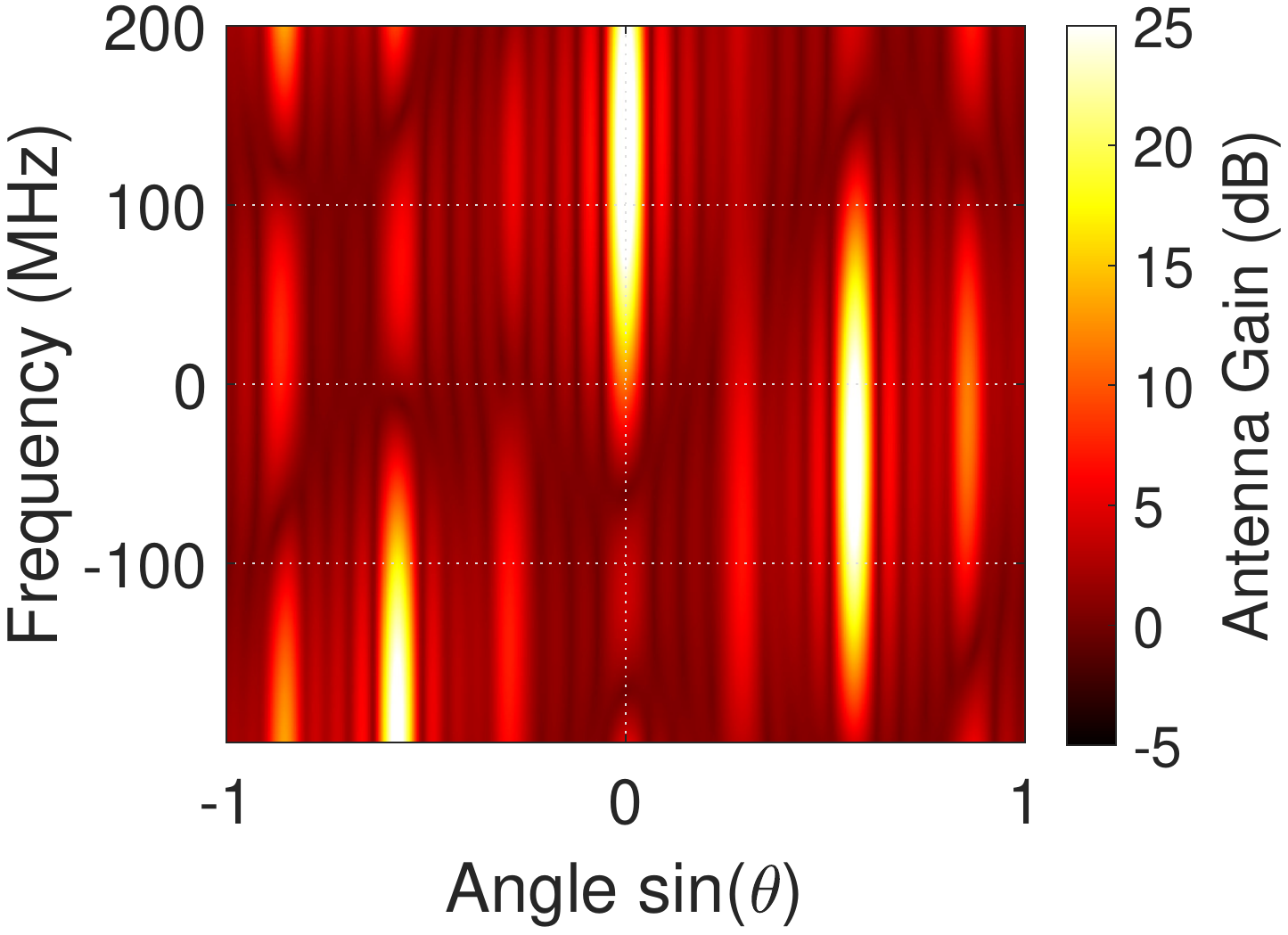}
    \label{fig:image_math_case1}
  }\hfill
  \subfigure[Phases and Delays]{
    \includegraphics[width=.22\textwidth]{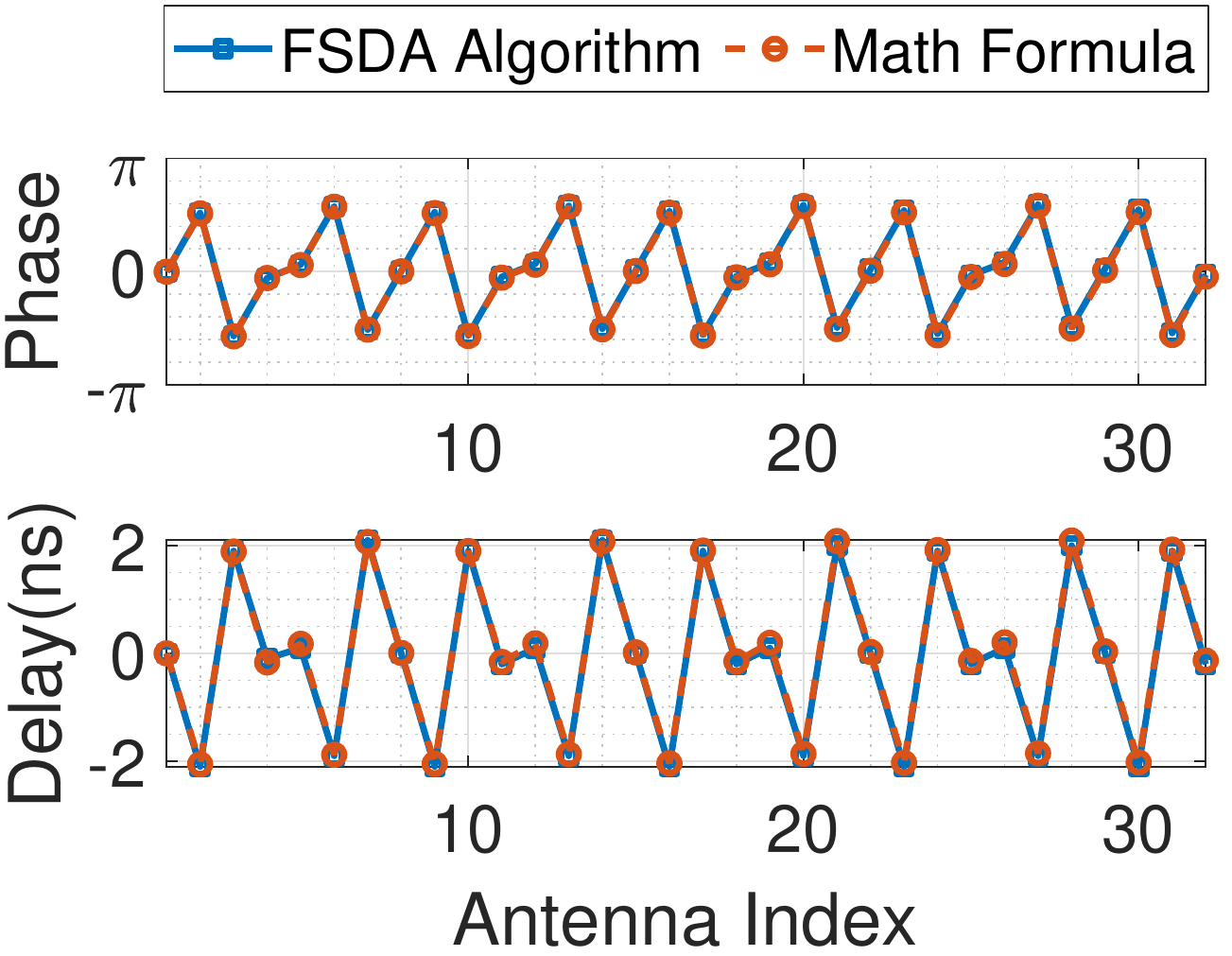}
    \label{fig:phase_delay_beamcase1}
  }\hfill
    \caption{Generalized 3-beam case: Comparing the Frequency-Space (F-S) beamforming response of FSDA algorithm and Maths. 
    }
    \label{fig:3beam}
\end{figure*}

\begin{figure*} [!t]
\centering
\subfigure[Desired F-S response]{
    \includegraphics[width=0.21\textwidth]{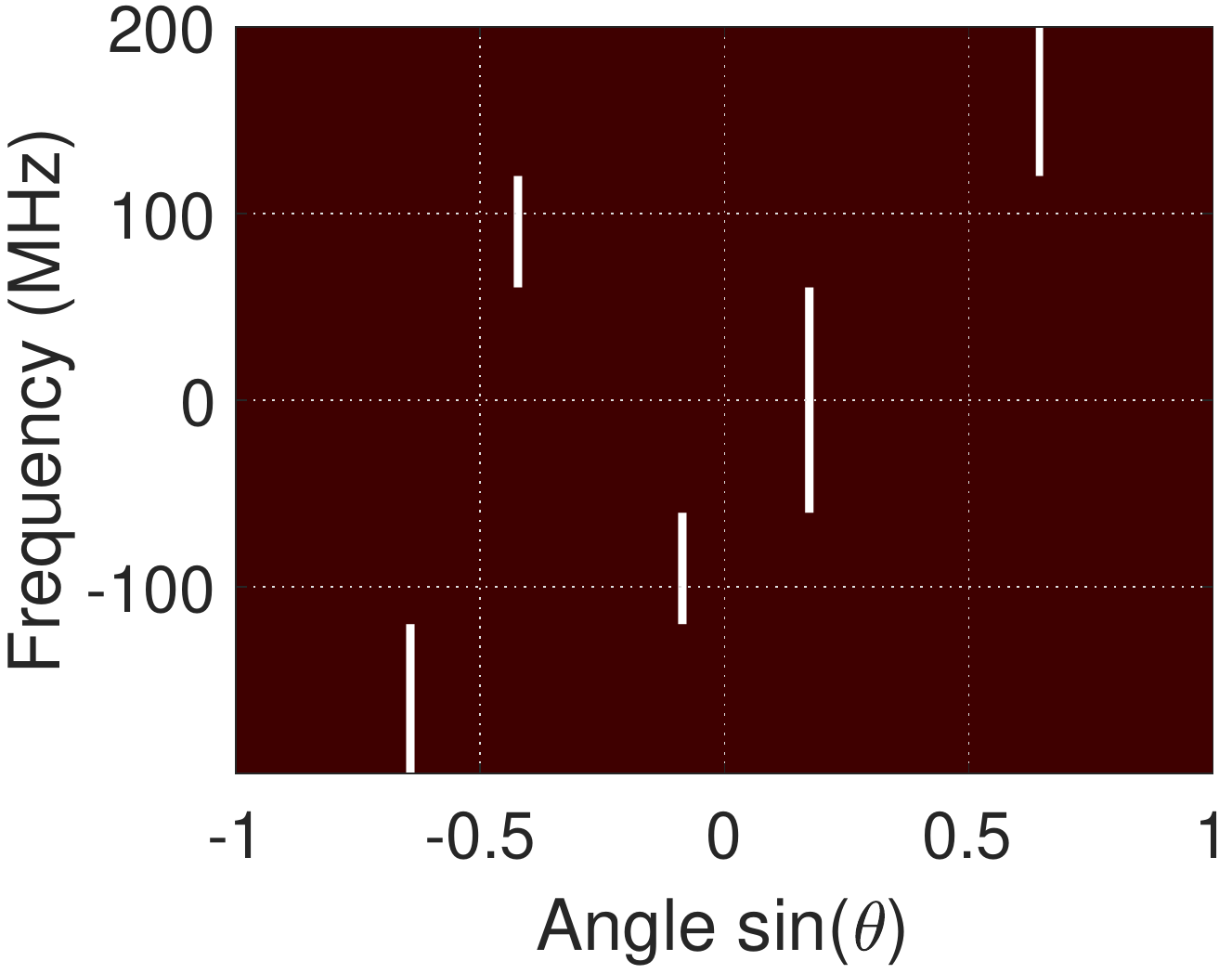}
    \label{fig:image_desired_case2}
  }\hfill
  \subfigure[FSDA F-S response]{
    \includegraphics[width=.23\textwidth]{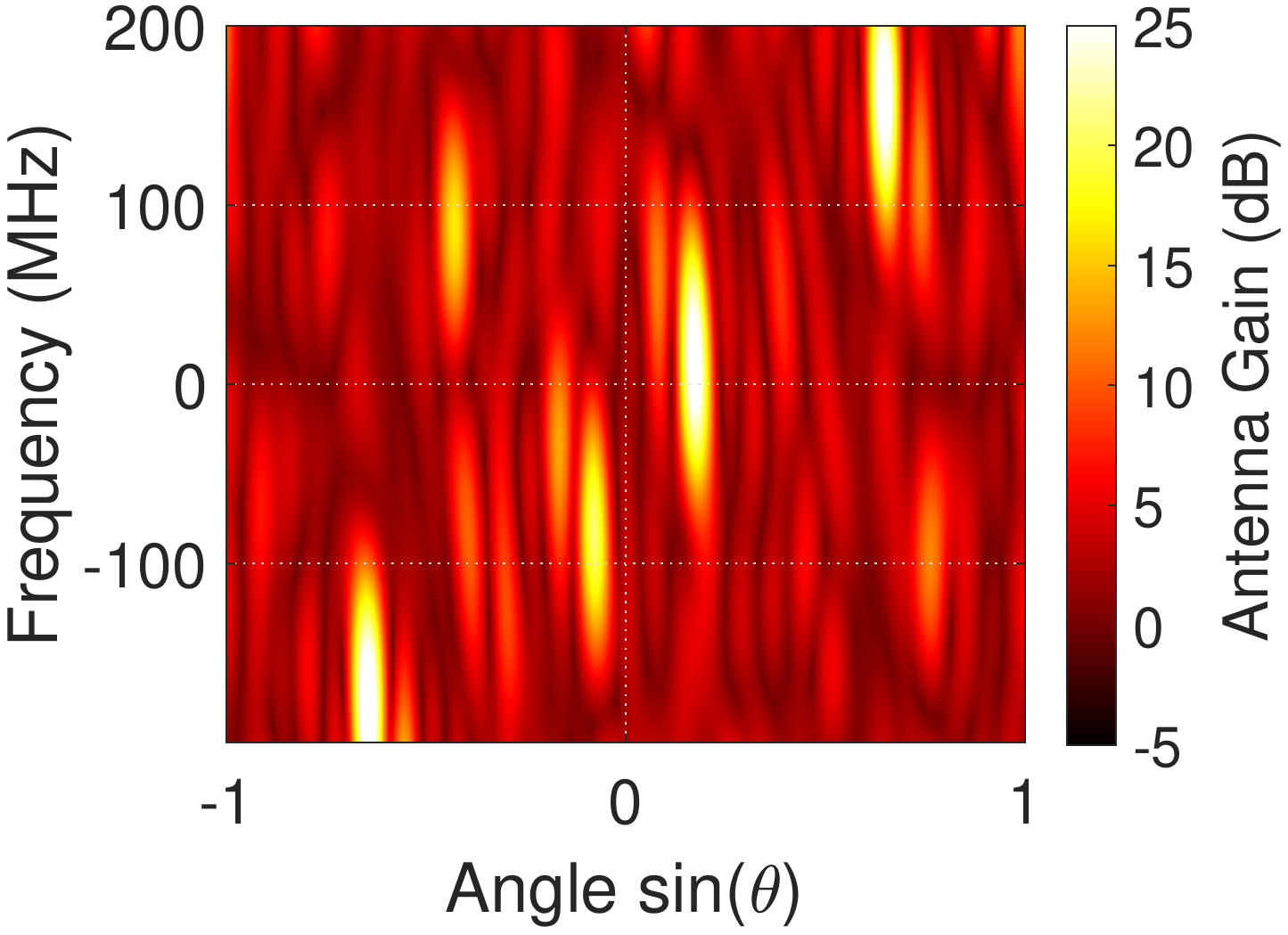}
    \label{fig:image_fsda_case2}
  }\hfill
  \subfigure[Math F-S response]{
    \includegraphics[width=.23\textwidth]{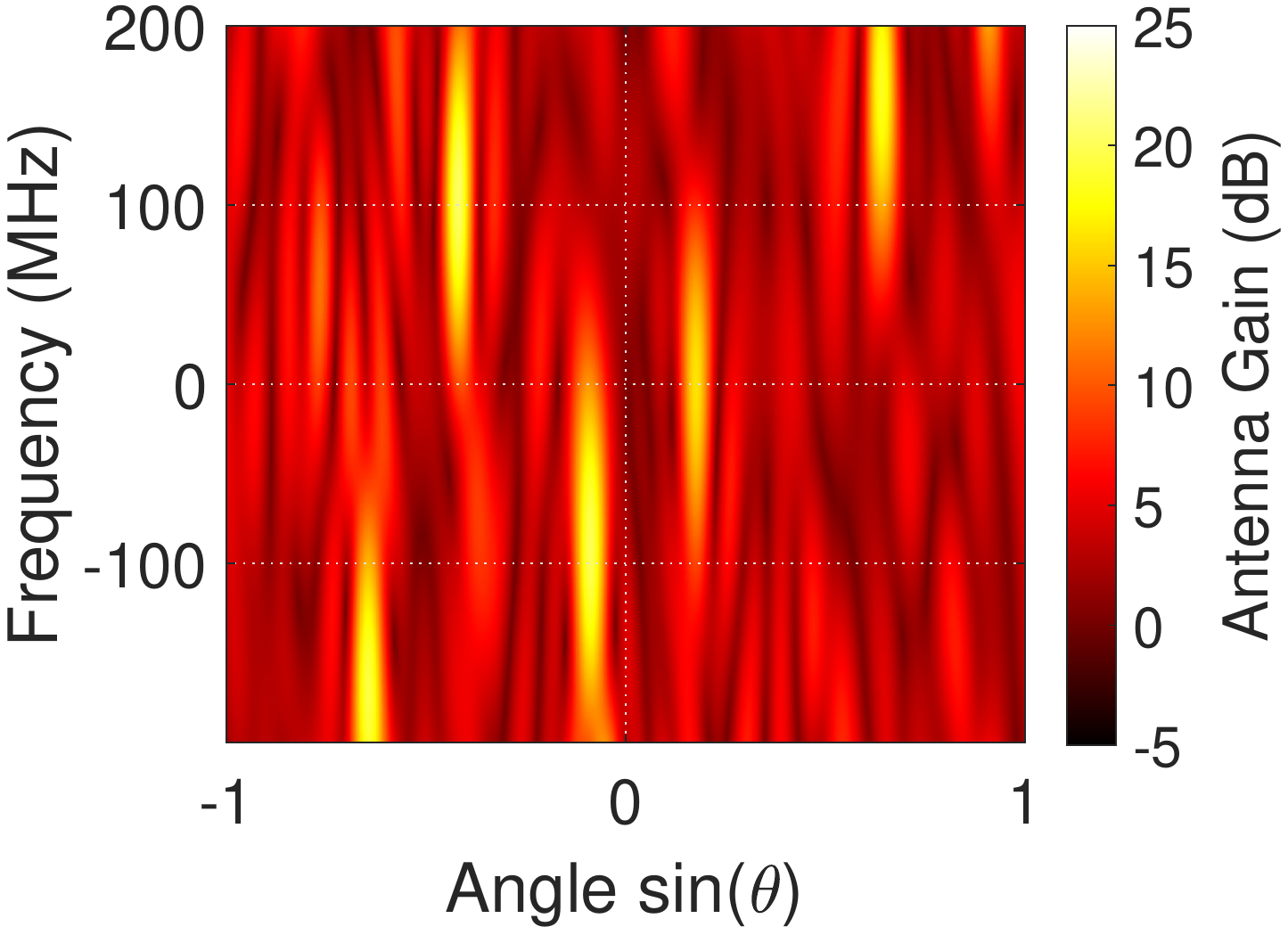}
    \label{fig:image_math_case2}
  }\hfill
  \subfigure[Phases and Delays]{
    \includegraphics[width=.22\textwidth]{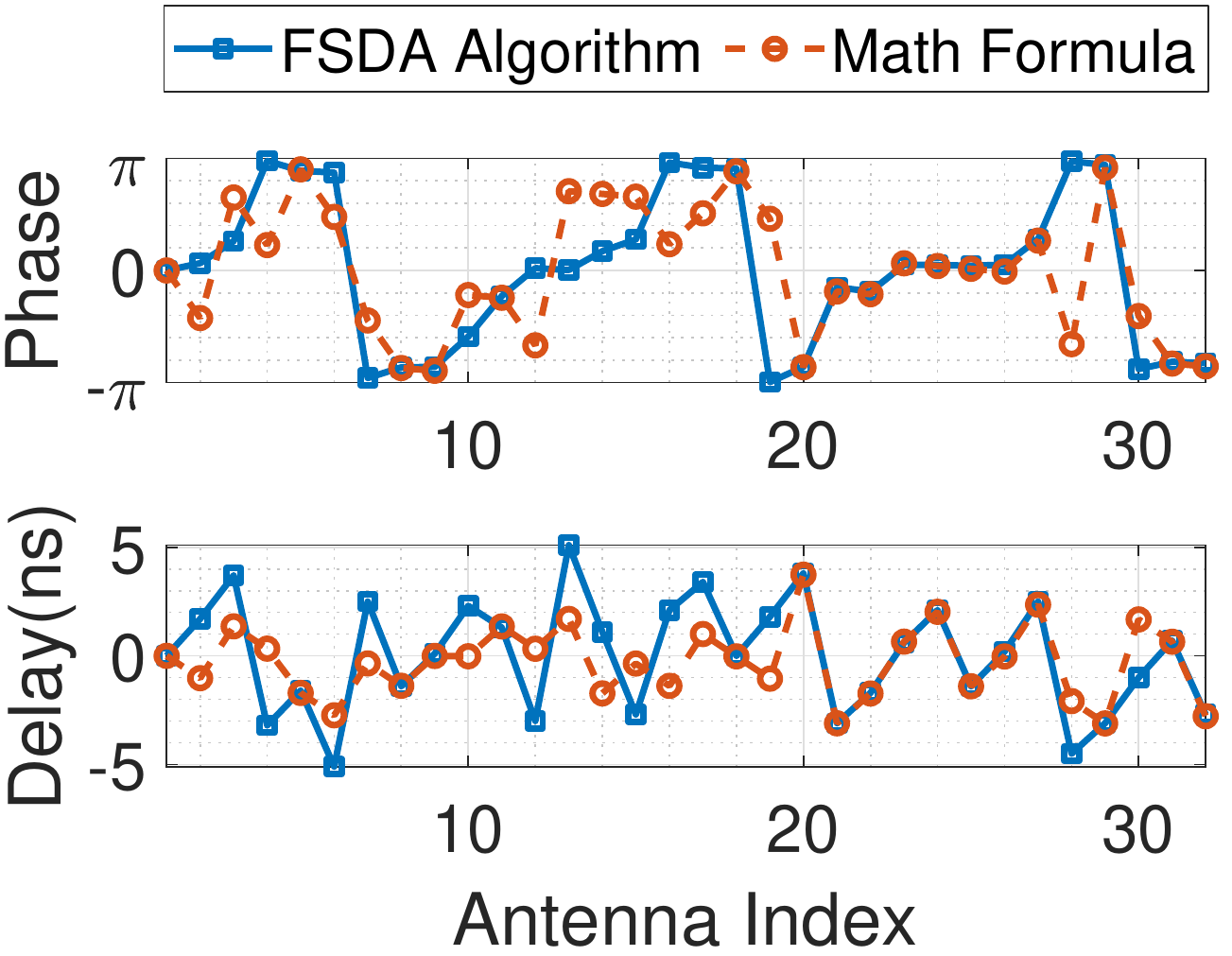}
    \label{fig:phase_delay_beamcase2}
  }\hfill
    \caption{Generalized 5-beam case: Comparing the Frequency-Space (F-S) beamforming response of FSDA algorithm and Maths. 
    }
    \label{fig:5beam}
\end{figure*}

\begin{figure*} [!t]
\centering
\subfigure[Desired F-S response]{
    \includegraphics[width=0.21\textwidth]{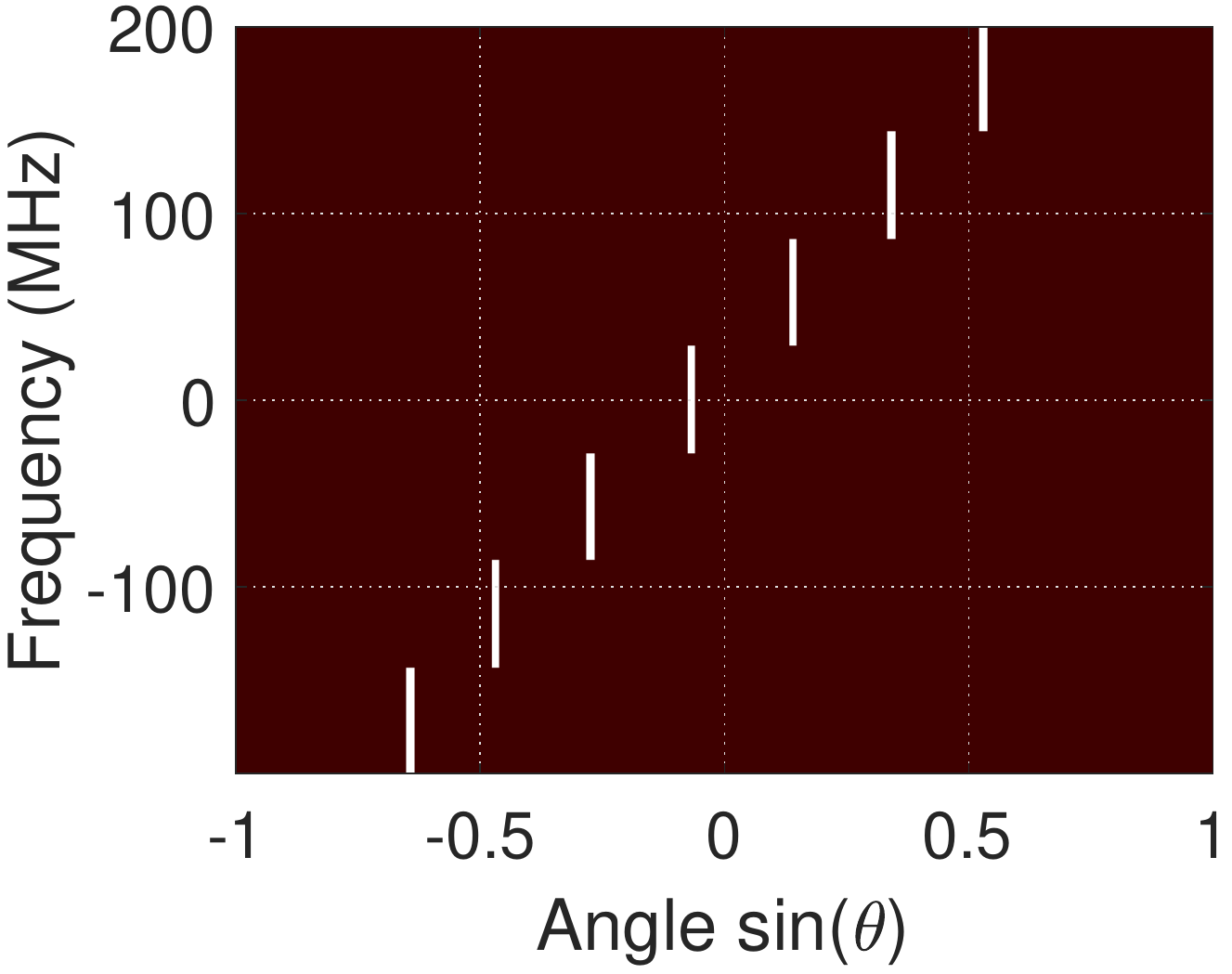}
    \label{fig:image_desired_case3}
  }\hfill
  \subfigure[FSDA F-S response]{
    \includegraphics[width=.23\textwidth]{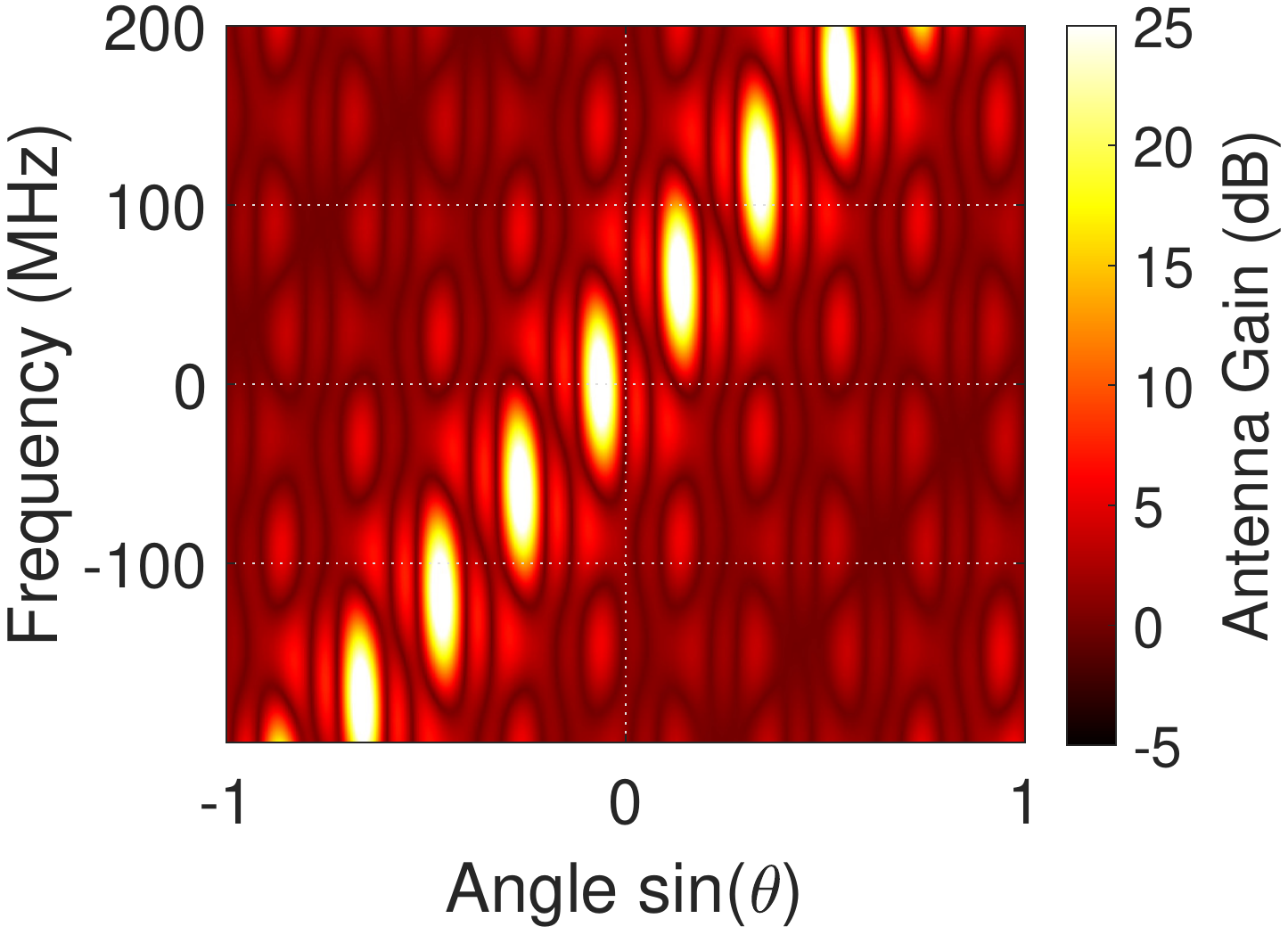}
    \label{fig:image_fsda_case3}
  }\hfill
  \subfigure[Math F-S response]{
    \includegraphics[width=.23\textwidth]{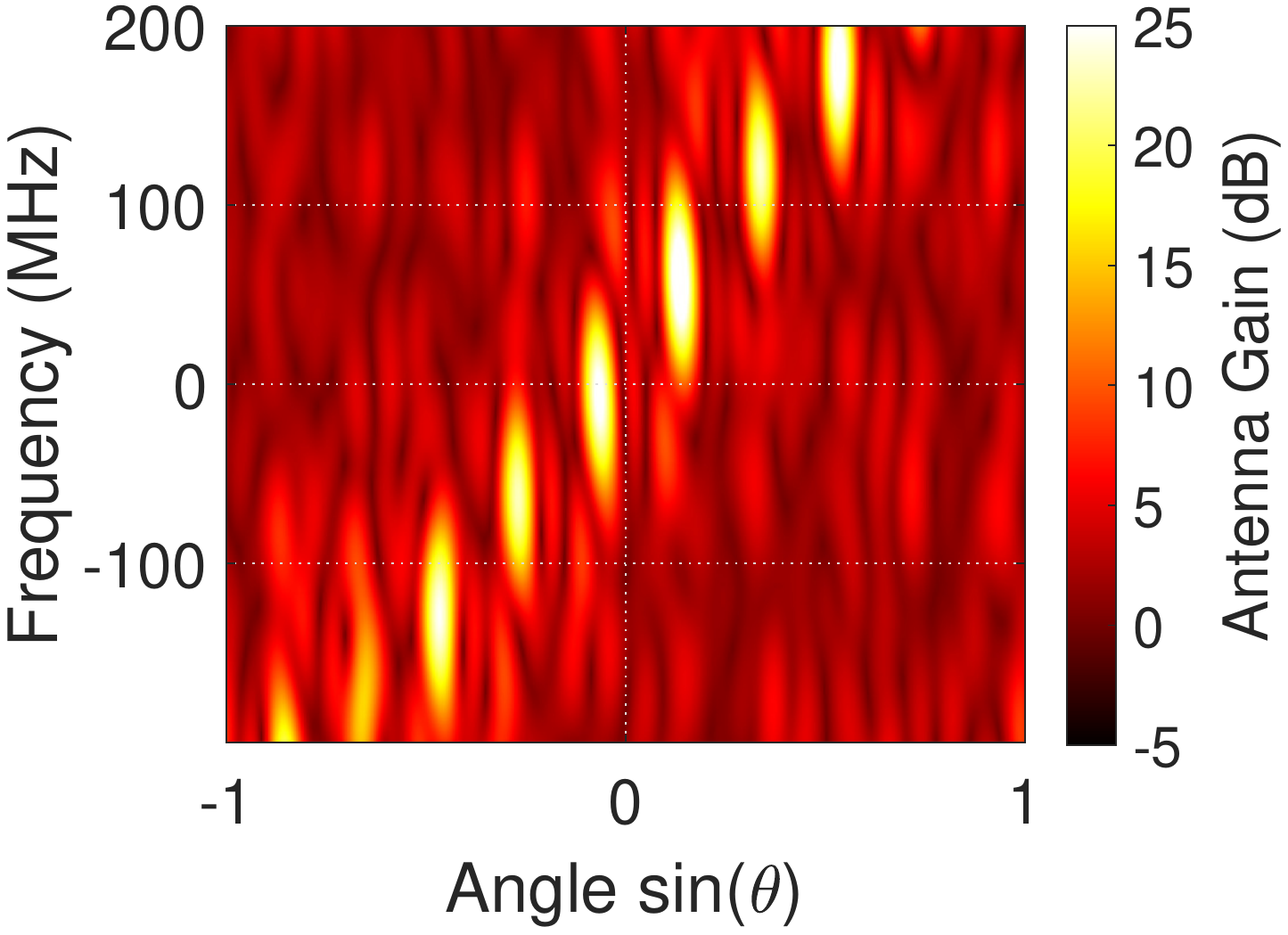}
    \label{fig:image_math_case3}
  }\hfill
  \subfigure[Phases and Delays]{
    \includegraphics[width=.22\textwidth]{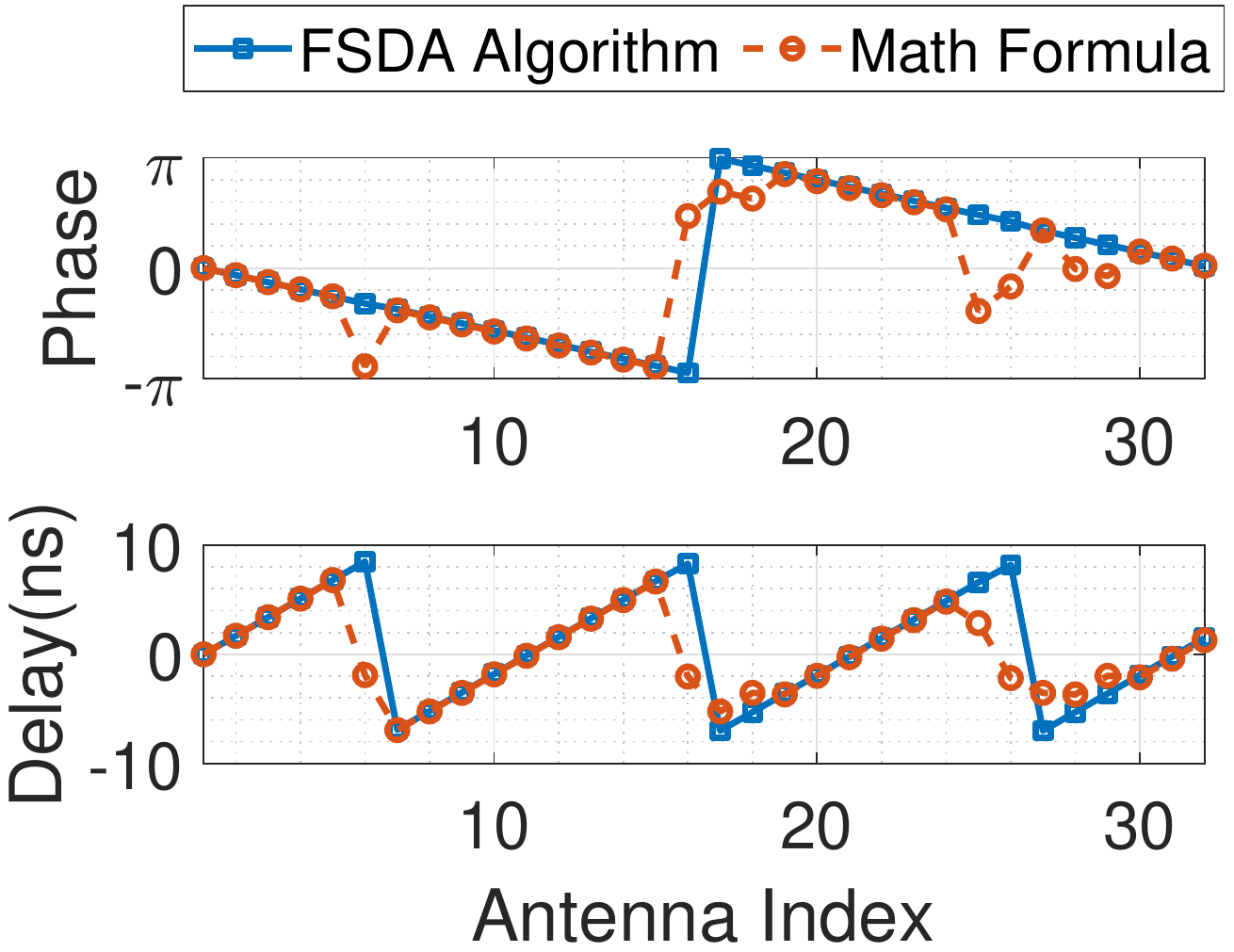}
    \label{fig:phase_delay_beamcase3}
  }\hfill
    \caption{Uniform 7-beam case: Comparing the Frequency-Space (F-S) beamforming response of FSDA algorithm and Maths. 
    }
    \label{fig:6beam}
\end{figure*}

We now discuss insights into how we obtain the expression for the unknown integer constant $k_d$ to bind the error in line fitting. Recall in the two-beam system, we have seen that by strategically adding an integer multiple of $2\pi$ to the phase of one beam, we can reduce the error in line fitting without changing the actual value of the phase. This is due to the concept of phase wrapping, where adding $2\pi$ to a phase value results in the same point on the complex plane.

To generalize this insight for an arbitrary number of beams, we need to ensure that the phase difference between any two consecutive beams is less than $\pi$. Because if the phase difference goes more than $\pi$, then we can always add $2\pi$ to the phase of one of the two consecutive beams to make the phase difference less than $\pi$. By doing this, we can ensure that the step size doesn't grow high for a large antenna index, thus minimizing the error in line fitting. To obtain the expression for the integer constant $k_d$ for each beam $d$, we start by fixing the phase of the first beam and then adjust the phase of the second beam with respect to the first beam and follow this process to all consecutive antennas: adjust the phase of $d$th beam with respect to the $d-1$th beam. This ensures that no two consecutive beams have a phase difference greater than $\pi$. This results in the final expression for $k_d$ in (\ref{eq:kd_gen}), which expresses the integer constant multiple of $2\pi$ that is added to the phase of each beam to achieve the desired bound on the phase difference. This bound on phase difference leads to a bound on the error in line fitting, resulting in an accurate line fitting for generalized multi-beam systems.

\subsection{Evaluation of Mathematical Multi-beam patterns}
Here we provide an extensive evaluation for \name under different scenarios. We compare the beamforming patterns that are produced by Mathematical expression vs FSDA algorithm. We also show the impact of delay and phase quantization to show the practical performance of \name.

\noindent
\textbf{$\blacksquare$ Comparison of Maths expression vs. FSDA algorithm:}
We present several examples of joint frequency-space beamforming using the DPA architecture. We utilize a 32-element linear antenna array with programmable phase and delay elements per antenna in various scenarios, including three-beam, five-beam, and seven-beam configurations as shown in Figure~\ref{fig:3beam},~\ref{fig:5beam}, and ~\ref{fig:6beam} respectively. In, Figure~\ref{fig:3beam}(a), we show a desired frequency-space (F-S) image that illustrates high intensity at the desired frequency-direction pairs, where we want high array gain. We then show two F-S images that were obtained through DPA after programming them using the FSDA algorithm and a mathematical expression in Figure~\ref{fig:3beam}(b) and (c), respectively. Additionally, we present the set of delays and phase values obtained from the FSDA algorithm and mathematical expression at each of the 32 antennas in Figure~\ref{fig:3beam}(d).

\begin{enumerate}[leftmargin=*]
    \item 
    DPA can reasonably achieve the desired multi-beam patterns with an arbitrary number of beams, beam directions and beam bandwidths. We observe that the antenna gain is high in the desired frequency-direction pairs and low elsewhere, as expected. This verifies that the linear approximation of the non-linear per-antenna phase profile is fairly accurate. However, at the edges of each beam, there may be a gradual transition instead of a sharp one. This can be seen at the boundary of the frequency band supported by each beam, but it only occurs along the frequency band and not in the angular domain. The beams remain sharp in the intended directions and do not spread to other unwanted angular directions.
    % Nonetheless, the effect of this approximation can be seen at the boundary of each beam. For instance, instead of a desired sharp transition from one beam to another, there is a gradual transition from one beam to another. We can observe it at the boundary of frequency band supported by each beam. However, note that the gradual transition occurs only along the frequency band and not along the space or angular domain. Each beam remains sharp along the desired beam directions and doesnot spread at other undesired angular directions in space.
    \item We compared two F-S images created by the FSDA algorithm and a mathematical expression and found that the patterns produced by the FSDA algorithm are more precise and sharp. For example, the FSDA can create small beam bandwidths, while the mathematical patterns have wider bandwidths. The difference in the final F-S images produced by each method can be attributed to small variations in the set of delays and phases obtained by the FSDA and the mathematical approach, which we observed in~\ref{fig:phase_delay_beamcase3}. The difference in the FSDA and mathematical approach is caused by the heuristic method of obtaining the integer constant $k_d$ used to scale the phase profile of beam $d$, which reduces the error in line fitting. The mathematical approach is more conservative, as it prefers solutions with shorter delays over longer delays obtained by the FSDA for the same desired F-S response. However, the difference in the patterns obtained from both approaches is small, and both methods demonstrate the ability to create frequency-dependent multi-beam patterns.
\end{enumerate}

\begin{figure} [!t]
\centering
\subfigure[Delay Quantization]{
    \includegraphics[width=0.22\textwidth]{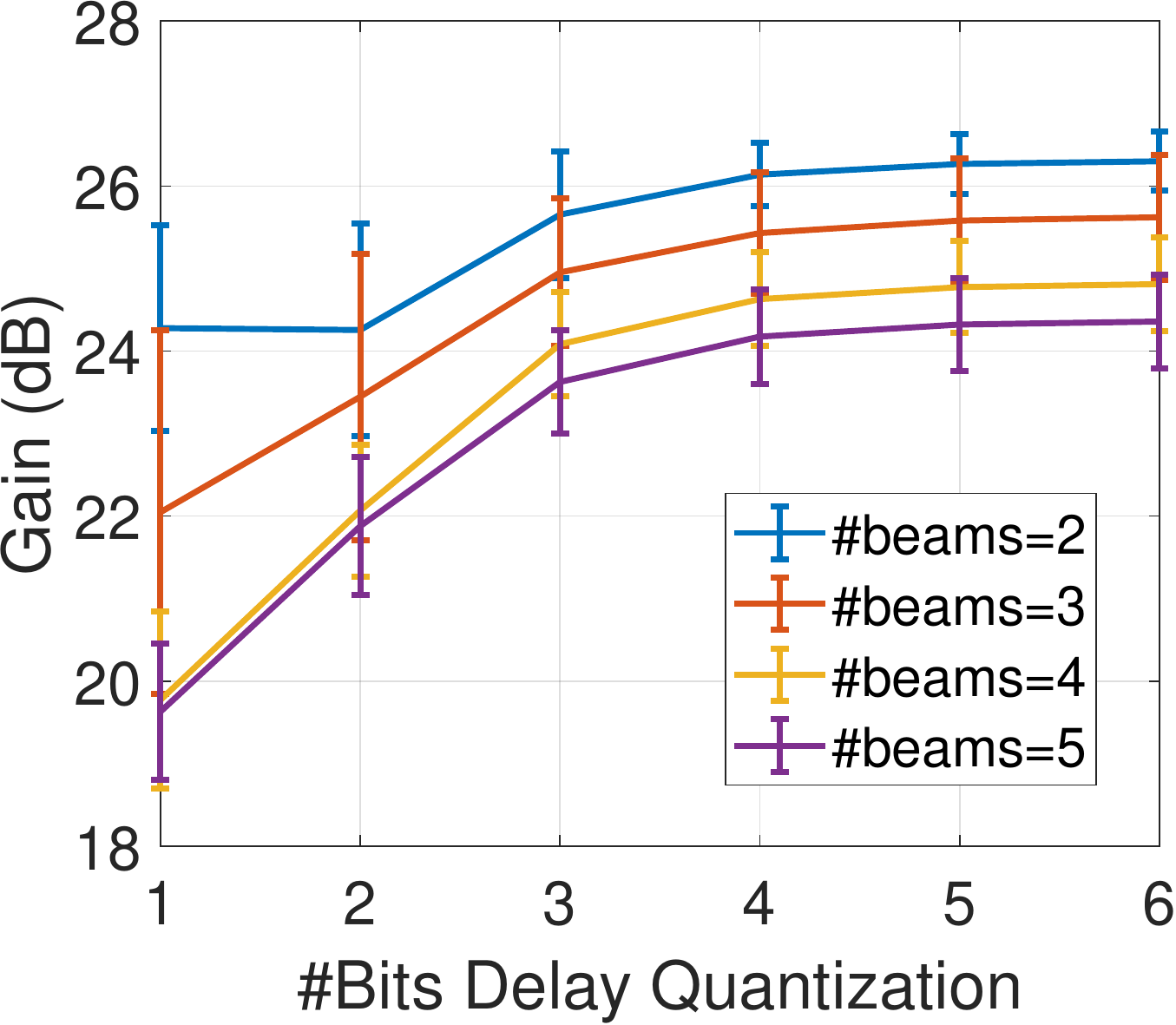}
    \label{fig:delay_quantization}
  }\hfill
  \subfigure[Phase Quantization]{
    \includegraphics[width=.22\textwidth]{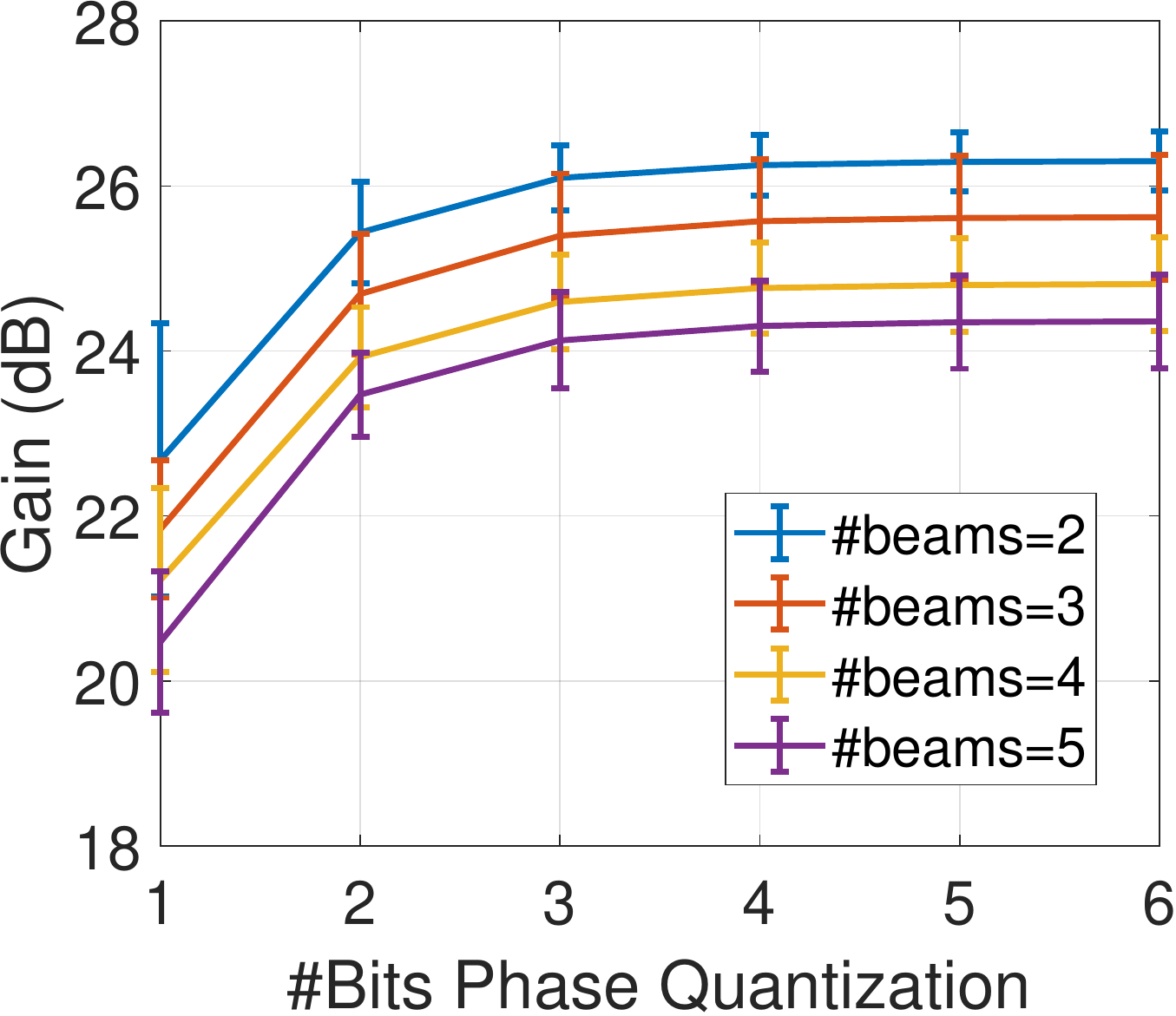}
    \label{fig:phase_quantization}
  }\hfill
    \caption{Effect of delay and phase quantization in overall array gain of 16 antennae \dpa. 
    }
    \label{fig:quantization}
\end{figure}

% \begin{figure} [!t]
%      \begin{subfigure}[b]{0.23\textwidth}
%          \centering
%          \includegraphics[width=\textwidth]{figures/delay_quantization.pdf}
%          \caption{Delay Quantization}
%          \label{fig:delay_quantization}
%      \end{subfigure}
%      \hfill
%     \begin{subfigure}[b]{0.23\textwidth}
%          \centering
%          \includegraphics[width=\textwidth]{figures/phase_quantization.pdf}
%          \caption{Phase Quantization}
%          \label{fig:phase_quantization}
%      \end{subfigure}
%     \caption{Effect of delay and phase quantization in overall array gain of 16 antenna \name. }
%     \label{fig:quantization}
%     \vspace{-0.2cm}
%  \end{figure}
 
\noindent
\textbf{$\blacksquare$ Effect of phase and delay quantization:}
We are now investigating the effect of quantization on the performance of the \name system. The DPA antenna array hardware must be designed with a certain level of quantization in the values of delays and phases. The more quantization bits used, the better the array's performance will be, but this also increases the complexity of the hardware design. This creates a trade-off between the number of quantization bits and the gain performance of the array. We have shown in Figure~\ref{fig:quantization} that the array gain improves as the number of quantization bits is increased from 1 to 6 for both delay and phase quantization. However, the improvement in the array gain becomes small for higher numbers of bits and reaches a saturation point. Specifically, the array gain improves by less than 1 dB when the quantization is increased from 3 bits. Therefore, practical \name hardware can be designed with a minimum of 3-bit quantization in phase/delay while maintaining high performance.

% We now study the impact of delay and phase quantization on the performance of \name. \name antenna array hardware need to be designed with certain quantization in the values of delays and phases. Higher the number of quantization bits gives better array performance. However, more bits comes with a higher complexity of hardware design. Therefore, there is a tradeoff between number of quantization bits and the gain performance of the array. We show in Figure~\ref{fig:quantization} that the array gain improves as the quantization bits are increased from 1 to 6 for both delay and phase quantization. The improvement in array gain, however, becomes small for higher number of bits and achieves saturation. Specifically, the array gain improves by less than 1 dB when the quantization is increased from 3 bit. Therefore, practical \name hardware can be designed with a minimum of 3 bit quantization in phase/delay with high performance.

\subsection{Conclusion for Appendix}

In this appendix, we presented a detailed mathematical analysis of the delay-phased array (DPA) and derived a closed-form mathematical expression for delay and phase values. We proposed an innovative approach for representing the per-antenna phase profile as a function of frequency and demonstrated that an ideal phase profile is in the form of a step function with the number of steps equal to the number of beams and the step size and width dependent on the beam directions and beam-bandwidth, respectively. We emphasized the challenge of realizing this non-linear phase profile using analog hardware, as it does not provide precise control over each frequency component in the signal. To overcome this limitation, we developed an approximate near-optimal method that utilizes a line-fitting method to derive the best-fit line for the given step function. The slope of the best-fit line gives the delay estimate, while the y-intercept gives the phase estimate. This method enables the derivation of delay and phase values in closed form and provides a deeper understanding of their behavior.

In particular, we showed that the delay values are always bounded within a small range that is independent of the number of antennas, similar to how the phase is bounded by $2\pi$. This allows for the creation of practical hardware with shorter delay lines, in contrast to traditional Rainbow-link architectures~\cite{li2022rainbow}, which require long delay lines. We also demonstrated the application of DPA in flexible directional frequency multiplexing, where frequency-dependent multi-beam patterns can be used to extend mmWave connectivity to IoT and URLLC applications by concurrently serving many devices with a small chunk of bandwidth while maintaining a directional link with all devices. Additionally, it can aid in concurrent control and communication by creating multiple pencil beams, with one dedicated beam carrying beam scan control messages and the other beams carrying communication data toward the direction of an active user.

\newpage
\bibliographystyle{IEEEtran}
\bibliography{references}

\begin{thebibliography}{10}
\providecommand{\url}[1]{#1}
\csname url@samestyle\endcsname
\providecommand{\newblock}{\relax}
\providecommand{\bibinfo}[2]{#2}
\providecommand{\BIBentrySTDinterwordspacing}{\spaceskip=0pt\relax}
\providecommand{\BIBentryALTinterwordstretchfactor}{4}
\providecommand{\BIBentryALTinterwordspacing}{\spaceskip=\fontdimen2\font plus
\BIBentryALTinterwordstretchfactor\fontdimen3\font minus
  \fontdimen4\font\relax}
\providecommand{\BIBforeignlanguage}[2]{{%
\expandafter\ifx\csname l@#1\endcsname\relax
\typeout{** WARNING: IEEEtran.bst: No hyphenation pattern has been}%
\typeout{** loaded for the language `#1'. Using the pattern for}%
\typeout{** the default language instead.}%
\else
\language=\csname l@#1\endcsname
\fi
#2}}
\providecommand{\BIBdecl}{\relax}
\BIBdecl

\bibitem{li2022rainbow}
R.~Li, H.~Yan, and D.~Cabric, ``Rainbow-link: Beam-alignment-free and
  grant-free mmw multiple access using true-time-delay array,'' \emph{IEEE
  Journal on Selected Areas in Communications}, 2022.

\bibitem{yan2019wideband}
H.~Yan, V.~Boljanovic, and D.~Cabric, ``Wideband millimeter-wave beam training
  with true-time-delay array architecture,'' in \emph{2019 53rd Asilomar
  Conference on Signals, Systems, and Computers}.\hskip 1em plus 0.5em minus
  0.4em\relax IEEE, 2019, pp. 1447--1452.

\bibitem{boljanovic2021fast}
V.~Boljanovic, H.~Yan, C.-C. Lin, S.~Mohapatra, D.~Heo, S.~Gupta, and
  D.~Cabric, ``Fast beam training with true-time-delay arrays in wideband
  millimeter-wave systems,'' \emph{IEEE Transactions on Circuits and Systems I:
  Regular Papers}, vol.~68, no.~4, pp. 1727--1739, 2021.

\bibitem{ghasempour2020single}
Y.~Ghasempour, R.~Shrestha, A.~Charous, E.~Knightly, and D.~M. Mittleman,
  ``Single-shot link discovery for terahertz wireless networks,'' \emph{Nature
  communications}, vol.~11, no.~1, pp. 1--6, 2020.

\bibitem{ghaderi2019integrated}
E.~Ghaderi, A.~S. Ramani, A.~A. Rahimi, D.~Heo, S.~Shekhar, and S.~Gupta, ``An
  integrated discrete-time delay-compensating technique for large-array
  beamformers,'' \emph{IEEE Transactions on Circuits and Systems I: Regular
  Papers}, vol.~66, no.~9, pp. 3296--3306, 2019.

\bibitem{boljanovic2020true}
V.~Boljanovic, H.~Yan, C.-C. Lin, S.~Mohapatra, D.~Heo, S.~Gupta, and
  D.~Cabric, ``True-time-delay arrays for fast beam training in wideband
  millimeter-wave systems,'' \emph{arXiv preprint arXiv:2007.08713}, 2020.

\bibitem{qualcomm2017augmented}
A.~R. Qualcomm, ``Augmented and virtual reality: the first wave of 5g killer
  apps,''
  \url{https://www.qualcomm.com/content/dam/qcomm-martech/dm-assets/documents/cr-qcom-140.pdf},
  2017.

\bibitem{pocovi2018achieving}
G.~Pocovi, H.~Shariatmadari, G.~Berardinelli, K.~Pedersen, J.~Steiner, and
  Z.~Li, ``Achieving ultra-reliable low-latency communications: Challenges and
  envisioned system enhancements,'' \emph{IEEE Network}, vol.~32, no.~2, pp.
  8--15, 2018.

\bibitem{3gpp2020}
\BIBentryALTinterwordspacing
``{Speech and multimedia Transmission Quality (STQ); QoS parameters and test
  scenarios for assessing network capabilities in 5G performance
  measurements},'' {3rd Generation Partnership Project (3GPP)}, TR {ETSI TR 103
  702 V1.1.1 }, 2020-11. [Online]. Available:
  \url{http://www.etsi.org/deliver/etsi_tr/138900_138999/138901/14.00.00_60/tr_138901v140000p.pdf}
\BIBentrySTDinterwordspacing

\bibitem{mangiante2017vr}
S.~Mangiante, G.~Klas, A.~Navon, Z.~GuanHua, J.~Ran, and M.~D. Silva, ``Vr is
  on the edge: How to deliver 360 videos in mobile networks,'' in
  \emph{Proceedings of the Workshop on Virtual Reality and Augmented Reality
  Network}, 2017, pp. 30--35.

\bibitem{benesty2019array}
J.~Benesty, I.~Cohen, and J.~Chen, \emph{Array Processing}.\hskip 1em plus
  0.5em minus 0.4em\relax Springer, 2019.

\bibitem{nagulu2021full}
A.~Nagulu, A.~Gaonkar, S.~Ahasan, S.~Garikapati, T.~Chen, G.~Zussman, and
  H.~Krishnaswamy, ``A full-duplex receiver with true-time-delay cancelers
  based on switched-capacitor-networks operating beyond the delay--bandwidth
  limit,'' \emph{IEEE Journal of Solid-State Circuits}, vol.~56, no.~5, pp.
  1398--1411, 2021.

\bibitem{ghaderi2020four}
E.~Ghaderi and S.~Gupta, ``A four-element 500-mhz 40-mw 6-bit adc-enabled
  time-domain spatial signal processor,'' \emph{IEEE Journal of Solid-State
  Circuits}, vol.~56, no.~6, pp. 1784--1794, 2020.

\bibitem{jain2020mmobile}
I.~K. Jain, R.~Subbaraman, T.~H. Sadarahalli, X.~Shao, H.-W. Lin, and
  D.~Bharadia, ``mmobile: Building a mmwave testbed to evaluate and address
  mobility effects,'' in \emph{Proceedings of the 4th ACM Workshop on
  Millimeter-Wave Networks and Sensing Systems}, 2020, pp. 1--6.

\bibitem{caudill2021real}
D.~Caudill, J.~Chuang, S.~Y. Jun, C.~Gentile, and N.~Golmie, ``Real-time mmwave
  channel sounding through switched beamforming with 3-d dual-polarized
  phased-array antennas,'' \emph{IEEE Transactions on Microwave Theory and
  Techniques}, vol.~69, no.~11, pp. 5021--5032, 2021.

\bibitem{palacios2018adaptive}
J.~Palacios, ``{Adaptive Codebook Optimization for Beam Training on
  Off-the-Shelf IEEE 802.11ad Devices},'' in \emph{Proceedings of the 24rd
  Annual International Conference on Mobile Computing and Networking}.\hskip
  1em plus 0.5em minus 0.4em\relax ACM, 2018.

\bibitem{giordani2018tutorial}
M.~Giordani, M.~Polese, A.~Roy, D.~Castor, and M.~Zorzi, ``{A tutorial on beam
  management for 3GPP NR at mmWave frequencies},'' \emph{IEEE Communications
  Surveys \& Tutorials}, vol.~21, no.~1, pp. 173--196, 2018.

\bibitem{github-nokia-wireless}
alvarovalcarce, ``Nokia wireless suite - github reference,''
  \url{https://github.com/nokia/wireless-suite} (2020/11/12).

\bibitem{jain2019impact}
I.~K. Jain, R.~Kumar, and S.~S. Panwar, ``The impact of mobile blockers on
  millimeter wave cellular systems,'' \emph{IEEE Journal on Selected Areas in
  Communications}, vol.~37, no.~4, pp. 854--868, 2019.

\bibitem{3gpp_138_901}
\BIBentryALTinterwordspacing
``{5G; Study on channel model for frequencies from 0.5 to 100 GHz},'' {3rd
  Generation Partnership Project (3GPP)}, TR {138 901 V14.0.0}, 2017-05.
  [Online]. Available:
  \url{http://www.etsi.org/deliver/etsi_tr/138900_138999/138901/14.00.00_60/tr_138901v140000p.pdf}
\BIBentrySTDinterwordspacing

\bibitem{jain2021two}
I.~K. Jain, R.~Subbaraman, and D.~Bharadia, ``Two beams are better than one:
  towards reliable and high throughput mmwave links,'' in \emph{Proceedings of
  the 2021 ACM SIGCOMM 2021 Conference}, 2021, pp. 488--502.

\bibitem{aykin2019multi}
I.~Aykin, B.~Akgun, and M.~Krunz, ``Multi-beam transmissions for blockage
  resilience and reliability in millimeter-wave systems,'' \emph{IEEE Journal
  on Selected Areas in Communications}, vol.~37, no.~12, pp. 2772--2785, 2019.

\bibitem{zhang2018multibeam}
J.~A. Zhang, X.~Huang, Y.~J. Guo, J.~Yuan, and R.~W. Heath, ``Multibeam for
  joint communication and radar sensing using steerable analog antenna
  arrays,'' \emph{IEEE Transactions on Vehicular Technology}, vol.~68, no.~1,
  pp. 671--685, 2018.

\bibitem{hassanieh2018fast}
H.~Hassanieh, O.~Abari, M.~Rodriguez, M.~Abdelghany, D.~Katabi, and P.~Indyk,
  ``Fast millimeter wave beam alignment,'' in \emph{Proceedings of the 2018
  Conference of the ACM Special Interest Group on Data Communication}.\hskip
  1em plus 0.5em minus 0.4em\relax ACM, 2018, pp. 432--445.

\bibitem{zhu2018high}
D.~Zhu, J.~Choi, Q.~Cheng, W.~Xiao, and R.~W. Heath, ``High-resolution angle
  tracking for mobile wideband millimeter-wave systems with antenna array
  calibration,'' \emph{IEEE Transactions on Wireless Communications}, vol.~17,
  no.~11, pp. 7173--7189, 2018.

\bibitem{rotman2016true}
R.~Rotman, M.~Tur, and L.~Yaron, ``True time delay in phased arrays,''
  \emph{Proceedings of the IEEE}, vol. 104, no.~3, pp. 504--518, 2016.

\bibitem{garakoui2015compact}
S.~K. Garakoui, E.~A. Klumperink, B.~Nauta, and F.~E. van Vliet, ``Compact
  cascadable gm-c all-pass true time delay cell with reduced delay variation
  over frequency,'' \emph{IEEE journal of solid-state circuits}, vol.~50,
  no.~3, pp. 693--703, 2015.

\bibitem{wadaskar20213d}
A.~Wadaskar, V.~Boljanovic, H.~Yan, and D.~Cabric, ``{3D Rainbow Beam Design
  for Fast Beam Training with True-Time-Delay Arrays in Wideband
  Millimeter-Wave Systems},'' in \emph{2021 55th Asilomar Conference on
  Signals, Systems, and Computers}.\hskip 1em plus 0.5em minus 0.4em\relax
  IEEE, 2021, pp. 85--92.

\bibitem{boljanovic2021compressive}
V.~Boljanovic and D.~Cabric, ``Compressive estimation of wideband mmw channel
  using analog true-time-delay array,'' in \emph{2021 IEEE Workshop on Signal
  Processing Systems (SiPS)}.\hskip 1em plus 0.5em minus 0.4em\relax IEEE,
  2021, pp. 170--175.

\bibitem{tan2021wideband}
J.~Tan and L.~Dai, ``Wideband beam tracking in thz massive mimo systems,''
  \emph{IEEE Journal on Selected Areas in Communications}, vol.~39, no.~6, pp.
  1693--1710, 2021.

\bibitem{tan2019delay}
------, ``Delay-phase precoding for thz massive mimo with beam split,'' in
  \emph{2019 IEEE Global Communications Conference (GLOBECOM)}.\hskip 1em plus
  0.5em minus 0.4em\relax IEEE, 2019, pp. 1--6.

\bibitem{ratnam2022joint}
V.~V. Ratnam, J.~Mo, A.~Alammouri, B.~L. Ng, J.~Zhang, and A.~F. Molisch,
  ``Joint phase-time arrays: A paradigm for frequency-dependent analog
  beamforming in 6g,'' \emph{IEEE Access}, vol.~10, pp. 73\,364--73\,377, 2022.

\bibitem{garg202028}
R.~Garg, G.~Sharma, A.~Binaie, S.~Jain, S.~Ahasan, A.~Dascurcu,
  H.~Krishnaswamy, and A.~S. Natarajan, ``A 28-ghz beam-space mimo rx with
  spatial filtering and frequency-division multiplexing-based single-wire if
  interface,'' \emph{IEEE Journal of Solid-State Circuits}, 2020.

\bibitem{cho2018true}
M.-K. Cho, I.~Song, and J.~D. Cressler, ``A true time delay-based sige
  bi-directional t/r chipset for large-scale wideband timed array antennas,''
  in \emph{2018 IEEE Radio Frequency Integrated Circuits Symposium
  (RFIC)}.\hskip 1em plus 0.5em minus 0.4em\relax IEEE, 2018, pp. 272--275.

\bibitem{hu20151}
F.~Hu and K.~Mouthaan, ``A 1--20 ghz 400 ps true-time delay with small delay
  error in 0.13 $\mu$m cmos for broadband phased array antennas,'' in
  \emph{2015 IEEE MTT-S International Microwave Symposium}.\hskip 1em plus
  0.5em minus 0.4em\relax IEEE, 2015, pp. 1--3.

\bibitem{shen2020mobility}
L.-H. Shen and K.-T. Feng, ``Mobility-aware subband and beam resource
  allocation schemes for millimeter wave wireless networks,'' \emph{IEEE
  Transactions on Vehicular Technology}, vol.~69, no.~10, pp. 11\,893--11\,908,
  2020.

\bibitem{zhang2019joint}
W.~Zhang, Y.~Wei, S.~Wu, W.~Meng, and W.~Xiang, ``Joint beam and resource
  allocation in 5g mmwave small cell systems,'' \emph{IEEE Transactions on
  Vehicular Technology}, vol.~68, no.~10, pp. 10\,272--10\,277, 2019.

\end{thebibliography}

% \begin{table}[htbp]
% \caption{Table Type Styles}
% \begin{center}
% \begin{tabular}{|c|c|c|c|}
% \hline
% \textbf{Table}&\multicolumn{3}{|c|}{\textbf{Table Column Head}} \\
% \cline{2-4} 
% \textbf{Head} & \textbf{\textit{Table column subhead}}& \textbf{\textit{Subhead}}& \textbf{\textit{Subhead}} \\
% \hline
% copy& More table copy$^{\mathrm{a}}$& &  \\
% \hline
% \multicolumn{4}{l}{$^{\mathrm{a}}$Sample of a Table footnote.}
% \end{tabular}
% \label{tab1}
% \end{center}
% \end{table}

% \begin{figure}[htbp]
% \centerline{\includegraphics{fig1.png}}
% \caption{Example of a figure caption.}
% \label{fig}
% \end{figure}

% \section*{Acknowledgment}

\end{document}